\documentclass[a4paper, 12pt]{article}
\usepackage{indentfirst}
\usepackage{graphicx}
\usepackage{epstopdf}
\usepackage{epsfig}
\usepackage{overpic}
\usepackage{overcite}
\usepackage{array}
\usepackage{color}
\usepackage{multirow}
\usepackage{float}
\usepackage{subfigure}
\usepackage{amsmath,amssymb, amsthm}
\usepackage[mathscr]{euscript}
\usepackage{latexsym,bm}
\usepackage{cite}
\usepackage{booktabs}

\usepackage{bbm}

\newcommand\comment[1]{}

\graphicspath{{images/}}
{\theoremstyle{remark} \newtheorem{remark}{\textup{\textbf{Remark}}}}
\newtheorem{proposition}{Proposition}

\def\mi{\mathbbm{i}}
\def\me{\mathbbm{e}}

\begin{document}
\bibliographystyle{unsrt}

\title{An advective-spectral-mixed method for time-dependent many-body Wigner simulations}
\author{Yunfeng Xiong\footnotemark[2],
\and Zhenzhu Chen\footnotemark[2],
\and Sihong Shao\footnotemark[2] $^,$\footnotemark[1]}
\renewcommand{\thefootnote}{\fnsymbol{footnote}}
\footnotetext[2]{LMAM and School of Mathematical Sciences, Peking University, Beijing 100871, China.}
\footnotetext[1]{To
whom correspondence should be addressed. Email:
\texttt{sihong@math.pku.edu.cn}}
\date{\today}
\maketitle

\begin{abstract}
As a phase space language for quantum mechanics,
the Wigner function approach bears a close analogy to classical mechanics and has been drawing growing attention, especially in simulating quantum many-body systems.
However, deterministic numerical solutions have been almost exclusively confined to one-dimensional one-body systems
and few results are reported even for one-dimensional two-body problems. This paper serves as the first attempt to solve the time-dependent many-body Wigner equation through a grid-based advective-spectral-mixed method. The main feature of the method is to resolve the linear advection in $(\bm{x},t)$-space by an explicit three-step characteristic scheme coupled with the piecewise cubic spline interpolation, while the Chebyshev spectral element method in $\bm k$-space is adopted for accurate calculation of the nonlocal pseudo-differential term.
Not only the time step of the resulting method is not restricted by the usual CFL condition and thus a large time step is allowed, but also the mass conservation can be maintained. In particular,
for the system consisting of identical particles,
the advective-spectral-mixed method can also rigorously preserve physical symmetry relations. The performance is validated through several typical numerical experiments, like the Gaussian barrier scattering, electron-electron interaction and a Helium-like system,
where the third-order accuracy against both grid spacing and time stepping is observed.

\vspace*{4mm}

\end{abstract}


\section{Introduction}


Ever since its invention in 1932, the Wigner
function (or (quasi) distribution) has provided a convenient way to render quantum mechanics in phase space\cite{Wigner1932}. It allows one to express macroscopically measurable quantities, such as currents and heat fluxes, in statistical forms as usually does in classical statistical mechanics\cite{tatarskiui1983,JacoboniBordone2004,DiasPrata2004}, thereby facilitating its applications in nanoelectronics\cite{bk:MarkowichRinghoferSchmeiser1990,th:Biegel1997}, non-equilibrium statistical mechanics\cite{bk:Balescu1975} and quantum optics\cite{bk:Schleich2011}.
Actually, a whole branch of experimental physics exists, known as quantum tomography, which purpose is reconstructing the Wigner function from measurements\cite{bk:Leonhardt1997,LeibfriedPfauMonroe1998}.
The most appealing feature of the Wigner equation is that, distinct from the Schr\"{o}dinger wavefunction approach, it shares many analogies to the classical mechanism and simply reduces to the classical counterpart when the reduced Planck constant vanishes\cite{Zurek1991}. Besides, the intriguing mathematical structure of the Wigner equation has also been employed in some advanced topics, such as the deformation quantization\cite{Zachos2002}.

Despite its great advantages, solving the Wigner equation has
presented one of the most mathematical challenging problems, since
the partial integro-differential equation is defined over $2 \times
d \times N$-dimensional phase space, where $d$ is the dimension of
space and $N$ is the number of involved particles, making it even
more complicated than the many-body Schr\"{o}dinger equation. For
the one-dimensional one-body situation, the first try conducted
by Frensely in simulating the resonant tunneling diode uses the
first-order upwind finite difference method (FDM)
\cite{Frensley1987,Frensley1990} and after that several second-order
FDMs were introduced \cite{JensenBuot1991,th:Biegel1997}. Later, a
plane wave approximation of the Wigner function\cite{Ringhofer1990}
and an operator splitting
scheme\cite{SuhFeixBertrand1991,ArnoldRinghofer1996} were proposed.
Recently, several high-order methods have been well designed to
capture accurately strong quantum effects, such as a cell average
spectral element method (SEM)\cite{ShaoLuCai2011}, moment methods\cite{LiLuWangYao2014,FurtmaierSucciMendoza2015}, a
WENO-solver\cite{DordaSchurrer2015}, etc. Among all those solvers,
the cell average SEM has proven to be very reliable as it presents a
simple but natural (precise) way to discretize the pseudo-differential
term and avoids tremendously the artificial dissipation for the advection
process\cite{ShaoLuCai2011}. It has to be noted that, to our
knowledge, all aforementioned deterministic methods have not yet
been extended to many-body Wigner
simulations\cite{NedjalkovSchwahaSelberherr2013}, even for the
one-dimensional two-body case.

Very recently, a Monte Carlo method (MCM) based on signed particles
for many-body Winger simulations has attracted a lot of attention
due to its simplicity as well as the satisfactory scaling on parallel
machines\cite{NedjalkovKosinaSelberherrRinghoferFerry2004,
NedjalkovSchwahaSelberherr2013,SellierNedjalkovDimov2014}. It has
enabled a direct simulation of many-body Wigner problems, such as
the strongly correlated indistinguishable
fermions\cite{SellierDimov2015}, and its accuracy for the
one-dimensional one-body problem has been validated by comparing with the cell average SEM\cite{ShaoSellier2015}.
Despite the promising progress, it has also been mentioned that particle-based
stochastic methods might not be very suitable for the problems where
phase space quantities vary over several orders of
magnitude\cite{CervenkaEllinghausNedjalkov2015}. Moreover, the
highly oscillating structure of the Wigner function due to the spatial coherence\cite{Zurek1991,LeibfriedPfauMonroe1998} makes it a
challenging task for both deterministic and stochastic methods to
capture precisely the quantum interference and correlation. To give a better description of the quantum phenomena in a wider dynamic range
and, at least, to provide a reliable reference solution for stochastic
methods, high-order accurate
deterministic methods for many-body Wigner simulations 
are highly needed.

This work serves as the first attempt for accurate
deterministic numerical solutions of the many-body  Wigner transport
equation, instead of resorting to the Wigner paths or many-body
Schr\"{o}dinger equations \cite{CancellieriBordoneJacoboni2007}. To
resolve the nonlocal pseudo-differential term, we adopt the
Chebyshev spectral element method\cite{ShaoLuCai2011} in $\bm
k$-space for it accurately resolves the oscillations of the Wigner
function, and avoids the artificial periodization at the same time. Another major obstacle lies in the discretization of the
advection term, because time steps employed by explicit Runge-Kutta
integrators are strictly limited by the Courant-Friedrichs-Lewy
(CFL) condition, thereby hampering the efficiency. The  first-order
upwind FDM, although alleviating this restriction, fails to provide
satisfactory results due to the numerical
dissipation\cite{th:Biegel1997,ShaoLuCai2011}. In order to overcome this obstacle,
a semi-Lagrangian-type characteristic method, which tracks the exact
Lagrangian advection on the spatial space grid in $\bm x$-space,
will be introduced in this work. The resulting advective-spectral-mixed method
relaxes the CFL restriction on the time step and ameliorates the
numerical dissipation significantly. Moreover, it maintains the mass
conservation and shows the third-order accuracy against both grid
spacing and time stepping when an explicit Adams three-step
method\cite{bk:HairerNorsettWanner1993} coupled with the piecewise
cubic spline interpolation\cite{bk:Boor2001} is implemented.

The proposed advective-spectral-mixed method allows us to study the quantum dynamics of two identical particles in phase space. We will illustrate how the physical symmetry relation is naturally embedded in the Wigner equation and preserved by the advective-spectral-mixed method. In fact, the effect of the Pauli exclusion principle and the uncertainty principle can be shown directly in phase space by simulating the electron-electron scattering and a Helium-like system.

The rest of the paper is organized as follows. In Section \ref{sec:theory}, we briefly review the many-body Wigner formalism with a discussion on the physical symmetry relation for a quantum system composed of identical particles. In Section \ref{sec:method}, the advective-spectral-mixed method is presented, while
related numerical analysis is given in Section \ref{sec:analysis}. Section \ref{sec:result} conducts several typical numerical experiments to verify the accuracy and convergence of the proposed method, and also shows the quantum dynamics of two electrons under different potentials in phase space. Concluding remarks and
further discussions are delineated in Section \ref{sec:conclusion}.

\section{The many-body Wigner formalism}
\label{sec:theory}

In this section, we briefly review the Wigner representation of quantum mechanics, and study
physical symmetry relations for a system composed of identical particles.
For numerical purpose, the truncated Wigner equation is introduced by exploiting the decay of the Wigner function for large wavenumbers,
and then a sufficient and necessary condition for such truncated Wigner equation
to maintain the mass conservation is derived.

\subsection{The Wigner equation}

The Wigner function $f(\bm{x}, \bm{k}, t)$ living in the phase space $(\bm{x},\bm{k})\in\mathbb{R}^{2dN}$ for the position $\bm{x}$ and the wavevector $\bm{k}$, introduced by Wigner in his pioneering work\cite{Wigner1932},  is defined by the Weyl-Wigner transform of the density matrix $\rho(\bm{r}, \bm{s}, t)$,
\begin{equation}
\begin{split}
&\rho\left(\bm{r},\bm{s},t\right)=\sum_{i}p_{i}\psi_{i}\left(\bm{r},t\right)\psi^{\dagger}_{i}\left(\bm{s},t\right),\\
&f\left(\bm{x}, \bm{k}, t\right)= \int_{\mathbb{R}^{Nd}} \textup{d} \bm{y}  \me^{-\mi \bm{k} \cdot \bm{y}} \rho\left(\bm{x}+\frac{\bm{y}}{2}, \bm{x}-\frac{\bm{y}}{2}, t\right),
\end{split}
\end{equation}
where $p_{i}$ gives the probability of occupying the $i$-th state,  $N$ is the number of involved particles, and $d$ denotes the dimension of space. Starting from the quantum Liouville equation, we are able to evaluate the derivative of the Wigner function and then arrive at the Wigner equation
\begin{equation}
\frac{\partial }{\partial t}f\left(\bm{x}, \bm{k}, t\right)+\frac{\hbar \bm{k}}{m} \cdot \nabla_{\bm{x}} f\left(\bm{x},\bm{k}, t\right)=\Theta_{V}\left[f\right]\left(\bm{x}, \bm{k}, t\right), \label{eq.Wigner}
\end{equation}
where
\begin{align}
\Theta_{V}\left[f\right]\left(\bm{x}, \bm{k}, t\right)
&=\int \textup{d} \bm{k^{\prime}} f\left(\bm{x},\bm{k}^{\prime},t\right)V_{w}\left(\bm{x},\bm{k}-\bm{k}^{\prime},t\right), \label{def_PDO}\\
V_{w}\left(\bm{x},\bm{k},t\right)&=\frac{1}{\mi\hbar \left(2\pi\right)^{N\cdot d}}\int \text{d}\bm{y} \me^{-i\bm{k}\cdot \bm{y}} D_{V}\left(\bm{x}, \bm{y}, t\right), \label{Wigner_kernel} \\
D_{V}\left(\bm{x}, \bm{y}, t\right)&=V\left(\bm{x}+\frac{\bm{y}}{2}, t\right)-V\left(\bm{x}-\frac{\bm{y}}{2}, t\right). \label{Dv}
\end{align}
Here the nonlocal pseudo-differential term $\Theta_{V}[f](\bm{x}, \bm{k}, t)$ contains the quantum information,
 $D_{V}(\bm{x}, \bm{y}, t)$ denotes a central difference of the potential function $V(\bm{x}, t)$,
the Wigner kernel $V_{w}(\bm{x}, \bm{k}, t)$ is defined   through the Fourier transform of $D_{V}(\bm{x}, \bm{y}, t)$,
$\hbar$ is the reduced Planck constant and $m$ is the particle mass (for simplicity, we assume all $N$ particles have the same mass throughout this work).

The Wigner function $f(\bm{x}, \bm{k}, t)$ can be used to calculate the particle density $n(\bm{x}, t)$ and the current density $\bm{j}(\bm{x}, t)$ by
\begin{align}
n\left(\bm{x}, t\right) & =\int f\left(\bm{x}, \bm{k}, t\right) \textup{d}\bm{k}, \\
\bm{j}\left(\bm{x}, t\right) &=\frac{\hbar}{m}\int \bm{k}f\left(\bm{x}, \bm{k}, t\right) \textup{d}\bm{k}.
\end{align}
Since it is easy to verify that
\begin{equation}\label{mass_conservation}
\int \text{d} \bm{k} \int \text{d} \bm{k^{\prime}} f\left(\bm{x},\bm{k}^{\prime},t\right)V_{w}\left(\bm{x},\bm{k}-\bm{k}^{\prime},t\right)=0,
\end{equation}
we can derive the continuity equation
\begin{equation}
\frac{\partial }{\partial t} n\left(\bm{x}, t\right) +\nabla_{\bm{x}}\cdot  \bm{j}\left(\bm{x}, t\right)=0,
\end{equation}
which corresponds to the conservation of the first moment (i.e., total particle number or mass)
\begin{equation}\label{eq:mass}
\frac{\textup{d}}{\textup{dt}}\iint f\left(\bm{x}, \bm{k},t\right) \textup{d}\bm{x} \textup{d}\bm{k}=0.
\end{equation}

Furthermore, if the potential $V(\bm{x}, t)$ allows a Taylor expansion
in $\bm{x}$-space, then $D_{V}(\bm{x}, \bm{y}, t)$ depends only on the odd derivatives, as shown in the following
\begin{equation}\label{Taylor_expansion_V}
D_{V}\left(\bm{x}, \bm{y}, t\right)=\mi \hbar \sum_{l=0}^{+\infty}\frac{\left(\mi\hbar/2\right)^{2l}}{\left(2l+1\right)!} \nabla^{2l+1}_{\bm{x}}V\left(\bm{x},t\right) \cdot\left(-\frac{\mi \bm{y}}{\hbar}\right)^{2l+1}.
\end{equation}
Substituting Eq.~\eqref{Taylor_expansion_V} into Eq.~\eqref{Wigner_kernel} and using the basic properties of Fourier transform, we can readily obtain the Moyal expansion of the Wigner equation\cite{bk:MarkowichRinghoferSchmeiser1990, bk:Schleich2011}
 \begin{equation}\label{Moyal_expansion}
 \begin{split}
\frac{\partial }{\partial t}f+\frac{ \bm{p}}{m} \cdot \nabla_{\bm{x}} f=&\nabla_{\bm{x}}V \cdot \nabla_{\bm{p}}f\\
&+\sum_{l=1}^{+\infty} \frac{\left(-1\right)^{l}}{\left(2l+1\right)!}\left(\frac{\hbar}{2}\right)^{2l} \nabla^{2l+1}_{\bm{x}}V \cdot \nabla^{2l+1}_{\bm{p}}f,
\end{split}
\end{equation}
where $\bm{p}=\hbar \bm{k}$ is the momentum.
It can be easily observed there that,
when $\hbar \to 0$, the Wigner equation reduces immediately to the classical Vlasov equation\cite{SonnendruckerRocheBertrand1999},
the Liouville part of the Boltzmann equation;
the quantum evolution governed by the Wigner potential couples
all odd derivatives of the potential,
whereas the classical evolution is determined only by the first derivative. That is,
within the phase space formalism of quantum mechanics,
quantum dynamics can be naturally connected to classical dynamics\cite{Zurek1991} and thus a unified treatment of both is possible\cite{NedjalkovKosinaSelberherrRinghoferFerry2004}.

\subsection{Physical symmetry relation}
\label{sec:physical:sym}

As the simplest but most appealing many-body problem,
the system composed of identical particles
has been extensively studied,
where symmetry relations play a key role. Next
we will investigate those symmetry relations within the Wigner function formalism. Hereafter the formulation  will be mostly illustrated for the one-dimensional two-body situation for simplicity, and generalization to arbitrary-sized phase space is straightforward.

The one-dimensional two-body Wigner function for a pure state reads
\begin{equation}\label{1d2b}
\begin{split}
f\left(x_{1}, x_{2}, k_{1}, k_{2}, t\right)=&\iint \text{d}y_{1} \textup{d}y_{2} \me^{-\mi k_{1}y_{1}-\mi k_{2}y_{2}}\\
&\times \psi\left(x_{1}+\frac{y_{1}}{2}, x_{2}+\frac{y_{2}}{2}, t\right) \psi^{\dag}\left(x_{1}-\frac{y_{1}}{2}, x_{2}-\frac{y_{2}}{2}, t\right),
\end{split}
\end{equation}
where $\psi\left(x_{1}, x_{2}, t\right)$ is the wave function describing a quantum system composed of two identical particles, and the superscript $\dag$ denotes the complex conjugate. When the position coordinates of two identical particles are interchanged, the wavefunction  $\psi\left(x_{1}, x_{2}, t\right)$ either remains unaffected for bosons or changes sign for fermions.
In contrast, the Wigner function $f(x_1,x_2,k_1,k_2,t)$ satisfies the same symmetry relation for both cases\cite{CancellieriBordoneBertoni2004}
\begin{equation}\label{symmetry_relation}
f\left(x_{1}, x_{2}, k_{1}, k_{2}, t\right)=f\left(x_{2}, x_{1}, k_{2}, k_{1}, t\right),
\end{equation}
which can be also readily verified from Eq.~\eqref{1d2b}.
Actually, we will further show that the Wigner equation holds the symmetry relation \eqref{symmetry_relation} when time evolves provided that
\begin{equation}\label{sym_potential}
V\left(x_{1}, x_{2}, t\right)=V\left(x_{2}, x_{1}, t\right).
\end{equation}
Before that, for convenience, the Wigner equation \eqref{eq.Wigner} is reformulated into
the following evolution system of the initial value problem,
\begin{equation}\label{Abstract_form}
\left\{
\begin{split}
&\partial_{t} f- Af- B\left(t\right)f=0, \,\,\, t\in \left[0,T\right],\\
&f\left(t=0\right)=f_{0}\in L^{2}\left(\mathbb{R}^{4}\right),
\end{split}
\right.
\end{equation}
where the operators $A$ and $B(t)$ are defined as follows
\begin{align}
A: f\in D(A)   &\to
Af= -\frac{\hbar}{m} \bm{k} \cdot \nabla_{\bm{x}}f \in
L^{2}\left(\mathbb{R}^{4}\right), \\
B(t): f\in L^{2}\left(\mathbb{R}^{4}\right) &\to
B\left(t\right)f=\Theta_{V}f \in
L^{2}\left(\mathbb{R}^{4}\right), \label{B(t)}
\end{align}
with $D\left(A\right)=\left\{f\in L^{2}\left(\mathbb{R}^{4}\right): \bm{k} \cdot \nabla_{\bm{x}}f \in L^{2}\left(\mathbb{R}^{4}\right)\right\}$.
When $V(\bm{x},t)$ is bounded, we have the following estimate\cite{ArnoldRinghofer1996}
\begin{equation}
\left\Vert B\left(t\right) \right\Vert_{L^{2}\left(\mathbb{R}^{4}\right)} \leq 2 \left\Vert V\left(\bm{x},t\right) \right\Vert_{L^{\infty}\left(\mathbb{R}^{2}\right)},
\end{equation}
which ensures the boundness of $B(t)$.

\begin{proposition}\label{pro:sym}
 Let $\sigma: L^{2}(\mathbb{R}^{4}) \to L^{2}(\mathbb{R}^{4})$ be an isomorphism, defined as
\begin{equation}\label{sigma}
\sigma f\left(x_{1}, x_{2}, k_{1}, k_{2}, t\right)=f\left(x_{2}, x_{1}, k_{2}, k_{1}, t\right).
\end{equation}
Then $\sigma f=f$ for $t\in \left[0,T\right]$ if the following conditions are satisfied:\\
\textup{(H1)} $f \in C^{1}\left(\left[0,T\right]: L^{2}\left(\mathbb{R}^{4}\right) \right)$;\\
\textup{(H2)} $B\left(t\right)$ is bounded for $t\in \left[0,T\right]$;\\
\textup{(H3)} $V\left(x_{1}, x_{2}, t\right)$ satisfies
Eq.~\eqref{sym_potential};\\
\textup{(H4)} $\sigma f_{0}=f_{0}$.
\end{proposition}

\begin{proof}
According to Theorem 2.3 in Chapter 5\cite{bk:Pazy1983},  $A+B(t)$ is a stable family of infinitesimal generators in $L^{2}\left(\mathbb{R}^{4}\right)$
for a hyperbolic system and $B(t)$ is bounded (see (H2)). Consequently,
Theorem 5.3 in Chapter 5\cite{bk:Pazy1983} further guarantees the existence and uniqueness of a classical solution $f$ for the Wigner system \eqref{Abstract_form}
provided that $\textup{(H1)}$ is satisfied.

Let $A_{1}=\sigma A \sigma^{-1}$ and $D\left(A_{1}\right)=\left\{f\in L^{2}\left(\mathbb{R}^{4}\right): \sigma^{-1}f \in D\left(A\right)\right\}$. By the chain rule,
it can be easily verified that
\begin{equation}\label{sigmaA}
\left(k_{1}\frac{\partial}{\partial x_{1}} + k_{2}\frac{\partial}{\partial x_{2}}\right)\sigma f=\sigma\left(k_{1}\frac{\partial}{\partial x_{1}} + k_{2}\frac{\partial}{\partial x_{2}}\right)f,
\,\,\,\forall f\in D\left(A_{1}\right).
\end{equation}
By the definition of $B\left(t\right)$ in Eq.~\eqref{B(t)} and the condition \textup{(H3)}, direct algebraic calculations yield
\begin{equation}\label{sigmaB}
\sigma B\left(t\right)= B\left(t\right) \sigma.
\end{equation}
Hence, combining Eqs.~\eqref{Abstract_form}, \eqref{sigma}, \eqref{sigmaA}
and \eqref{sigmaB} leads to
\begin{equation}
\partial_{t} \sigma f=\sigma \partial_{t} f=\sigma \left(A+B\left(t\right)\right) f= \left(A+B\left(t\right)\right) \left(\sigma f\right),
\end{equation}
which implies that $\sigma f$ is also a classical solution of the system \eqref{Abstract_form}. In consequence,
we have $\sigma f=f$ for $t\in \left[0,T\right]$ due to the uniqueness and the condition $\textup{(H4)}$.
\end{proof}

Although Proposition \ref{pro:sym} seems not very difficult,
the physical implication is quite important,
because it tells us that the Pauli exclusion principle for fermions is naturally embedded in the Wigner equation,
provided that the initial data corresponds to the antisymmetric wave functions. More importantly,
we will show later that such symmetry relation can be still inherited by the numerical solutions calculated from the proposed advective-spectral-mixed method (see Section \ref{sec:analysis:sym}).

\subsection{The truncated Wigner equation}

As shown in the Wigner equation \eqref{eq.Wigner},
the nonlocal pseudo-differential term poses the first challenge in seeking approximations for the Wigner function.
Considering the decay of the Wigner distribution when $|\bm{k}|\to+\infty$ due to the Riemann-Lebesgue lemma,
a simple nullification of the distribution outside a sufficiently large $\bm{k}$-domain is employed in this paper. It should be noted that truncating the infinite series in the Moyal expansion \eqref{Moyal_expansion} provides another way for numerical purpose\cite{HugMenkeSchleich1998I}, but we will not use it in this work. Suppose the Wigner function $f(\bm{x}, \bm{k}, t)$ is sought in
a sufficiently large $\bm{k}$-domain,
denoted by $\mathcal{K}_{1} \times \mathcal{K}_{2}$
with the size $\left|\mathcal{K}_{i}\right|=k_{i, \textup{max}}- k_{i, \textup{min}}\, (i=1,2)$.
Then the truncated Wigner equation reads
\begin{align}
\frac{\partial }{\partial t}f\left(\bm{x}, \bm{k}, t\right)+&\frac{\hbar \bm{k}}{m} \cdot \nabla_{\bm{x}} f\left(\bm{x},\bm{k}, t\right)=\Theta_{V}^T\left[f\right]\left(\bm{x}, \bm{k}, t\right), \label{eq.Wigner_truncated}\\
\Theta_{V}^T\left[f\right]\left(\bm{x}, \bm{k}, t\right)
&=\iint_{\mathcal{K}_{1}\times\mathcal{K}_{2}} \textup{d} \bm{k^{\prime}} f\left(\bm{x},\bm{k}^{\prime},t\right)V_{w}^T\left(\bm{x},\bm{k}-\bm{k}^{\prime}\right),\label{PDO}
\end{align}
where (and hereafter) we have only considered the time-independent potential. 
Since the $\bm{k}$-integration in  Eq.~\eqref{PDO} ranges in a finite region $\mathcal{K}_{1} \times \mathcal{K}_{2}$, we only need the information of the Wigner kernel on a finite bandwidth,
i.e., the truncated Wigner kernel of the following form
\begin{equation}\label{Wigner_kernel_finite}
V^{T}_{w}\left(x_{1},x_{2},k_{1},k_{2}\right) =V_{w}\left(x_{1},x_{2},k_{1},k_{2}\right) \textup{rect}\left(\frac{k_{1}}{2|\mathcal{K}_{1}|}\right)  \textup{rect}\left(\frac{k_{2}}{2|\mathcal{K}_{2}|}\right),
\end{equation}
where $\textup{rect}(k)$ is the rectangular function
\begin{equation}
\textup{rect}\left(k\right)=\left\{\begin{split}& 1, \quad \left|k\right|<\frac{1}{2},\\
&0, \quad \left|k\right| \geq \frac{1}{2}.\end{split}\right.
\end{equation}
It can be readily verified that
Proposition \ref{pro:sym} still holds for the truncated Wigner equation \eqref{eq.Wigner_truncated} provided $\mathcal{K}_1=\mathcal{K}_2$. That is, the physical symmetry relation \eqref{symmetry_relation} is also preserved by the truncated Wigner function.

On the other hand, starting from the Poisson summation formula for the Wigner kernel $V_{w}$ in the whole phase space (see Eq.~\eqref{Wigner_kernel})
\begin{equation}\label{Poisson_summation}
\begin{split}
 &\sum_{n_{1}=-\infty}^{+\infty}\sum_{n_{2}=-\infty}^{+\infty}V_{w}\left(x_{1},x_{2},k_{1}+n_{1}\frac{2\pi }{\Delta y_{1}},k_{2}+n_{2}\frac{2\pi }{\Delta y_{2}}\right)\\
 =&\frac{1}{4\mi \hbar \pi^{2}}\sum_{\mu=-\infty}^{+\infty}\sum_{\nu=-\infty}^{+\infty}\Delta y_{1} \Delta y_{2} D_{V}\left(x_{1}, x_{2}, y_{\mu}, y_{\nu}\right) \me^{-\mi k_{1} y_{\mu}-\mi k_{2} y_{\nu}},
\end{split}
\end{equation}
where $y_{\mu}=\mu \Delta y_{1}$ and $y_{\nu}=\nu \Delta y_{2}$ with $\Delta y_{i}\,(i=1,2)$ being the spacing in $\bm{y}$-space,
we can easily obtain
\begin{equation}\label{Poisson_summation_truncated}
V^{T}_{w}\left(x_{1},x_{2},k_{1},k_{2}\right)=\frac{\Delta y_{1}\Delta y_{2}}{4 \mi \hbar \pi^{2}} \sum_{\mu=-\infty}^{+\infty}\sum_{\nu=-\infty}^{+\infty}D_{V}\left(x_{1}, x_{2}, y_{\mu}, y_{\nu}\right)\me^{-\mi k_{1} y_{\mu}-\mi k_{2} y_{\nu}},
\end{equation}
provided that the central period $[-\pi/\Delta y_1,\pi/\Delta y_1]\times[-\pi/\Delta y_2,\pi/\Delta y_2]$ contains the computational domain $\mathcal{K}_{1} \times \mathcal{K}_{2}$. That is, Eq.~\eqref{Poisson_summation_truncated} holds only under
the Nyquist condition
\begin{equation}\label{Nyquist_criterion}
\left|\mathcal{K}_{i}\right|\Delta y_{i}\le 2\pi, \,\,\, i=1,2.
\end{equation}

Moreover, from Eq.~\eqref{mass_conservation}, in order to maintain the mass conservation,
for any $y_{\mu}$ and  $y_{\nu}$, it is required that
\begin{equation}
\iint_{\mathcal{K}_{1}\times \mathcal{K}_{2}}  \me^{-\mi k_{1} y_{\mu}-\mi k_{2} y_{\nu}} \textup{d} k_{1} \textup{d} k_{2}=0,
\end{equation}
which can be achieved by a sufficient condition\cite{Frensley1987,ShaoLuCai2011}
\begin{equation}\label{conservation_condition}
\left|\mathcal{K}_{i}\right|\Delta y_{i}=2\pi, \,\,\, i=1,2.
\end{equation}
In summary, combining Eqs.~\eqref{Nyquist_criterion} and \eqref{conservation_condition} implies that the above constraint on the length of $\bm{k}$-domain is not only sufficient but also necessary,
which may be pointed out for the first time in the literature.

\section{Numerical scheme}
\label{sec:method}

This section is devoted into elaborating our advective-spectral-mixed method for time-dependent many-body Wigner simulations in two aspects. The first lies that a semi-Lagrange-type characteristic method\cite{SonnendruckerRocheBertrand1999, CrouseillesMehrenbergerSonnendrucker2010} in $(\bm{x},t)$-space will adopted. This advective approximation of the Wigner equation fully exploits the integral formulation based on the semigroup theory and exactly follows the spatial characteristic lines backward in time. That is, it can be implemented in an explicit way with the help of the Adams multistep solvers as well as piecewise spline interpolations. More importantly, it allows large time steps for it is not restricted by the usual CFL condition.
The second aspect is the spectral element method\cite{bk:ShenTangWang2011},
a natural choice regarding to the Fourier transform nature of the Wigner potential\cite{ShaoLuCai2011},
will be employed to discrete the nonlocal pseudo-differential term. This spectral discretization is able to
give rise to a close representation of the pseudo-differential term and provides a highly accurate approximation because all integrals are analytically implemented in virtue of the global spectral expansion
in $\bm{k}$-space.

Now we want to solve the truncated Wigner equation \eqref{eq.Wigner_truncated} in a finite domain
$\mathcal{X}_{1}\times\mathcal{X}_{2}\times \mathcal{K}_{1}\times \mathcal{K}_{2}$.
A uniform grid mesh with the spacing $\Delta x_i\,(i=1,2)$ in $\bm{x}$-space
\begin{align}
\mathcal{X}_{1}\times\mathcal{X}_{2} &= \bigcup_{q_{1}, q_{2}}\mathcal{X}_{q_{1},q_{2}} , \quad \mathcal{X}_{q_{1},q_{2}}=\left[x_{1,q_{1}-1}, x_{1, q_{1}}\right] \times \left[x_{2,q_{2}-1}, x_{2, q_{2}}\right], \\
x_{i, 0} &=x_{i, \textup{min}}, \quad x_{i, q_i}=x_{i, \textup{min}}+(q_i-1)\Delta x_{i}, \quad i=1,2,
\label{xgrid}
\end{align}
is used, while the $\bm{k}$-domain is divided into $M_1M_2$
non-overlapping elements as follows
\begin{equation}
\mathcal{K}_{1}\times\mathcal{K}_{2}=\bigcup_{r_{1}=1}^{M_{1}} \bigcup_{r_{2}=1}^{M_{2}}\mathcal{K}_{r_{1}}\times \mathcal{K}_{r_{2}},
\end{equation}
with $\mathcal{K}_{r_{i}}=[d_{r_i}, d_{r_i+1}]$,
and then the Gauss-Chebyshev collocation points\cite{ShaoLuCai2011} will be chosen in each element.

\subsection{The advective approach in $(\bm{x},t)$-space}
\label{sec:method:adv}

The essential difference between the Wigner equation and
the classical Vlasov equation lies in the nonlocal nature of the Wigner kernel, making it entirely not trivial to follow the characteristic lines in $\bm{k}$-space.
To this end, the integral form of the truncated Wigner equation \eqref{eq.Wigner_truncated} using the semigroup theory\cite{bk:Pazy1983} is the start point now, instead of the operator splitting scheme\cite{SonnendruckerRocheBertrand1999}. For simplicity, let
\begin{equation}\label{eq:g}
g\left(\bm{x},\bm{k},t\right)=\Theta_{V}^T\left[f\right]\left(\bm{x},\bm{k},t\right).
\end{equation}
Applying the variation-of-constant formula\cite{bk:Pazy1983} into
the Wigner equation \eqref{Abstract_form} leads to
the following {\sl equivalent integral formulation}
\begin{equation}\label{mild_solution}
f\left(\bm{x},\bm{k},t\right)=\me^{\left(t-t_{0}\right)A}f\left(\bm{x},\bm{k},t_{0}\right)+\int_{t_{0}}^{t}\me^{\left(t-\tau\right)A}g\left(\bm{x},\bm{k},\tau\right)\textup{d}\tau.
\end{equation}

The operator $T(\Delta t)=\me^{\Delta t A}$ is a $C_{0}$-semigroup of isometries on $L^{2}(\mathbb{R}^{4})$, describing the Lagrangian advection in $(\bm{x},t)$-space
\begin{equation}\label{eq:Tdt}
T\left(\Delta t\right)f\left(\bm{x}, \bm{k}, \tau\right)=f\left(\bm{X}\left(\tau+\Delta t; \bm{x}, \tau\right), \bm{k}, \tau\right),
 \end{equation}
where $\Delta t$ is the time increment,
and $\bm{X}(t;\bm{x}_{0},t_{0})$ is the spatial characteristic curve at the end time $t$, starting from $\bm{x}_{0}$ and $t_{0}$, and satisfies the following dynamic system
\begin{equation}\label{characteristic_curve}
\frac{\textup{d} \bm{X}\left(t;\bm{x}_{0},t_{0}\right)}{\textup{d}t}=-\bm{v}, \quad \bm{X}\left(t_{0}; \bm{x}_{0}, t_{0}\right)=\bm{x}_{0}, \quad t\ge t_{0},
\end{equation}
with the velocity $\bm{v} = {\hbar \bm{k}}/{m}$. Since the solution of Eq.~\eqref{characteristic_curve} is explicitly given by
\begin{equation}
\bm{X}\left(t; \bm{x}_{0}, t_{0}\right)=\bm{x}_{0}-\bm{v}\left(t-t_{0}\right),
\end{equation}
Eq.~\eqref{eq:Tdt} turns out to be 
\begin{equation}\label{x_shift}
T\left(\Delta t\right)f\left(\bm{x}, \bm{k}, \tau\right)=f\left(\bm{x}\left(\Delta t\right),\bm{k},\tau\right)=f\left(\bm{x}-\bm{v}\Delta t, \bm{k}, \tau\right),
 \end{equation}
where $\bm{x}(\Delta t)$ denotes the displacement occurring in the time interval $\Delta t$.
Let $t^{n}=n\Delta t$. Then
combining Eqs.~\eqref{mild_solution} and \eqref{x_shift} yields
\begin{equation}\label{Integral_form_1}
\begin{split}
f\left(\bm{x}, \bm{k},t^{n+1}\right)=&f\left(\bm{x}\left(t^{n+1}-t^{n}\right), \bm{k},t^{n}\right)+\int_{t^{n}}^{t^{n+1}} g\left(\bm{x}\left(t^{n+1}-\tau\right),\bm{k},\tau\right)\textup{d}\tau,
\end{split}
\end{equation}
which constitutes the main object for numerical approximations.

The first approximation comes from replacing the integrand $g\left(\bm{x}\left(t^{n+1}-\tau\right),\bm{k},\tau\right)$
in  Eq.~\eqref{Integral_form_1} with a Lagrangian polynomial in the spirit of the Adams multistep solvers.
The general $p$-step formula for approximating Eq.~\eqref{Integral_form_1} usually reads
\begin{equation}\label{implicit_advective_approx}
f^{n+1}\left(\bm{x}, \bm{k}\right)= f^{n}\left(\bm{x}-\bm{v}\Delta t, \bm{k}\right)+ \Delta t \sum_{s=0}^{p}{\gamma}_{s} g^{n+1-s}\left(\bm{x}-s\bm{v}\Delta t, \bm{k}\right),
\end{equation}
where $f^{l}(\bm{x},\bm{k})\,(l=n,n+1)$  denotes the numerical approximations of $f(\bm{x},\bm{k},t^{l})$ (the same convention is used for $g$),
and the coefficients ${\gamma}_{s}$ can be determined through the root condition and certain algebraic relations (for details, one can refer to\cite{bk:HairerNorsettWanner1993}).
In this work, we will use three typical solvers as shown below.
\begin{itemize}
\item Explicit Euler method
\begin{equation}\label{explicit_Euler}
f^{n+1}\left(\bm{x},\bm{k}\right)=f^{n}\left(\bm{x}-\bm{v}\Delta t,\bm{k}\right)+\Delta t g^{n}\left(\bm{x}-\bm{v}\Delta t,\bm{k}\right).
\end{equation}
\item Implicit midpoint method
\begin{equation}\label{implicit_mid_point}
f^{n+1}\left(\bm{x},\bm{k}\right)=f^{n}\left(\bm{x}-\bm{v}\Delta t,\bm{k}\right)+\frac{1}{2}\Delta tg^{n+1}\left(\bm{x},\bm{k}\right)+\frac{1}{2}\Delta tg^{n}\left(\bm{x}-\bm{v}\Delta t,\bm{k}\right).
\end{equation}
\item Explicit three-step method
\begin{equation}\label{explicit_3_step}
\begin{split}
f^{n+1}\left(\bm{x},\bm{k}\right)=&f^{n}\left(\bm{x}-\bm{v}\Delta t,\bm{k}\right)+\frac{23}{12}\Delta t g^{n}\left(\bm{x}-\bm{v}\Delta t,\bm{k}\right)
\\&-\frac{16}{12}\Delta t g^{n-1}\left(\bm{x}-2\bm{v}\Delta t,\bm{k}\right)+\frac{5}{12}\Delta t g^{n-2}\left(\bm{x}-3\bm{v}\Delta t,\bm{k}\right).
\end{split}
\end{equation}
\end{itemize}
Obviously, the above three methods are of the order $\mathcal{O}(\Delta t)$, $\mathcal{O}(\Delta t^2)$,
and $\mathcal{O}(\Delta t^3)$, respectively.
In practice,
low order methods can be used to provide
the missing starting points for high order ones.
For example,
at the initial stage, the missing two points needed in the three-step method can be obtained using both Euler and midpoint methods with a relatively smaller time step in a prediction-correction manner.

The second approximation lies in
interpolating the function values,
$f^{n}(\bm{x}-\bm{v}\Delta t, \bm{k})$ and $g^{n+1-s}(\bm{x}-s\bm{v}\Delta t, \bm{k}), (s=0,1,\cdots,p)$,
required by Eq.~\eqref{implicit_advective_approx},
because the shifted points $(\bm{x}-s\bm{v}\Delta t, \bm{k})$ might not be located exactly at the grids in Eq.~\eqref{xgrid}. For this, the piecewise cubic spline interpolation is adopted here since it appears to be  a good comprise between accuracy and cost\cite{SonnendruckerRocheBertrand1999, CrouseillesMehrenbergerSonnendrucker2010}.
In general, the piecewise bicubic spline $S\left(x_{1}, x_{2}\right)$ is defined as
\begin{equation}\label{general_bicubic_sp}
S\left(x_{1}, x_{2}\right)=\sum_{\nu=0}^{3} \sum_{\kappa=0}^{3} \eta_{\nu \kappa} \left(x_{1}-x_{1, i}\right)^{\nu} \left(x_{2}-x_{2, j}\right)^{\kappa} ,
\end{equation}
for $ \left(x_{1}, x_{2}\right) \in \left[x_{1, i}, x_{1, i+1}\right] \times \left[x_{2, j}, x_{2, j+1}\right] $, which requires the evaluation of the coefficient table $\left(\eta_{\nu \kappa}\right)$. For the convenience of numerical analysis, we can equivalently evaluate $S\left(x_{1}, x_{2}\right)$ by several one-dimensional cubic splines, that yields a more compact formulation.
For a fixed $\left(k_{1}, k_{2}\right)$ in $\bm{k}$-space, we first perform a one-dimensional cubic spline $C_{j}\left(x_{1}\right)$ for each grid point $x_{2,j}$ in $x_{2}$-direction,
\begin{equation}\label{eq:1dpcsi}
\begin{split}
C_{j}\left(x_{1}\right)=&\frac{s_{i+1, j}}{6\Delta x_{1}}\left(x_{1}-x_{1, i}\right)^{3}+\left(\frac{f^{n}_{i+1,j}}{\Delta x_{1}}-\frac{s_{i+1,j}\Delta x_{1}}{6}\right)\left(x-x_{1,i}\right)\\
&+\frac{s_{i, j}}{6\Delta x_{1}}\left(x_{1, i+1}-x_{1}\right)^{3}+\left(\frac{f^{n}_{i,j}}{\Delta x_{1}}-\frac{s_{i,j}\Delta x_{1}}{6}\right)\left(x_{1,i+1}-x_{1}\right)
\end{split}
\end{equation}
for any $x_{1}\in \left[x_{1, i}, x_{1, i+1}\right]$, and $f^{n}_{i, j}=f\left(x_{1, i},  x_{2, j}, k_{1}, k_{2}, t^{n}\right)$ (for brevity, we omit $k_{1}, k_{2}$).
Once the coefficient table $\left(s_{i,j}\right)$ is determined, the bicubic splines can be expressed as:
 \begin{equation}\label{eq:pcsi}
\begin{split}
S\left(x_{1}, x_{2}\right)=&\frac{\sigma_{j+1}\left(x_{1}\right)}{6\Delta x_{2}}\left(x_{2}-x_{2, j}\right)^{3}+\left(\frac{C_{j}\left(x_{1}\right)}{\Delta x_{2}}-\frac{\sigma_{j+1}\left(x_{1}\right)\Delta x_{2}}{6}\right)\left(x_{2}-x_{2,j}\right)\\&+\frac{\sigma_{j}\left(x_{1}\right)}{6\Delta x_{2}}\left(x_{2, j+1}-x_{2}\right)^{3}
+\left(\frac{C_{j}\left(x_{1}\right)}{\Delta x_{2}}-\frac{\sigma_{j}\left(x_{1}\right)\Delta x_{2}}{6}\right)\left(x_{2,j+1}-x_{2}\right)
\end{split}
\end{equation}
for  $x_{2}\in \left[x_{2, j}, x_{2, j+1}\right]$. The coefficients $\left(\sigma_{j}\left(x_{1}\right)\right)$ depend on the interpolated values $C_{j}\left(x_{1}\right)$ and satisfy
\begin{equation}
\sigma_{j}\left(x_{1}\right)=0 \quad \textup{for} \quad x_{1}\notin \mathcal{X}_{1},
\end{equation}
because the cubic splines are only defined in $\mathcal{X}_{1}$.

The proposed advective approach makes full use of the exact Lagrangian advection and the integration exactly follows
the spatial characteristic lines backward in time. The time step is not restricted by the usual CFL condition,
and thus large time steps are allowed.
This may be the most notable feature of the method. However, due to the Moyal expansion \eqref{Moyal_expansion},
the time step may be still influenced a little bit by the deformational Courant number $\|\Delta t\cdot \nabla_{\bm{x}} V \|\leq 1$ as observed in the Vlasov community\cite{SonnendruckerRocheBertrand1999}.
A further remark will be given in numerical experiments (see Section \ref{sec:result:ee}).

When running simulations in the computational domain $\mathcal{X}_{1}\times\mathcal{X}_{2}\times \mathcal{K}_{1}\times \mathcal{K}_{2}$,
the boundary conditions in $\bm{x}$-space are required by the backward characteristic lines. Usually, the inflow boundary conditions are used in the literature\cite{Frensley1990,ShaoLuCai2011,JiangCaiTsu2011}.
For studying an isolated quantum system,
this work sets the boundary condition of cubic splines as not-a-knot type\cite{bk:Boor2001} and makes a simple nullification outside a sufficiently large computational domain. In this way, the outgoing waves move outside the domain transparently. How to effectively implement the
inflow boundary conditions within the advective approach
is still a going-on project.

\subsection{The spectral element method in $\bm{k}$-space}
\label{sec:method:sem}

The remaining task is to deal with the discretization in $\bm{k}$-space. Regarding to the Fourier transform nature of the nonlocal Wigner potential,
the Chebyshev spectral element method will be employed,
for such spectral discretization
provides a highly accurate spectral approximation for the pseudo-differential term\cite{ShaoLuCai2011}.

Take the element $\mathcal{K}_{r_{1}}\times \mathcal{K}_{r_{2}}$ as an example. The spectral approximation of the Wigner function reads
\begin{equation}
f\left(\bm{x} ,k_{1},k_{2}, t\right) \approx \sum_{l_{1}=0}^{N_{1}-1}
\sum_{l_{2}=0}^{N_{2}-1}a_{r_{1}r_{2},l_{1}l_{2}}
\left(\bm{x},t\right)C_{l_{1}}
\left(k_{1}\right)C_{l_{2}}\left(k_{2}\right),
\,\,\, \left(k_{1}, k_{2}\right) \in \mathcal{K}_{r_{1}}\times \mathcal{K}_{r_{2}},
\label{spectral_representation}
\end{equation}
where
\begin{equation}\label{eq:base}
C_{l_i}\left(k\right)  =
T_{l_i}\left(\eta\right), \,\,\,
k=\hat{d}_{r_i}+\frac{|\mathcal{K}_{r_i}|\eta}{2},
\,\,\,
\hat{d}_{r_i}  = d_{r_i}+\frac{|\mathcal{K}_{r_i}|}{2},
\,\,\, i=1,2,
\end{equation}
and $T_{l_i}(\eta)$ is the Chebyshev polynomial of the first kind. Substituting Eqs.~\eqref{Poisson_summation_truncated}, \eqref{eq:g}
and \eqref{spectral_representation} into Eq.~\eqref{PDO},
we arrive at the spectral approximation for
the truncated pseudo-differential term
\begin{equation}
\label{spec_nonlocal_term}
\begin{split}
g\left(\bm{x}, k_{1},k_{2}, t\right) & \approx \frac{\Delta y_{1}\Delta y_{2}}{4 \mi\hbar \pi^{2} } \sum_{\mu=-\infty}^{+\infty}\sum_{\nu=-\infty}^{+\infty}\me^{-\mi k_{1} y_{\mu}-\mi k_{2} y_{\nu}}D_{V}\left(\bm{x}, y_{\mu}, y_{\nu},t\right) \\
& \times \sum_{r_{1}=1}^{M_{1}} \sum_{r_{2}=1}^{M_{2}} \sum_{l_{1}=0}^{N_{1}-1} \sum_{l_{2}=0}^{N_{2}-1}a_{r_{1}r_{2},l_{1}l_{2}}
\left(\bm{x},t\right)O_{r_1r_2,l_1l_2}\left(y_{\mu}, y_{\nu}\right),
\end{split}
\end{equation}
where the double integral $O_{r_1r_2,l_1l_2}(y_{\mu}, y_{\nu})$ reads
\begin{equation}\label{dint}
O_{r_1r_2,l_1l_2}\left(y_{\mu}, y_{\nu}\right) =
\iint_{\mathcal{K}_{r_{1}}\times \mathcal{K}_{r_{2}}}e^{\mi k_{1}^{\prime} y_{\mu}+\mi k_{2}^{\prime} y_{\nu}} C_{l_{1}}\left(k^{\prime}_{1}\right)C_{l_{2}}\left(k^{\prime}_{2}\right)\textup{d} k_{1}^{\prime} \textup{d} k_{2}^{\prime}.
\end{equation}
The next key step is how to calculate the above integrals.
Using Eq.~\eqref{eq:base}, a direct calculation shows
\begin{equation}
\int_{\mathcal{K}_{r_{i}}}\me^{\mi k_{i}^{\prime} y }C_{l_{i}}\left(k_{i}^{\prime}\right) \textup{d} k_{i}^{\prime}=\frac{|\mathcal{K}_{r_i}|}{2}\me^{\mi y\hat{d}_{r_{i}}}O_{l_{i}}\left(\frac{|\mathcal{K}_{r_i}|y}{2}\right),\,\,\, i=1,2,
\end{equation}
and thus the double integral \eqref{dint} becomes
\begin{equation}
O_{r_1r_2,l_1l_2}\left(y_{\mu}, y_{\nu}\right)=\frac{|\mathcal{K}_{r_1}||\mathcal{K}_{r_2}|}{4}\me^{\mi y_{\mu}\hat{d}_{r_{1}}+\mi y_{\nu}\hat{d}_{r_{2}}}O_{l_{1}}\left(\frac{|\mathcal{K}_{r_1}|y_{\mu}}{2}\right)O_{l_{2}}\left(\frac{|\mathcal{K}_{r_2}|y_{\nu}}{2}\right).
\end{equation}
Here the oscillatory integral $O_{l}\left(z\right)$ is given by
\begin{equation}\label{osc_int}
O_{l}\left(z\right)=\int_{-1}^{1} e^{iz\eta} T_{l}\left(\eta\right) \textup{d} \eta,
\end{equation}
which can be represented as a linear combination of spherical Bessel functions of the first kind and thus can be calculated analytically
by exploiting the Legendre polynomial expansion of $e^{iz\eta}$ and $T_{l}\left(\eta\right)$. For more details, one can refer to \cite{ShaoLuCai2011}.

It remains to truncate the infinite summation with respect to $\mu$ and $\nu$ in Eq.~\eqref{spec_nonlocal_term}.
If the potential function $V(x_1,x_2)$ has a compact support, the matrix $D_{V}(x_{1}, x_{2}, y_{\mu}, y_{\nu})$ is sparse and it is convenient to determine the truncation threshold by counting the number of nonzero elements\cite{ShaoLuCai2011}.
However, this approach is not appropriate for the long-range Coulomb potential, especially for the electron-electron interaction $V_{\textup{ee}}$. As a matter of fact, for the many-body problem, the truncation of $y_{\nu}$ and $y_{\mu}$ is a subtle problem and so far we have not found a general way. This  truncation should also depend on how much quantum information one wants to involve in the simulations. But fortunately, we have found in the numerical experiments that a satisfactory result can be obtained for the Gaussian wave packet simulations by a finite sequence of discrete samples. Redundant sampling only leads to a very slight correction, at the cost of a dramatic decline in efficiency (see Fig.~\ref{fig_Ly} and related explanations in Section \ref{sec:result:ee}).

\begin{remark}
A simple test is presented here to validate the accuracy of
the approximation \eqref{spec_nonlocal_term} for the pseudo-differential term as well as to calibrate the computer code. Suppose the Wigner function is given by \begin{equation}
f\left(x_{1}, x_{2}, k_{1}, k_{2}\right)=\cos\left(\alpha k_{1}\right)\cos\left(\alpha k_{2}\right),
\end{equation}
then we have a close formula for the pseudo-differential term as
\begin{equation}\label{PDO_exact}
g\left(x_{1},x_{2}, k_{1},k_{2}\right)= \frac{\Delta y_{1}\Delta y_{2}}{4\mi\hbar \pi^{2}} \sum_{\mu=-\infty}^{+\infty}\sum_{\nu=-\infty}^{+\infty} D_{V}\left(x_{1}, x_{2}, y_{\mu}, y_{\nu}\right) \me^{-\mi k_{1} y_{\mu}-\mi k_{2} y_{\nu}}I_{\mu}I_{\nu},
\end{equation}
where
\begin{equation}
\begin{split}
I_{\mu,\nu}=&\frac{1}{2\left(\alpha+y_{\mu,\nu}\right)}\sin\left[\left(\alpha+y_{\mu,\nu}\right)k_{1,\textup{max}}\right]-\sin\left[\left(\alpha+y_{\mu,\nu}\right)k_{1,\textup{min}}\right]\\
& +\frac{1}{2\left(\alpha-y_{\mu,\nu}\right)}\sin\left[\left(\alpha-y_{\mu,\nu}\right)k_{1,\textup{max}}\right]-\sin\left[\left(\alpha-y_{\mu,\nu}\right)k_{1,\textup{min}}\right].
\end{split}
\end{equation}

To facilitate a comparison between the spectral approximation \eqref{spec_nonlocal_term} and the exact value given in Eq.~\eqref{PDO_exact},
we choose \begin{equation}
V\left(x_{1}, x_{2}\right)=\frac{1}{2\pi} \exp\left(-\frac{x_{1}^{2}+x_{2}^{2}}{2}\right),
\,\,\, \alpha=0.25,
\end{equation}
and take the $\bm{k}$-domain $[-{6\pi}/{5}, {6\pi}/{5}]^2$, which are divided into $4 \times 4$ elements and each element contains $16 \times 16$ collocation points. The numerical results from our implementation show that the difference is around $10^{-14}$ for the double precision computation
when $y_{\mu}$ and $y_{\nu}$ are truncated in $[-60, 60]$.
\end{remark}

\section{Numerical analysis}
\label{sec:analysis}

Suppose the Wigner function $f(\bm{x}, \bm{k}, t)$ is smooth enough and a sufficiently fine $\bm{k}$-mesh is used,
so that the Chebyshev spectral element method
can achieve a highly accurate spectral approximation in $\bm{k}$-space. Accordingly, by the standard numerical analysis on the equidistant mesh $\Delta x=\Delta x_{1}=\Delta x_{2}$, we have: When the time step $\Delta t$ is fixed,
the piecewise cubic spline interpolation \eqref{eq:pcsi}
yields a global error of the order $\mathcal{O}(\Delta x^{3})$ in $\bm{x}$-space; When the spacing $\Delta x$ is fixed, the error in $t$-space is
of the order $\mathcal{O}(\Delta t^{p})$ for an explicit $p$-step method and $\mathcal{O}(\Delta t^{p+1})$ for an implicit $p$-step method,
implying the order of $\mathcal{O}(\Delta t^3)$
for the explicit three-step method \eqref{explicit_3_step} adopted 
in the current implementation. 
The remaining of this section is to further illustrate that
the proposed third-order advective-spectral-mixed scheme for time-dependent many-body Wigner equation is capable of preserving the total mass as well as the
physical symmetry relation as stated in Proposition \ref{pro:sym}.

\subsection{Mass conservation}
\label{sec:analysis:mass}

Consider first the one-body truncated Wigner equation
\begin{equation}\label{1d_wigner}
\frac{\partial f\left(x,k, t\right)}{\partial t}+\frac{\hbar k}{m} \frac{\partial f\left(x,k, t\right)}{\partial x} +g\left(x, k, t\right)=0,
\end{equation}
in the domain $\mathcal{X}\times\mathcal{K}$, and we set $|\mathcal{K}|\Delta y=2\pi$ as requested by Eq.~\eqref{conservation_condition} with which we have
\begin{equation}\label{eq:gint}
G(x,t) : = \int_{\mathcal{K}} g\left(x, k, t\right)\textup{d}k \equiv 0.
\end{equation}

Suppose the $\bm{x}$-space is divided into $N$ non-overlapping
equidistant cells with the spacing $\Delta x$ plus two semi-bounded intervals:
\begin{equation}
\mathcal{X}_{-1}=(-\infty, x_{0}], \,\,\,\mathcal{X}_{i}=\left[x_{i}, x_{i+1}\right], \,\,\,  \mathcal{X}_{N}=[x_{N}, +\infty), \,\,\, i=0,1,\cdots,N-1,
\end{equation}
then $\mathcal{X}=\bigcup_{i=0}^{N-1}\mathcal{X}_{i}$.
According to Eq.~\eqref{implicit_advective_approx},
the explicit $p$-step approximation for Eq.~\eqref{1d_wigner} becomes
\begin{equation}\label{advective_form}
f^{n+1}\left(x, k\right)=f^{n}\left(x-h, k\right)+\Delta t \sum_{s=1}^{p} \gamma_{s}g^{n+1-s}\left(x-s h, k\right),
\end{equation}
where $h = \hbar k\Delta t/m$ denotes the shift occurring in $\Delta t$ for a given wavenumber $k$, and integrating it with respect to $x$ in the cell $\mathcal{X}_{i}$ leads to the \textit{conservative form}:
\begin{equation}\label{eq:cell}
\int_{x_{i}}^{x_{i+1}}f^{n+1}\left(x,k \right)\textup{d}x=\int_{x_{i}-h}^{x_{i+1}-h}f^{n}\left(x,k\right)\textup{d}x+\Delta t\sum_{s=1}^{p} \gamma_{s} \int_{x_{i}-s h}^{x_{i+1}-s h} g^{n+1-s}\left(x,k\right)\textup{d}x.
\end{equation}

Let's deal with the first term in the righthand side of Eq.~\eqref{eq:cell} for $i\in\{0,1,\cdots,N\}$. Without loss of generality,
we assume the shift $h\geq 0$ (i.e., $k\geq 0$,
the wave is traveling from left to right) and let $\beta = h/\Delta x - [h/\Delta x]$ (i.e., the remainder of $h/\Delta x$).

$\bullet$ If $[x_{i}-h, x_{i+1}-h]$ contains the first grid point $x_{0}$, let $i_0 = i$. That is, $[x_{i_0}-h, x_{i_0+1}-h] \subset (-\infty, x_{1}]$, then we have
\begin{equation}\label{eq:cell_part1}
\begin{split}
\int_{x_{i_0}-h}^{x_{i_0+1}-h}f^{n}(x,k)
\textup{d}x &=\int_{x_{0}}^{x_{i_0+1}-h}f^{n}\left(x,k\right)\textup{d}x +
\Phi^{n}_{\textup{in}}\left(k\right)
-\int_{-\infty}^{x_{i_0}-h}f^{n}(x,k)\textup{d}x\\
&=\frac{f_{0}^{n} }{2}\Delta x +\frac{f_{1}^{n} }{2}\Delta x -\frac{s_{0}}{24}\left(\Delta x\right)^{3}-\frac{s_{1}}{24}(\Delta x)^{3}+\frac{\beta \Delta x}{2}[-\beta f_{0}^{n}\\
&+(\beta-2)f_{1}^{n}]
+\frac{\beta^{2} \left(\Delta x\right)^{3}}{24}[-(\beta^{2}-2)s_{0}
+(\beta^{2}-4\beta+4)s_{1}]\\
&+\Phi^{n}_{\textup{in}}\left(k\right)
-\int_{-\infty}^{x_{i_0}-h}f^{n}(x,k)\textup{d}x,
\end{split}
\end{equation}
where we have used the one-dimensional piecewise cubic spline interpolation \eqref{eq:1dpcsi} in $[x_0,x_1]$ for $f^n(x,k)$ to calculate the second integral in the first line,
and
\begin{equation}\label{eq.def1}
\Phi_{\textup{in}}^{n}\left(k\right)
:=\int_{-\infty}^{x_{0}}f^{n}\left(x,k\right)\textup{d}x
\end{equation}
denotes the total inflow from the left to $x_0$ at $t=t_{n}$.

$\bullet$ For $i_0<i<N$, there must exist a gird point $x_{j} \in [x_{i}-h, x_{i+1}-h]$, and then $[x_{i}-h, x_{i+1}-h] \subset [x_{j-1}, x_{j+1}]$.
Using the one-dimensional piecewise cubic spline interpolation \eqref{eq:1dpcsi} in $[x_{j-1}, x_{j+1}]$ for $f^n(x,k)$ and integrating it directly in $[x_{i}-h, x_{i+1}-h]$ yields
\begin{equation}\label{eq:cell_part2}
\begin{split}
\int_{x_{i}-h}^{x_{i+1}-h}f^{n}(x, k)\textup{d}x&=\frac{f_{j}^{n}}{2}\Delta x +\frac{f_{j+1}^{n}}{2}\Delta x -\frac{s_{j}}{24}(\Delta x)^{3}-\frac{s_{j+1}}{24}(\Delta x)^{3}\\
&+\frac{\beta \Delta x}{2}\left[\beta f_{j-1}^{n}-\left(2\beta-2\right)f_{j}^{n}
+\left(\beta-2\right)f_{j+1}^{n}\right]
+\frac{\beta^{2} (\Delta x)^{3}}{24}\\
&\times[(\beta^{2}-2)s_{j-1}
- (2\beta^{2}- 4\beta+2)s_{j}+(\beta^{2}-4\beta+4)s_{j+1}].
\end{split}
\end{equation}

$\bullet$ For $i=N$, we denote $j_0 = j$ when $[x_{j-1}, x_{j+1}] \supset [x_{N-1}-h, x_{N}-h]$. A similar calculation to Eq.~\eqref{eq:cell_part1} leads to
\begin{equation}\label{eq:cell_part3}
\begin{split}
\int^{+\infty}_{x_{N}-h}f^{n}(x,k)\textup{d}x&=\int^{x_{j_{0}+1}}_{x_{N}-h}f^{n}(x,k)\textup{d}x+\int_{x_{j_{0}+1}}^{x_{N}}f^{n}(x,k)\textup{d}x
+\int_{x_{N}}^{+\infty}f^{n}(x,k)\textup{d}x \\
&=\int^{x_{N}-\beta \Delta x}_{x_{N}-h}f^{n}\left(x,k\right)\textup{d}x+\int_{x_{N}-\beta \Delta x}^{x_{N}}
f^{n}\left(x,k\right)\textup{d}x+\int_{x_{N}}^{+\infty}f^{n}
\left(x,k\right)\textup{d}x \\
&=\int^{x_{N}-\beta\Delta x}_{x_{N}-h}f^{n}\left(x, k\right)\textup{d}x+\frac{f_{N-1}^{n} }{2}\Delta x +\frac{f_{N}^{n} }{2}\Delta x -\frac{s_{N-1}}{24}\left(\Delta x\right)^{3}\\
&-\frac{s_{N}}{24}\left(\Delta x\right)^{3}+\frac{\beta \Delta x}{2}\left[\beta f_{N-1}^{n}-\left(\beta-2\right)f_{N}^{n}\right]+\Phi^{n}_{\textup{out}}\left(k\right)\\
&+\frac{\beta^{2} \left(\Delta x\right)^{3}}{24}[(\beta^{2}-2)s_{N-1}-(\beta^{2}-4\beta+4)s_{N}],
\end{split}
\end{equation}
where the first integral in the third line can be calculated in the same way as Eq.~\eqref{eq:cell_part2},
and
\begin{equation}\label{eq.def2}
\Phi^{n}_{\textup{out}}\left(k\right)
:=\int_{x_{N}}^{+\infty}f^{n}\left(x,k\right)\textup{d}x
\end{equation}
denotes the total outflow from the $x_N$ to right at $t=t_{n}$.

For the nonlocal term $g(x, k, t)$,
we can also
define a similar total ``inflow" and ``outflow" contributed by the source term as
\begin{equation}
\label{eq.def3}
\Psi^{n}_{\textup{in}}\left(k\right)=\int_{-\infty}^{x_{0}} g\left(x, k, t_{n}\right)\textup{d}k,\,\,\, \Psi^{n}_{\textup{out}}\left(k\right)=\int^{\infty}_{x_{N}} g\left(x, k, t_{n}\right)\textup{d}k,
\end{equation}
and thus derive similar expressions as shown in Eqs.~\eqref{eq:cell_part1}, \eqref{eq:cell_part2} and \eqref{eq:cell_part3}, for the second term in the righthand side of Eq.~\eqref{eq:cell} for $i\in\{0,1,\cdots,N\}$.
For simplicity, we neglect the details here.

Now using Eqs.~\eqref{eq:cell_part1}, \eqref{eq:cell_part2} and \eqref{eq:cell_part3},
the summation of Eq.~\eqref{eq:cell} with respect to $i$ from $-1$ to $N$ yields
\begin{equation}\label{eq:all_cell}
\begin{split}
\int_{\mathcal{X}}f^{n+1}\left(x, k\right)\textup{d}x=&\int_{\mathcal{X}}f^{n}\left(x, k\right)\textup{d}x+\Delta t\sum_{s=1}^{p} \gamma_{s} \int_{\mathcal{X}} g^{n+1-s}\left(x, k\right)\textup{d}x\\
&-\Phi_{\textup{in}}^{n+1}\left(k\right)+\Phi^{n}_{\textup{in}}\left(k\right)-\Phi^{n+1}_{\textup{out}}\left(k\right)+\Phi^{n}_{\textup{out}}\left(k\right)\\
&+\Delta t \sum_{s=1}^{p} \gamma_{s}\left[\Psi^{n+1-s}_{\textup{in}}\left(k\right)+\Psi^{n+1-s}_{\textup{out}}\left(k\right)\right],
\end{split}
\end{equation}
where we have
used the following basic property of piecewise cubic spline
\begin{equation}
\left(f^{n}_{0}+2\sum_{i=1}^{N-1}f_{i}^{n}+f^{n}_{N}\right)\frac{\Delta x}{2}-\left(s_{0}+2\sum_{i=1}^{N-1}s_{i}+s_{N}\right)\frac{\left(\Delta x\right)^{3}}{24} = \int_{\mathcal{X}}f^{n}\left(x, k\right)\textup{d}x.
\end{equation}

Finally, integrating Eq.~\eqref{eq:all_cell} with respect to $k$ in the domain $\mathcal{K}$ and using Eq.~\eqref{eq:gint}, we can readily obtain
\begin{equation}
\iint_{\mathcal{X}\times \mathcal{K}}  f^{n+1}\left(x, k\right)\textup{d}x \textup{d}k=\iint_{\mathcal{X}\times \mathcal{K}} f^{n}\left(x, k\right)\textup{d}x \textup{d}k,
\end{equation}
provided that the total inflow and outflow are in balance,
i.e., the total outflow cancels the total inflow
at any moment. While using the not-a-knot boundary condtions for piecewise cubic splines\cite{bk:Boor2001}, the Wigner function 
is able to cross the boundaries so transparently that
the total outflow often exceeds the total inflow in a small computational domain. 
Conseqeuntly, 
in order to reach the flow balance and thus the mass conservation,
a relatively large domain, allowing the desired Wigner function far away from the boundaries, must be used. More detailed discussion on this issue can be found in Section \ref{sec:result}.

The above approach to show mass conservation for the one-body situation can be straightforwardly extended to
the two-body situation by exploiting the fact that
the construction of two dimensional cubic splines can be performed through several one-dimensional splines (see Section \ref{sec:method:adv}). The details are neglected for saving space.

\subsection{Physical symmetry relation}
\label{sec:analysis:sym}

Proposition \ref{pro:sym} has shown that
the physical symmetry relation \eqref{symmetry_relation} is naturally embedded in the Wigner equation.
This section  will further show that
such physical symmetry relation is still preserved in the advective-spectral-mixed method. Actually, let $f^{n}$ be the numerical Wigner function at $t=n\Delta t$ calculated from Eqs.~\eqref{implicit_advective_approx} and \eqref{general_bicubic_sp},
and $\sigma$ be the isomorphism defined in Eq.~\eqref{sigma}.
Then we are able to show $\sigma f^{n}=f^{n}$ provided both
$V(x_{1}, x_{2}, t)=V(x_{2}, x_{1}, t)$ and
$\sigma f_{0}=f_{0}$ hold. The verification can be completed by induction on $n=0,1,\cdots$ as follows.

The initial data satisfying $\sigma f_{0}=f_{0}$
implies directly $\sigma f^{0}=f^{0}$ for $n=0$.
Suppose $\sigma f^{l}=f^{l}$ holds for $l=0,1,\cdots,n$.
Then we are going to show $\sigma f^{n+1}=f^{n+1}$, which
is reduced to verify
\begin{align}\label{equiv_relation}
\sigma f^{l}\left(x_{1}-v_{1}\Delta t, x_{2}-v_{2}\Delta t, k_{1}, k_{2}\right) &= f^{l}\left(x_{1}-v_{1}\Delta t, x_{2}-v_{2}\Delta t, k_{1}, k_{2}\right),\\
\sigma g^{l}\left(x_{1}-v_{1}\Delta t, x_{2}-v_{2}\Delta t, k_{1}, k_{2}\right) &= g^{l}\left(x_{1}-v_{1}\Delta t, x_{2}-v_{2}\Delta t, k_{1}, k_{2}\right),
\end{align}
for $l=n,n-1,\cdots,n+1-p$ by using the recursion approximation \eqref{implicit_advective_approx}.
Noting that the nonlocal pseudo-differential operator
$\Theta_V^T$ acts linearly on $f$ (see Eq.~\eqref{PDO}),
it is sufficient to verify Eq.~\eqref{equiv_relation} by the definition of $g$ in Eq.~\eqref{eq:g}.
For simplicity, we only consider the case $l = n$ and the others can be proved in the same way.

Both $f^{n}(x_{2}-v_{2}\Delta t, x_{1}-v_{1}\Delta t, k_{2}, k_{1})=\sigma f^{n}(x_{1}-v_{1}\Delta t, x_{2}-v_{2}\Delta t, k_{1}, k_{2})$ and
$f^{n}(x_{1}-v_{1}\Delta t, x_{2}-v_{2}\Delta t, k_{1}, k_{2})$ are calculated through the piecewise cubic spline interpolation \eqref{general_bicubic_sp} in the cell
$[x_{2, j}, x_{2, j+1}] \times [x_{1, i}, x_{1, i+1}]$ and
$[x_{1, i}, x_{1, i+1}] \times [x_{2, j}, x_{2, j+1}]$, respectively, where $x_{1}-v_{1}\Delta t \in [x_{1, i}, x_{1, i+1}], x_{2}-v_{2}\Delta t \in [x_{2, j}, x_{2, j+1}]$.
That is,
\begin{align}
f^{n}\left(x_{1}-v_{1} \Delta t, x_{2}-v_{2} \Delta t, k_{1}, k_{2}\right)&=\sum_{\nu=0}^{3} \sum_{\kappa=0}^{3}  \eta_{\nu\kappa} \beta_{1}^{\nu} \beta_{2}^{\kappa},\label{eq:fn1}\\
f^{n}\left(x_{2}-v_{2} \Delta t, x_{1}-v_{1} \Delta t, k_{2}, k_{1}\right)&=\sum_{\kappa=0}^{3} \sum_{\nu=0}^{3}  \tilde{\eta}_{\nu\kappa} \beta_{2}^{\nu} \beta_{1}^{\kappa}
=\sum_{\kappa=0}^{3} \sum_{\nu=0}^{3}  \tilde{\eta}_{\kappa\nu} \beta_{1}^{\nu}\beta_{2}^{\kappa},
\label{eq:fn2}
\end{align}
where $\beta_{1}=x_{1}-v_{1} \Delta t-x_{1, i}$,
$\beta_{2}=x_{2}-v_{2} \Delta t-x_{2, j}$,
and $\eta_{\nu \kappa}$,
$\tilde{\eta}_{\kappa \nu}$ are the interpolation coefficients. Thanks to the uniqueness of bicubic splines,
it can be easily obtained
\begin{equation}
\eta_{\nu \kappa}= \tilde{\eta}_{\kappa \nu }, \quad  \nu,\kappa=0, 1, 2, 3,
\end{equation}
because both interpolation points and not-a-knot boundary conditions are identical as well as
the function values on the interpolation points are the same due to the induction assumption.
Thus we obtain $\sigma f^n = f^n$ from Eqs.~\eqref{eq:fn1} and \eqref{eq:fn2} and the verification is finished.

\section{Numerical experiments}
\label{sec:result}

We implemented the advective-spectral-mixed method for
both one-body and two-body situations in one-dimensional space. In $\bm{k}$-space, we are able to take
the full advantage of fast Fourier transforms to improve the computational efficiency thanks to the Gauss-Chebyshev collocation points adopted in each element by calling the related subroutines in FFTPACK\cite{Swarztrauber1982},
and use the  subroutine BESSJY\cite{bk:PressTeukolskyVetterlingFlannery1992}
to calculate the spherical Bessel functions of the first kind requested by the oscillatory integral \eqref{osc_int}. In $\bm{x}$-space, the piecewise cubic splines \eqref{general_bicubic_sp} are referred to  the PSPLINE implementation --- a library of spline and Hermite cubic interpolation routines for 1d, 2d, and 3d datasets on rectilinear grids\cite{web:McCune2010}.
Since the calculations in $\bm{k}$- and $\bm{x}$-space are completely decoupled, a straightforward parallelization based on the multithread technology provided by OpenMP is further adopted to accelerate the simulations.

To visualize conveniently the two-body Wigner function in the computational domain
\begin{equation}
\Omega =  \mathcal{X}\times\mathcal{K},
\,\,\,
\mathcal{X}=\mathcal{X}_1\times\mathcal{X}_2,\,\,\,
\mathcal{K} = \mathcal{K}_1\times\mathcal{K}_2,
\end{equation}
we plot the reduced one-body Wigner function\cite{CancellieriBordoneJacoboni2007}
\begin{equation}\label{reduced_Wigner}
F(x, k, t) :=\iint_{\mathcal{X}_2\times\mathcal{K}_2} f(x, x_{2}, k, k_{2}, t) \textup{d} x_{2} \textup{d} k_{2}+\iint_{\mathcal{X}_1\times\mathcal{K}_1} f(x_{1}, x, k_{1}, k, t) \textup{d} x_{1} \textup{d} k_{1},
\end{equation}
which projects the two-body Wigner function onto the one-dimensional phase space. The numerical performance is evaluated by the $L^{2}$-error $\epsilon_{2}(t)$, the $L^{\infty}$-error $\epsilon_{\infty}(t)$,
the error for the physical symmetry relation  $\epsilon_{\textup{sym}}(t)$,
and the variation of total mass $\epsilon_{\textup{mass}}(t)$,
defined respectively as follows
\begin{align}
\epsilon_{2}(t) &=\left[\iint_{\Omega} \left(f^{\textup{ref}}\left(\bm{x},\bm{k},t\right)-f^{\textup{num}}\left(\bm{x},\bm{k},t\right)\right)^{2}\textup{d}\bm{x}\textup{d} \bm{k}\right]^{\frac{1}{2}},\label{eq:e2}\\
\epsilon_{\infty}(t) &=\max_{(\bm{x},\bm{k})\in\Omega}\left\{|f^{\textup{ref}}\left(\bm{x},\bm{k},t\right)-f^{\textup{num}}\left(\bm{x},\bm{k},t\right)|\right\}, \label{eq:ef}\\
\epsilon_{\textup{sym}}(t) &=\max_{(\bm{x},\bm{k})\in\Omega}\left\{\left|f^{\textup{num}}\left(x_{1}, x_{2}, k_{1}, k_{2}, t\right)-f^{\textup{num}}\left(x_{2}, x_{1}, k_{2}, k_{1}, t\right)\right|\right\}, \label{eq:esym}\\
\epsilon_{\textup{mass}}(t) &=\iint_{\Omega} f^{\textup{num}}\left(\bm{x},\bm{k},t\right)\textup{d}\bm{x}\textup{d} \bm{k}-\iint_{\Omega} f^{\textup{ref}}\left(\bm{x},\bm{k},t=0\right)\textup{d}\bm{x}\textup{d} \bm{k},\label{eq:emass} 
\end{align}
where $f^{\textup{ref}}$ and $f^{\textup{num}}$ denote the reference and numerical solution, respectively.
According to the numerical analysis shown in Section \ref{sec:analysis}, both $\epsilon_{2}$ and $\epsilon_{\infty}$,
depending on the mesh size and the truncation order,
should reflect the third-order convergence against the spatial spacing and the time step when a high spectral accuracy is reached in $\bm{k}$-space; $\epsilon_{\textup{sym}}$ must around the matching resolution provided both initial data and external potential are symmetric; the vanishing of $\epsilon_{\textup{mass}}$ relies on both Eq.~\eqref{eq:gint} and boundary conditions. Actually, in order to maintain an almost constant mass, on one hand, we should make sure the simulated quantum system be far away from the boundaries to guarantee the total outflow cancels the total inflow,
due to the not-a-knot boundary conditions adopted in the current implementation. On the other hand, Eq.~\eqref{eq:gint} holds only in the sense of spectral approximation and can be measured by $\epsilon_{G}(t)$ as follows
\begin{equation}
\epsilon_{G}(t) =
\max_{\bm{x}\in\mathcal{X}}\left\{\left|G^{\textup{num}}\left(\bm{x}, t\right)\right|\right\},
\label{eq:eG}
\end{equation}
where $G^{\textup{num}}(\bm{x},t)$ is corresponding numerical approximation for $G(\bm{x},t)$ defined in Eq.~\eqref{eq:gint},
because the spectral element method is employed in $\bm{k}$-direction and all related $\bm{k}$-integrals in Eq.~\eqref{eq:gint} are done analytically with the help of the spectral expansion \eqref{spectral_representation}. Hence, we expect a very tiny $\epsilon_{\textup{mass}}$ once enough collocation points in $\bm{k}$-space are placed in a sufficiently large computational domain as we will do in the following numerical simulations.

\begin{figure}[h]
    \centering
    \subfigure[Fermions.]{\label{fig:fermion}
    \includegraphics[width=2.9in,height=2.1in]{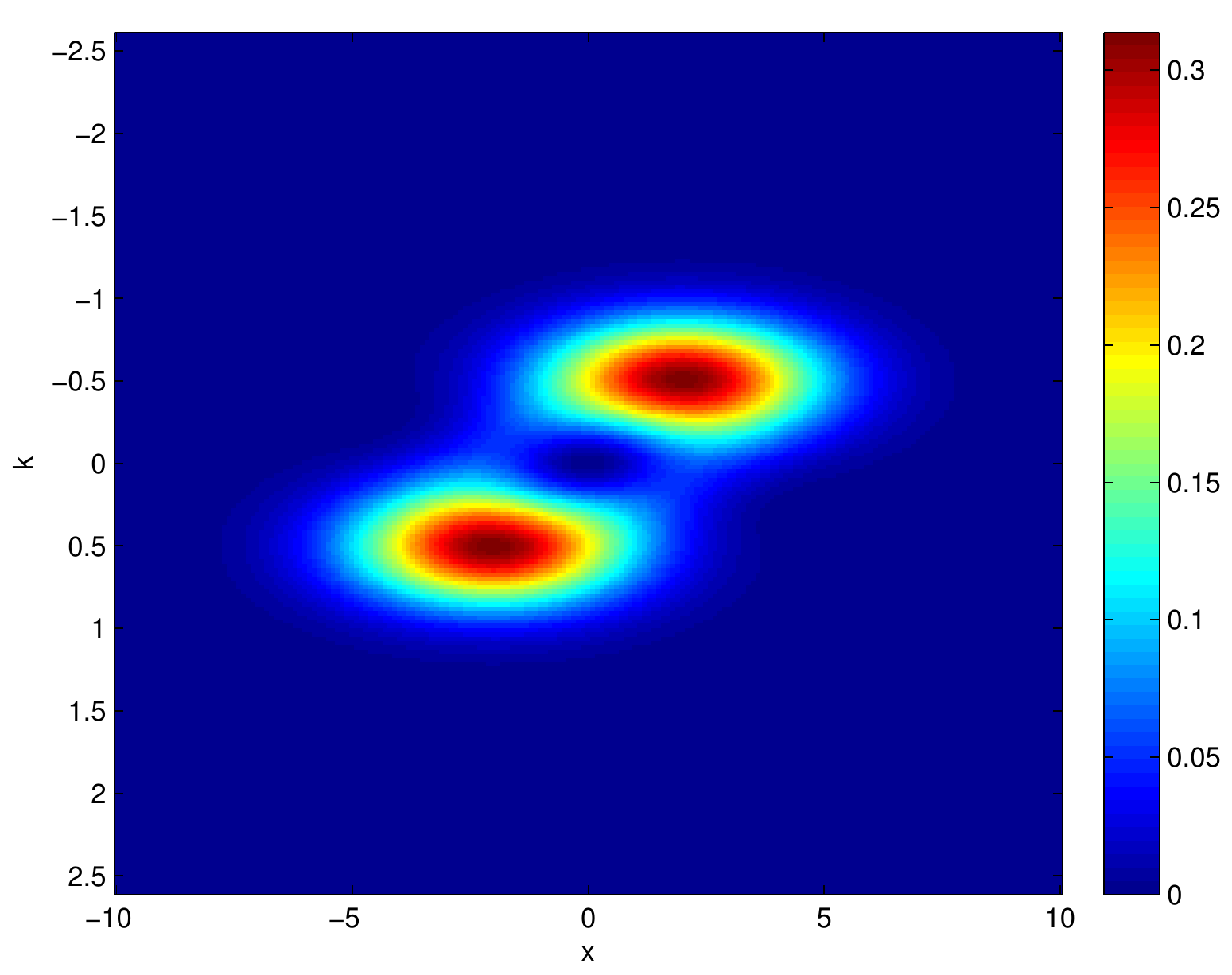}}
    \subfigure[Bosons.]{
    \includegraphics[width=2.9in,height=2.1in]{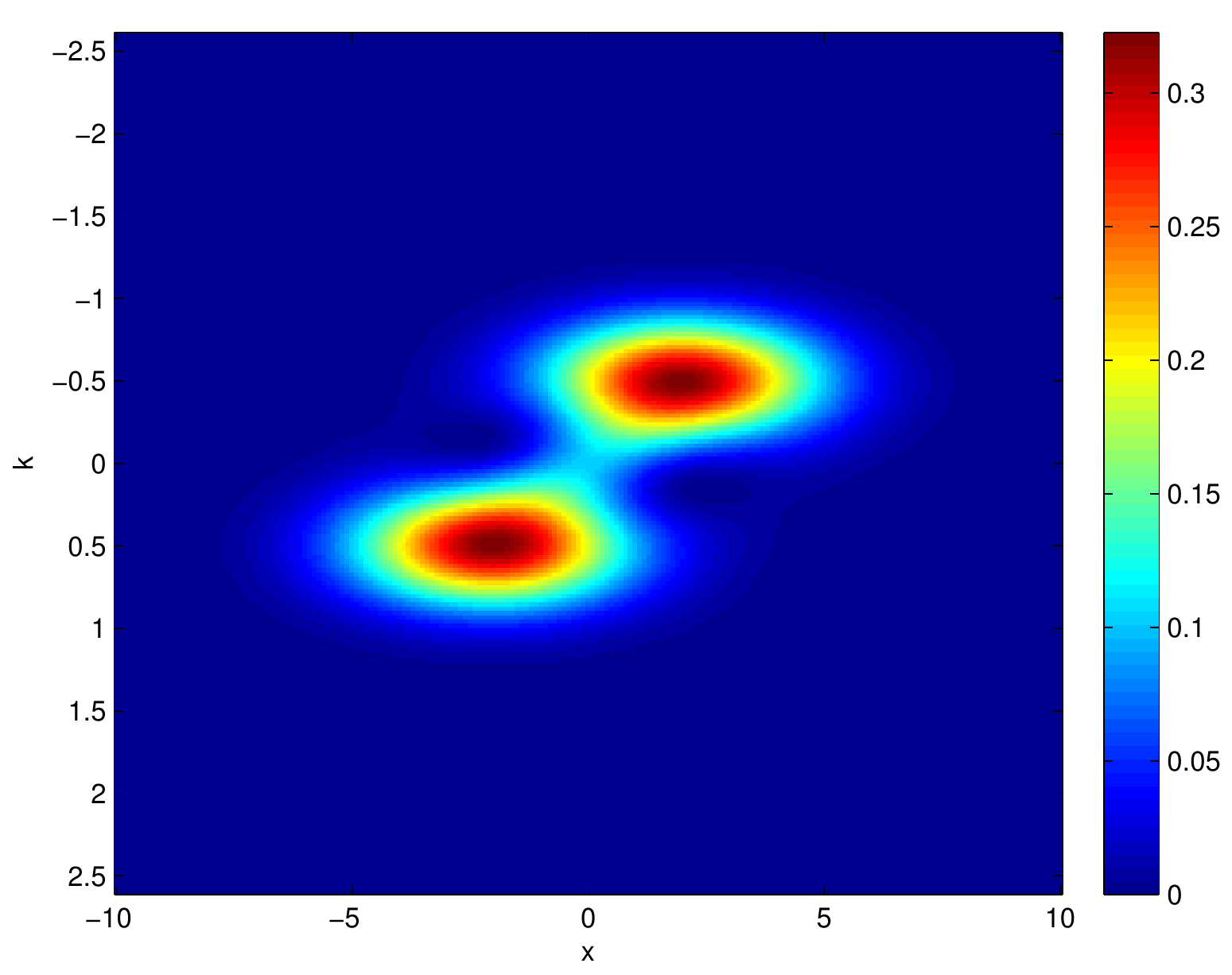}}
     \caption{\small The reduced Wigner function for two fermions and two bosons.}
     \label{fig_1}
\end{figure}

Throughout the simulations,
the atomic units $\hbar=m=e=1$
are adopted if not specified.
The initial data are constructed from Gaussian wave packets in quantum mechanics the wavefunction of which reads
\begin{equation}
\psi_{i}\left(x_i\right)=\frac{1}{\sqrt{a_{i}\sqrt{2\pi}}}\exp\left[-\frac{\left(x_i-x_{i}^{0}\right)^{2}}{4 a_i^{2}}+\mi k^{0}_{i}(x_i-x_i^0)\right],
\,\,\, i=1,2,
\end{equation}
where $x_{i}^{0}$ is the center of the wave at $t=0$,
$a_i$ is the minimum position spread,
and $k_{i}^{0}$ is the initial constant wavenumber.
The Wigner function for such Gaussian wave packet still has the Gaussian profile and its formulation is\cite{th:Biegel1997,ShaoLuCai2011,ShaoSellier2015}
\begin{equation}
\label{1D_Gaussian_wp}
f^{\textup{1D}}_{i,0}\left(x_{i},k_{i}\right)=\frac{1}{\pi} \exp\left[-\frac{\left(x_{i}-x_{i}^{0}\right)^{2}}{2a_i^{2}}-2a_i^{2}\left(k_{i}-k^{0}_{i}\right)^{2}\right], \,\,\, i=1,2.
\end{equation}
When two particles are uncorrelated,
the wave function satisfies $\psi\left(x_{1}, x_{2}\right)=\psi_1\left(x_{1}\right)\psi_2\left(x_{2}\right)$, and then the Wigner function is a simple product of two Gaussian wave packets, too. Namely, $f_{0}\left(x_{1}, x_{2}, k_{1}, k_{2}\right)=f_{1,0}^{\textup{1D}}\left(x_{1}, k_{1}\right)f_{2, 0}^{\textup{1D}}\left(x_{2}, k_{2}\right)$.
However, in order to treat a system composed of two indistinguishable fermions, we need to take into account the antisymmetric nature of the wave function $\psi\left(x_{1}, x_{2}\right)$. Such antisymmetric relation is usually fulfilled
via the Slater determinant as follows
\begin{equation}
\begin{split}
\psi\left(x_{1}, x_{2}\right)=&\frac{1}{\sqrt{2}}\begin{vmatrix} \psi_{1}\left(x_{1}\right) &   \psi_{2}\left(x_{1}\right)\\
 \psi_{1}\left(x_{2}\right)&  \psi_{2}\left(x_{2}\right)\end{vmatrix}=\frac{1}{\sqrt{2}}\psi_{1}\left(x_{1}\right)\psi_{2}\left(x_{2}\right)-\frac{1}{\sqrt{2}}\psi_{2}\left(x_{1}\right)\psi_{1}\left(x_{2}\right),
 \end{split}
\end{equation}
and then the corresponding Wigner function reads
\begin{equation}\label{asym_Wigner_init}
\begin{split}
&f_{0}^{\text{fermion}}\left(x_{1}, x_{2}, k_{1}, k_{2}\right)\\
=&\frac{1}{2\pi^{2}}\exp\left[-\frac{\left(x_{1}-x_{1}^{0}\right)^{2}}{2a^{2}}-\frac{\left(x_{2}-x_{2}^{0}\right)^{2}}{2a^{2}}-2a^{2}\left(k_{1}-k_{1}^{0}\right)^{2}-2a^{2}\left(k_{2}-k_{2}^{0}\right)^{2}\right]\\
+&\frac{1}{2\pi^{2}}\exp\left[-\frac{\left(x_{1}-x_{2}^{0}\right)^{2}}{2a^{2}}-\frac{\left(x_{2}-x_{1}^{0}\right)^{2}}{2a^{2}}-2a^{2}\left(k_{1}-k_{0}\right)^{2}-2a^{2}\left(k_{2}-k_{0}\right)^{2}\right]\\
 -&\frac{1}{\pi^{2}}\exp\left[-\frac{\left(x_{1}-x_{1}^{0}\right)^{2}+\left(x_{1}-x_{2}^{0}\right)^{2}+\left(x_{2}-x_{1}^{0}\right)^{2}+\left(x_{2}-x_{2}^{0}\right)^{2}}{4a^{2}}\right] \\
  \times& \exp\left[\frac{\left(x_{1}^{0}-x_{2}^{0}\right)^{2}}{4 a^{2}}-2a^{2}\left(k_{1}-\frac{k_{1}^{0}+k_{2}^{0}}{2}\right)^{2}-2a^{2}\left(k_{2}-\frac{k_{1}^{0}+k_{2}^{0}}{2}\right)^{2}\right] \\
  \times &\cos\left[\left(x_{1}^{0}-x_{2}^{0}\right)\left(k_{1}-k_{2}\right)-\left(k_{1}^{0}-k_{2}^{0}\right)\left(x_{1}-x_{2}\right)\right],
\end{split}
\end{equation}
where $a :=a_1\equiv a_2$. Similarly, we can construct the wave function for
a system composed of two indistinguishable bosons
\begin{equation}
\psi\left(x_{1}, x_{2}\right)=\frac{1}{\sqrt{2}}\psi_{1}\left(x_{1}\right)\psi_{2}\left(x_{2}\right)
+\frac{1}{\sqrt{2}}\psi_{2}\left(x_{1}\right)\psi_{1}\left(x_{2}\right),
\end{equation}
and obtain the corresponding Wigner function by replacing the factor $-1/\pi^{2}$ in the fourth line of Eq.~\eqref{asym_Wigner_init} with $1/\pi^{2}$.
Fig.~\ref{fig_1} plots the reduced Wigner function $F\left(x, k, t\right)$ (see Eq.~\eqref{reduced_Wigner})
for two fermions and two bosons by setting $x_{1}^{0}=-2$, $x_{2}^{0}=2$, $k_{1}^{0}=0.5$, $k_{2}^{0}=-0.5$, and $a=2$.
The exchange-correlation hole (or called the Fermi hole) at the centre $((x_{1}^{0}+x_{2}^{0})/2, (k_{1}^{0}+k_{2}^{0})/2)$
due to the Pauli exclusion principle, preventing the fermions  from occupying the same quantum state (position and momentum), is clearly shown there for the Fermi system.
On the contrary, such hole structure is not visible for the Boson system.

In the subsequent numerical simulations,
except for the first simulation in Section \ref{sec:result:barrier},
we will adopt the symmetric domains with respect to the origin point
\begin{equation}
\mathcal{X}_1 \equiv \mathcal{X}_2 = [-L_x, L_x], \,\,\, \mathcal{K}_1 \equiv \mathcal{K}_2 = [-L_k, L_k],
\end{equation}
as well as the same mesh for each particle with the spatial spacing $\Delta x$. According to the constraint \eqref{conservation_condition}, we can easily obtain the $\bm{y}$-spacing
\begin{equation}
\Delta y := \Delta y_1 \equiv \Delta y_2 =  \pi/L_k,
\end{equation}
and then the truncated Wigner potential \eqref{Poisson_summation_truncated} is calculated by further restricting both $y_\nu$ and $y_\mu$ in a symmetric domain
\begin{equation}\label{eq:ydom}
\mathcal{Y}_1 \equiv \mathcal{Y}_2 = [-L_y,L_y].
\end{equation}

\subsection{Free advection of two fermions}
\label{sec:result:free}

\begin{figure}[h]
    \centering
    \subfigure[$t=2$.]{
      \includegraphics[width=2.9in,height=2.1in]{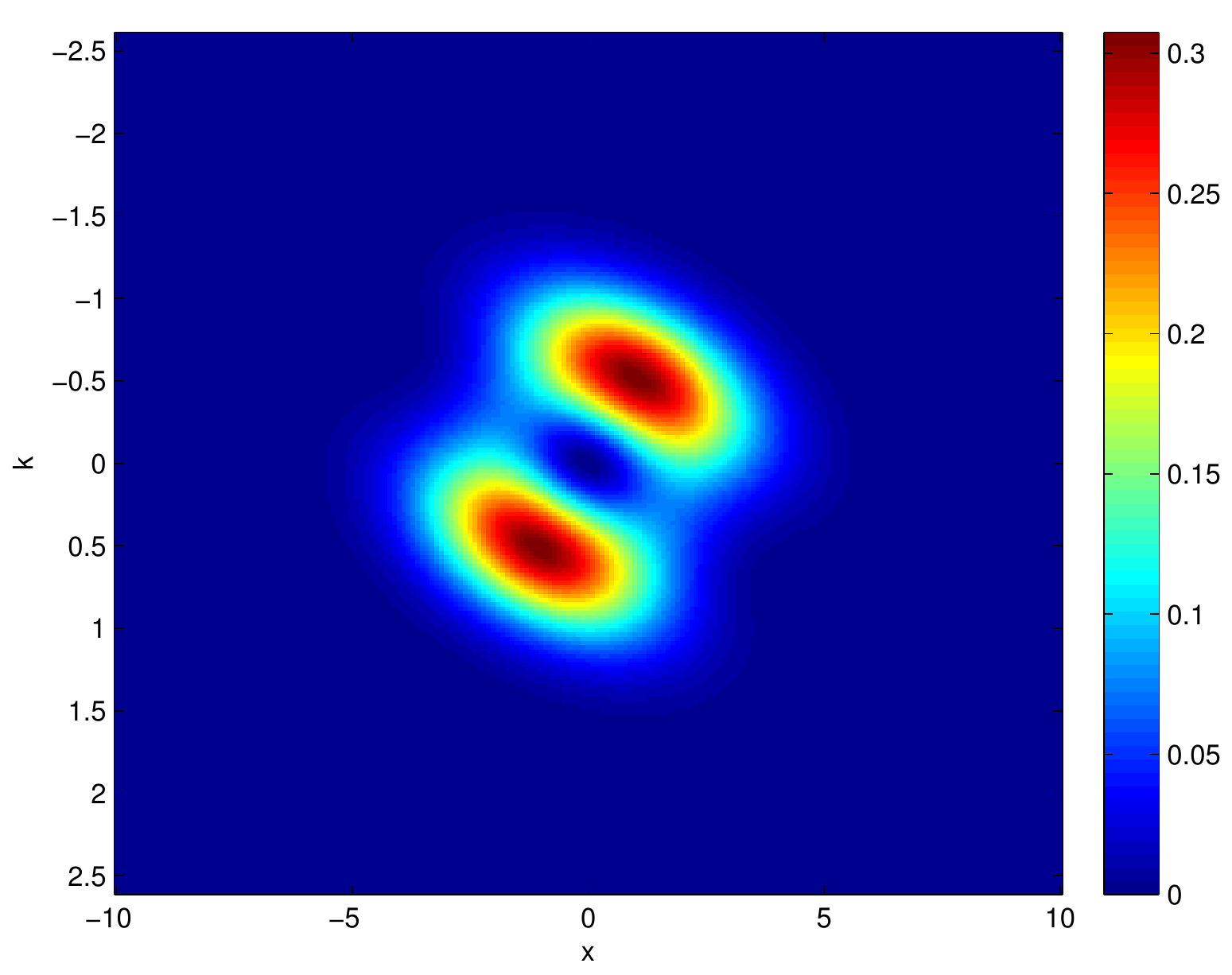}}
    \subfigure[$t=4$.]{
    \includegraphics[width=2.9in,height=2.1in]{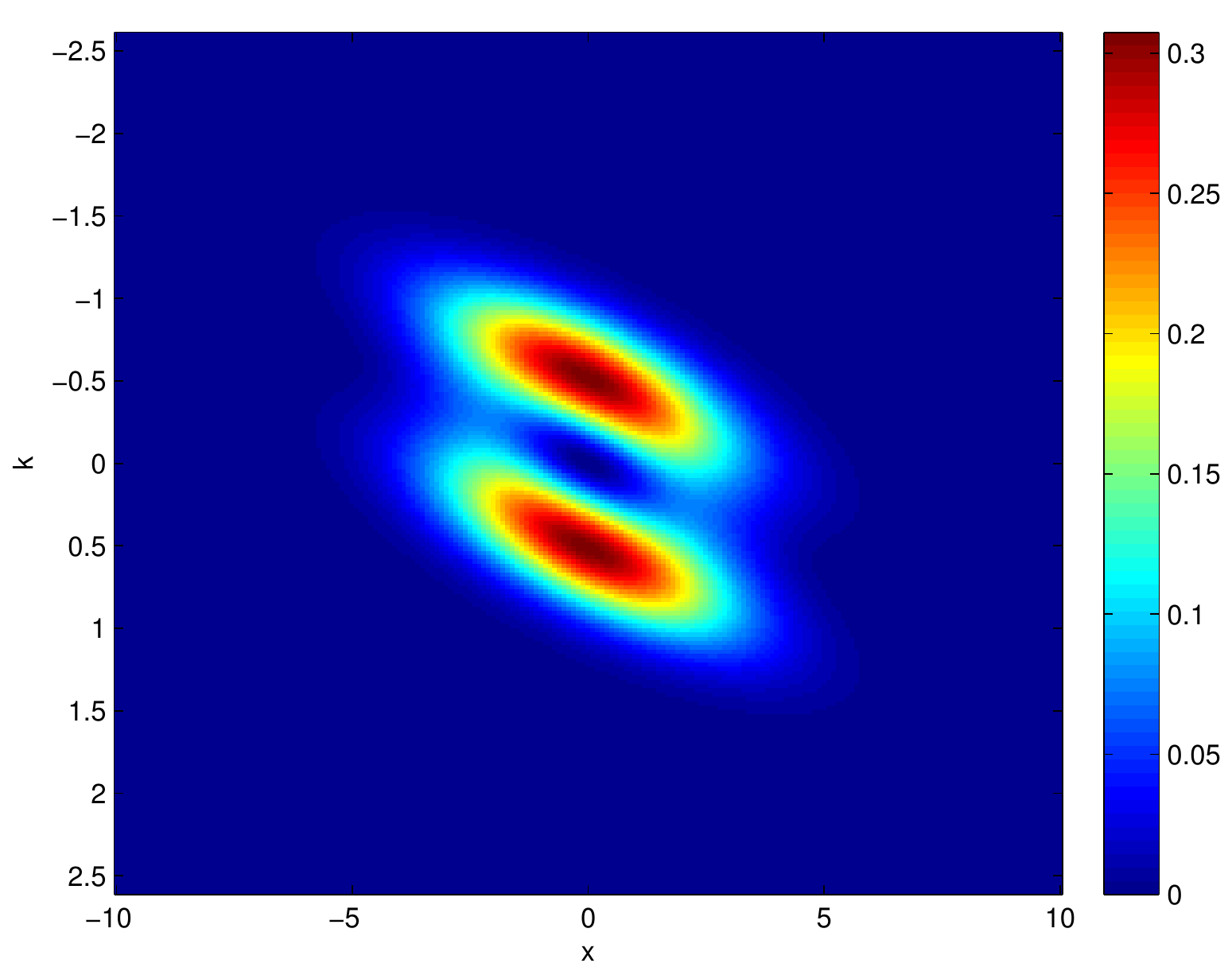}}\\
    \centering
    \subfigure[$t=6$.]{
    \includegraphics[width=2.9in,height=2.1in]{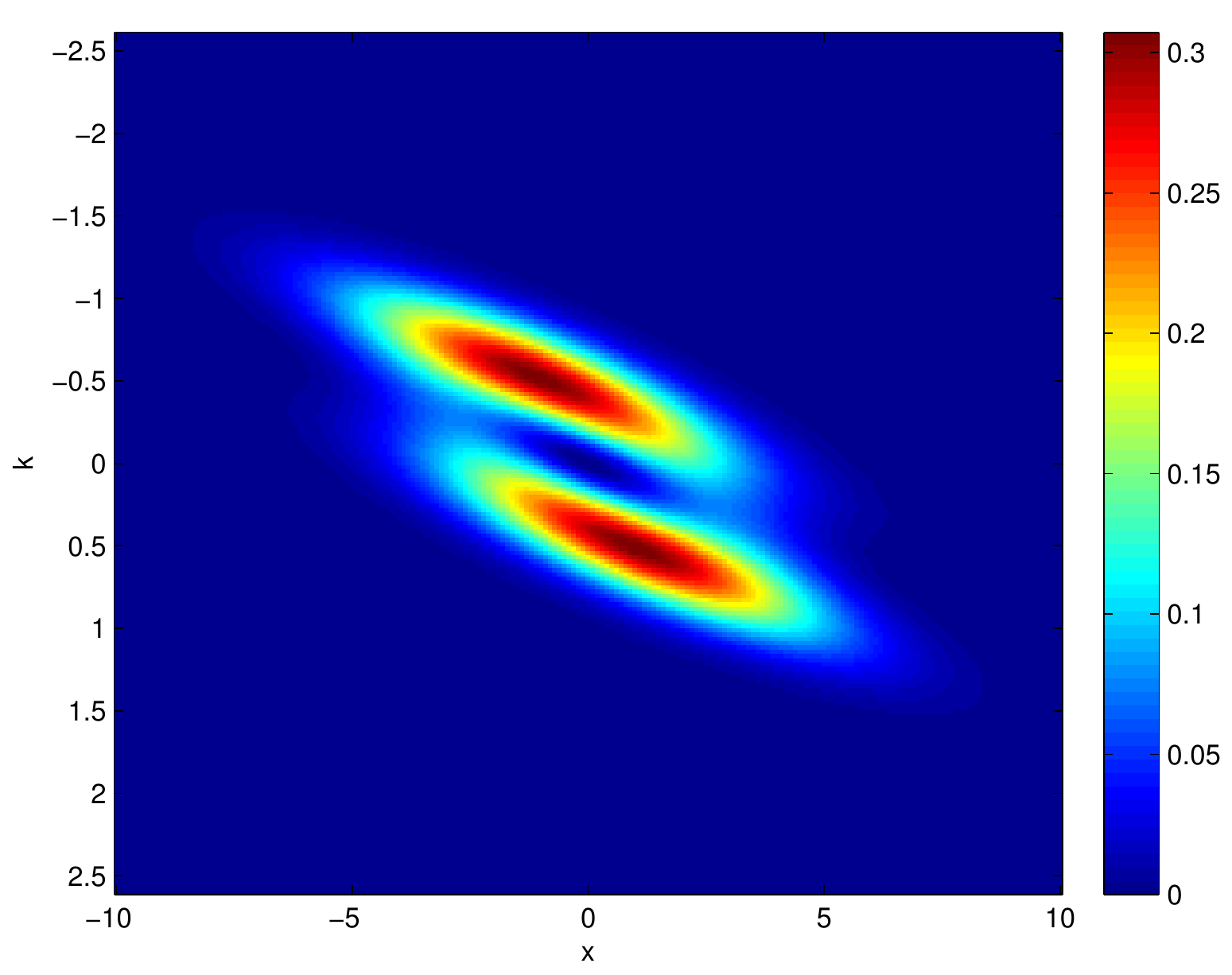}}
    \subfigure[$t=8$.]{
    \includegraphics[width=2.9in,height=2.1in]{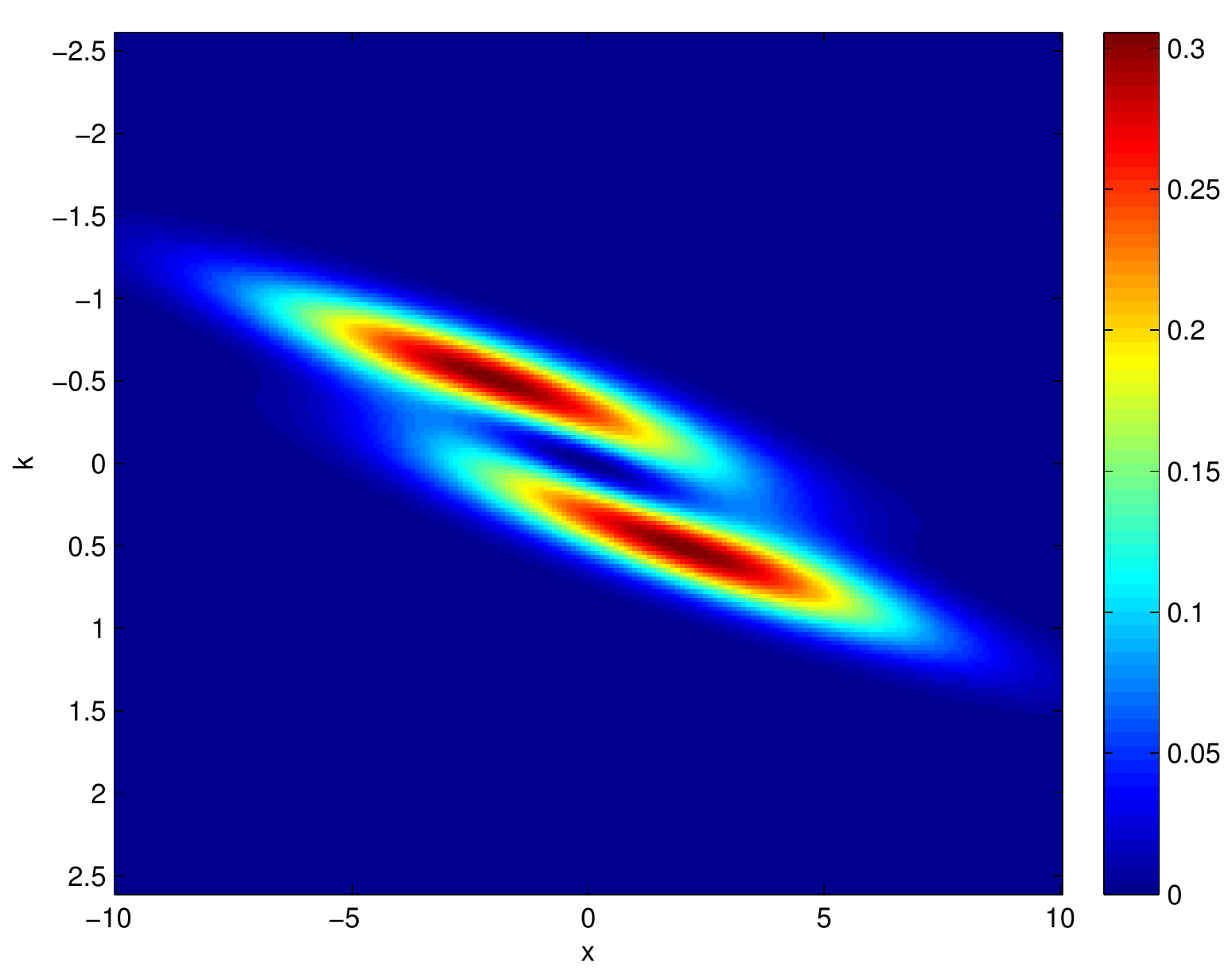}}
     \caption{\small The free advection of two fermions: The reduced Wigner function at different time instants.}
     \label{fig_4}
\end{figure}

\begin{figure}[h]
    \centering
        \subfigure[$L^2$-error history.]{
    \label{fig3:linear}
    \includegraphics[width=2.9in,height=2.1in]{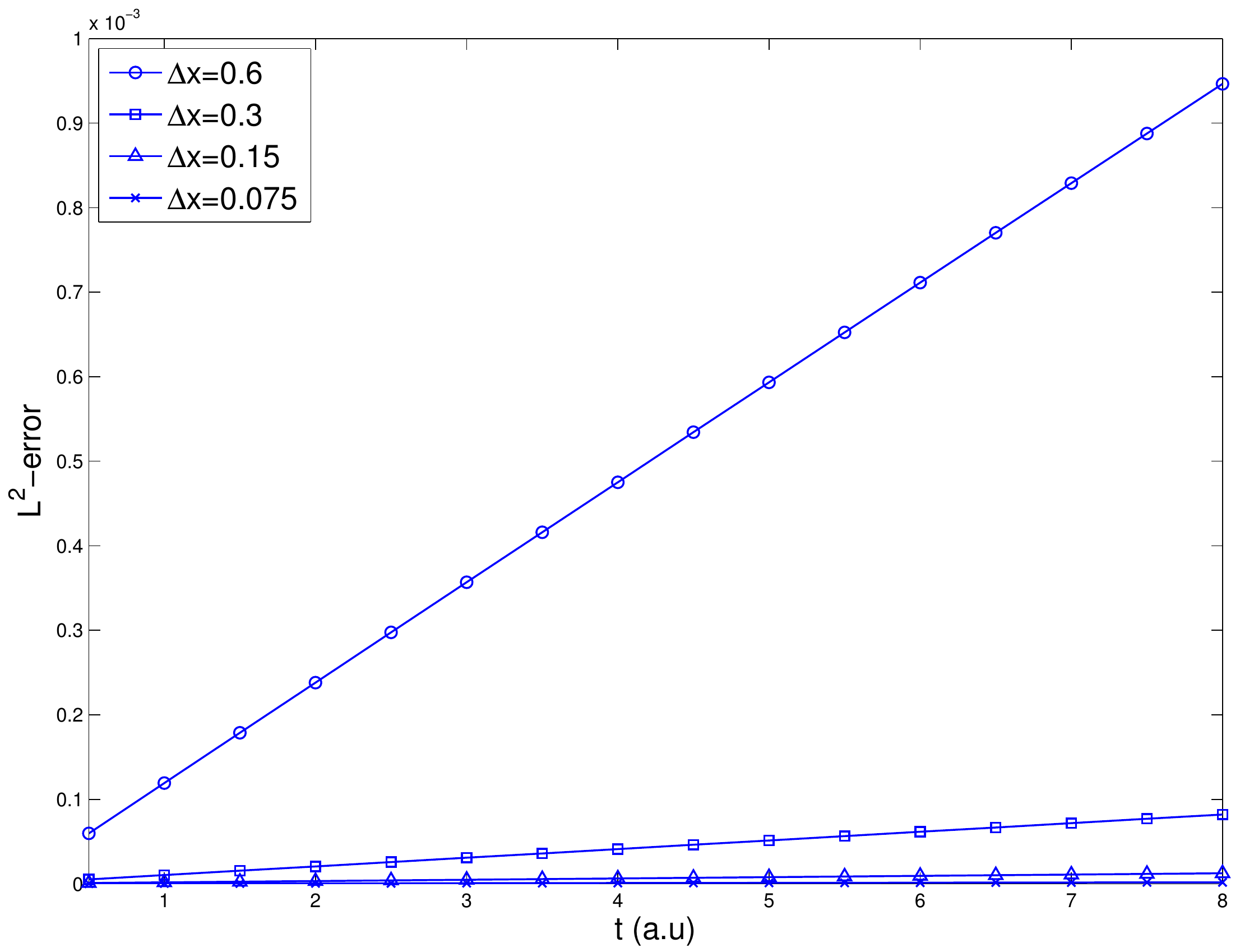}}
    \subfigure[Convergence order.]{
    \label{fig3:order}
    \includegraphics[width=2.9in,height=2.1in]{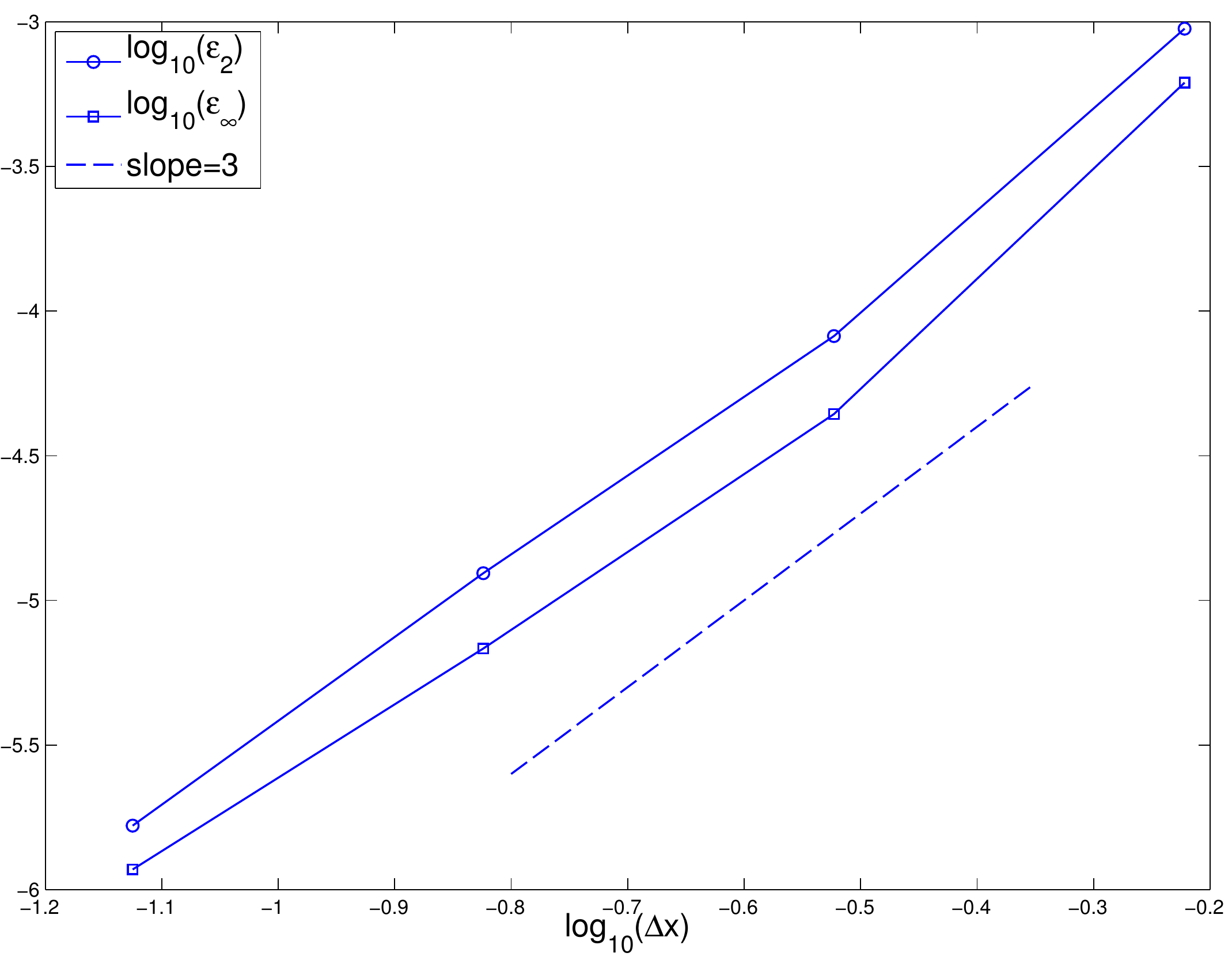}}
     \caption{\small The free advection of two fermions:
Numerical errors for different spatial spacing. }\label{fig_3}
\end{figure}

To verify the accuracy of the piecewise cubic spline interpolations \eqref{general_bicubic_sp}, the first experiment conducts the advection of two correlated fermions in the free space, i.e., the external potential $V(\bm{x})\equiv 0$.
In this situation, the Wigner equation \eqref{eq.Wigner}
has an analytical solution as follows
\begin{equation}
f\left(x_{1}, x_{2}, k_{1}, k_{2}, t\right)=f_{0}^{\text{fermion}}\left(x_{1}-{\hbar k_{1} t}/{m}, x_{2}-{\hbar k_{2} t}/{m}, k_{1}, k_{2}\right),
\end{equation}
the reduced Wigner functions of which are shown in Fig.~\ref{fig_4},
and Eqs.~\eqref{explicit_Euler}-\eqref{explicit_3_step} all reduce to the ``upwind" scheme.
It can be clearly shown there that the exchange-correlation hole
does exist around the central area during the evolution, but it becomes more and more narrow due to the dispersion
when two Gaussian waves move away.

We set $L_x = 15$, $L_k=5\pi/6$, $\Delta t=0.1$, $\Delta x=0.6,0.3,0.15,0.075$, and the end time $T=8$.
The $\bm{k}$-domain is divided into $4\times 4$ cells and each cell contained $24 \times 24$ Gauss-Chebyshev collocation points,
and the initial data is shown in Fig.~\ref{fig:fermion}.
Fig.~\ref{fig_3} plots
both $L^{2}$- and $L^{\infty}$-errors at the final time.
The error growth with time shows perfectly a linear dependence, for example, see the $L^{2}$-error history
in Fig.~\ref{fig3:linear}, and the error curves (in the logarithm scale) against the spatial spacing in Fig.~\ref{fig3:order} attains the theoretical convergence order of $3$ as we expected.

During the whole simulation, for above four spatial spacings, the errors for the physical symmetry relation  $\epsilon_{\textup{sym}}$ are always around the machine epsilon,
and the variations of total particle number $|\epsilon_{\textup{mass}}|$ are no more than $5.0990\times 10^{-6}$, both of which agree well with the theoretical results presented in Section \ref{sec:analysis}.
Actually, $\epsilon_{\textup{mass}}$ is still effected by the boundaries since Eq.~\eqref{eq:gint} holds trivially in the free space. As predicted by the numerical analysis in Section \ref{sec:analysis:mass},
enlarging the computational domain to cut down
the boundary effect will  further reduce $\epsilon_{\textup{mass}}$ to the machine epsilon,
for example, $|\epsilon_{\textup{mass}}|$ becomes no more than $1.9984\times 10^{-15}$ even for $\Delta x = 0.6$ when resetting $L_x=45$.

\subsection{Gaussian barrier scattering}
\label{sec:result:barrier}

\begin{figure}[h]
    \centering
    \subfigure[Errors vs. $\log_{10}\Delta x$ ($\Delta t=0.0125$).]{\label{fig_5_a}
    \includegraphics[width=2.9in,height=2.1in]{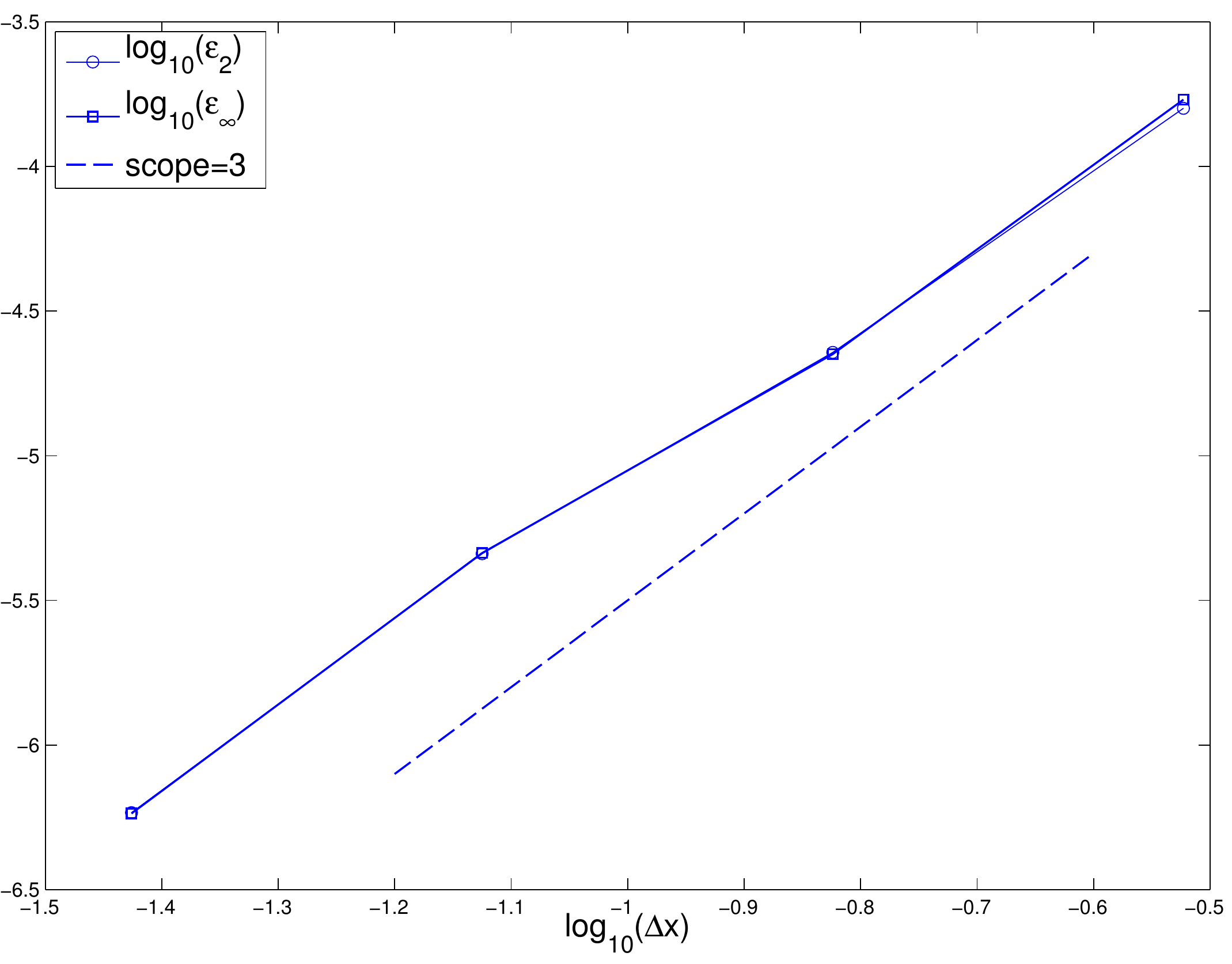}}
     \subfigure[Errors  vs. $\log_{10}\Delta t$ ($\Delta x=0.0375$).]{\label{fig_5_b}
    \includegraphics[width=2.9in,height=2.1in]{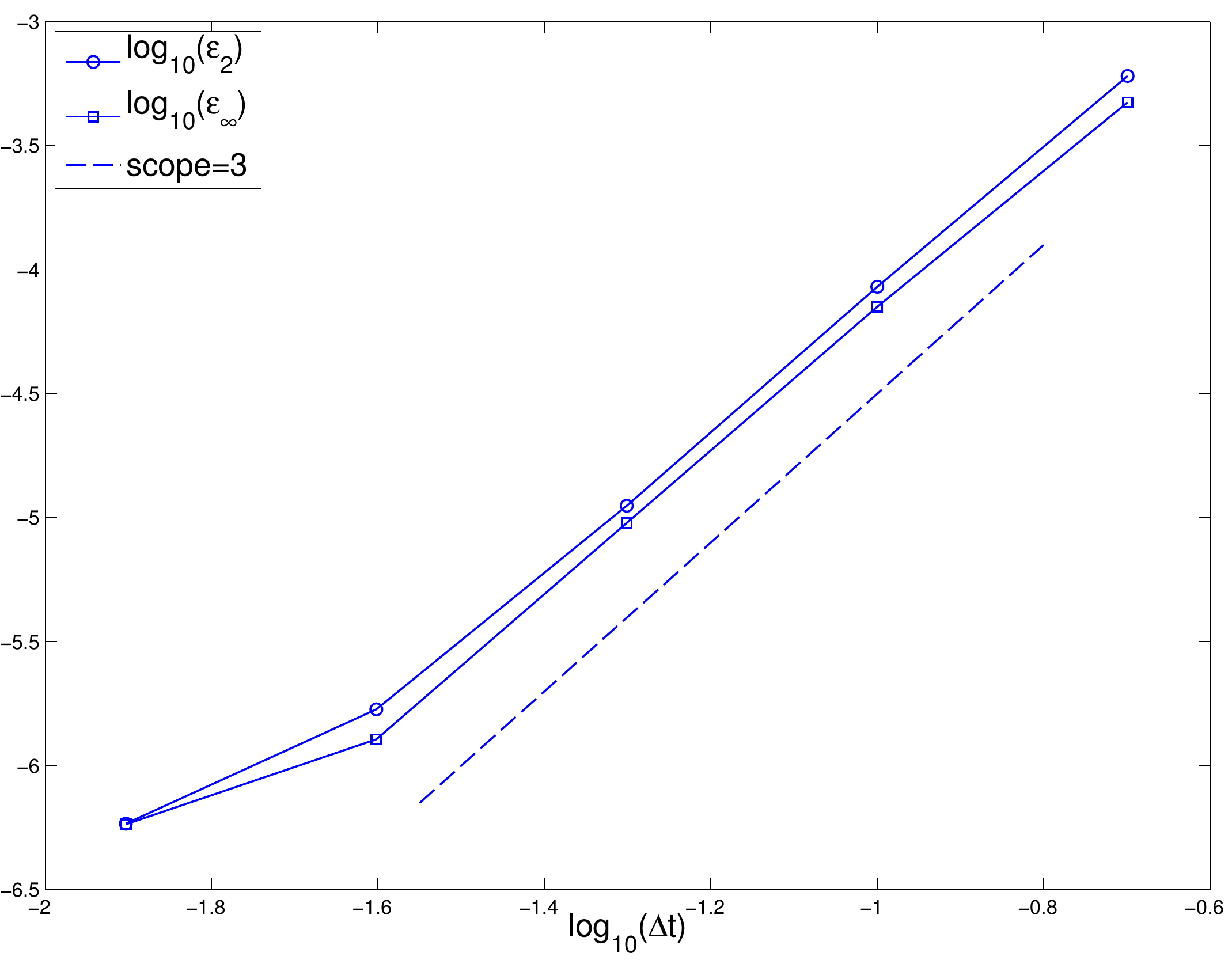}}
     \caption{\small The Gaussian barrier scattering for one particle: The convergence order with respect to the spatial spacing $\Delta x$ (fs) and the time step $\Delta t$ (nm).}
      \label{fig_5}
\end{figure}

To further validate the overall performance of the advective-spectral-mixed method, a comparison study with the cell average SEM for the one-body situation\cite{ShaoLuCai2011} is now conducted in simulating the
Gaussian barrier scattering for a Gaussian wave packet\cite{ShaoLuCai2011,ShaoSellier2015}.
The numerical solution calculated by the cell average SEM will be regarded as the reference solution in Eqs.~\eqref{eq:e2}-\eqref{eq:emass}.
A similar idea has also been recently used to study the accuracy of the signed particle MCM\cite{ShaoSellier2015},
and the same settings are adopted in this work.
Namely, the initial wave is given by Eq.~\eqref{1D_Gaussian_wp}
and the Gaussian barrier reads
\begin{equation}\label{eq:gpot1}
V\left(x\right)=H\exp\left[-\frac{\left(x-x_{B}\right)^{2}}{2\omega^{2}}\right].
\end{equation}
The parameters are: $\Omega=[0~\textup{nm}, 60~\textup{nm}] \times [-{5\pi}/{3}~\textup{nm}^{-1}, {5\pi}/{3}~\textup{nm}^{-1}]$, $x_{0}=-15~\textup{nm}$, $k_{0}=0.7~\textup{nm}^{-1}$, $a=2.825~\textup{nm}$, $H=0.3~\textup{eV}$, $x_{B}=-15~\textup{nm}$, $\omega=1~\textup{nm}$, the reduced Planck constant $\hbar=0.658211899~\textup{eV} \cdot \textup{fs}$, the effective mass $m=0.0665m_{e}$, the stationary electron mass $m_{e}=5.68562966~\textup{eV} \cdot \textup{fs}^{2} \cdot \textup{nm}^{-2}$, and the final time $T=20~\textup{fs}$.

In the comparison study,
the same $\bm{k}$- and $\bm{y}$- discretizations will be adopted, i.e.,
$\mathcal{K}$ is divided into $20$ elements and each element contains $30$ Gauss-Chebyshev points,
and $L_y=78~\textup{nm}$ which is determined by exploiting the sparse structure of $D_V(x,y)$ defined in Eq.~\eqref{Dv} (more details can be found in \cite{ShaoLuCai2011}).
For the SEM reference solution,
the $\bm{x}$-domain is divided into $10$ elements and each element contained $30$ Gauss-Lobatto points with the time step $\Delta t=0.002~\textup{fs}$.
The convergence of SEM has been thoroughly studied in \cite{ShaoLuCai2011} and the interaction dynamics is clearly shown in \cite{ShaoSellier2015}. Here we only focus on evaluating the convergence of the advective-spectral-mixed method against the grid spacing as well as the time stepping.
Table~\ref{Table_1} lists both $L^{\infty}$-error $\epsilon_{\infty}(t)$ and $L^{2}$-error $\epsilon_{2}(t)$ at the final time for different spatial spacing $\Delta x$ and time step $\Delta t$. In contrast to the strict CFL restriction in high-order Runge-Kutta time evolutions\cite{ShaoLuCai2011},
we can see there that larger time steps are now allowed,
for example, $\Delta t=0.2~\textup{fs}$ coupled with $\Delta x= 0.0375~\textup{nm}$ leads to errors no more than $6.0455\times 10^{-4}$.  Fig.~\ref{fig_5} further plots the errors with respect to $\Delta x$ and $\Delta t$ in logarithm scale.
We find there that the convergence order with respect to both spatial spacing and time step coincides very well with the theoretical prediction, i.e., the third-order accuracy, as mentioned in Section \ref{sec:analysis}.

\begin{table}[h]
 \centering
 \caption{\small The Gaussian barrier scattering for one particle: The $L^{\infty}$-error $\epsilon_{\infty}(t)$ and $L^{2}$-error $\epsilon_{2}(t)$ at $t=20~\textup{fs}$ for different spatial spacing $\Delta x$ (nm) and time step $\Delta t$ (fs). }
\label{Table_1}
 \begin{tabular}{cccc}
  \toprule
  \toprule
  $\Delta t$ & $\Delta x$ & $\epsilon_{\infty}(20)$ & $\epsilon_{2}(20)$  \\
  \midrule
  0.0125 & 0.0375 &  $5.8027\times 10^{-7}$ & $5.8365\times 10^{-7}$ \\
  0.0125 &  0.075 &  $4.6083\times 10^{-6}$ & $4.5847\times 10^{-6}$ \\
  0.0125 & 0.15 &  $2.2467\times 10^{-5}$ & $2.2741\times 10^{-5}$ \\
  0.0125 & 0.3 &  $1.7022\times 10^{-4}$ & $1.5876\times 10^{-4}$ \\
  0.025 & 0.0375 &  $1.2782\times 10^{-6}$ & $1.6885\times 10^{-6}$ \\
  0.025 &  0.075 &  $4.6083\times 10^{-6}$ & $4.5847\times 10^{-6}$ \\
  0.025 & 0.15 &  $2.2467\times 10^{-5}$ & $2.2741\times 10^{-5}$ \\
  0.025 &  0.3 &  $1.8849 \times 10^{-4}$ & $1.7686\times 10^{-4}$ \\
  0.05 & 0.0375 &  $9.5379\times 10^{-6}$ & $1.1187\times 10^{-5}$ \\
  0.05 &  0.075 &  $1.0055\times 10^{-5}$ & $1.1368\times 10^{-5}$ \\
  0.05 &  0.15 &  $3.7482\times 10^{-5}$ & $3.7395\times 10^{-5}$ \\
  0.05 & 0.3 &  $2.2980\times 10^{-4}$ & $2.2344\times 10^{-4}$ \\
  0.1  & 0.0375 &  $7.0873\times 10^{-5}$ & $8.5382\times 10^{-5}$ \\
  0.1 & 0.075 &  $7.1275\times 10^{-5}$ & $8.6693\times 10^{-5}$ \\
  0.1 & 0.15 &  $7.8537\times 10^{-5}$ & $1.0461\times 10^{-4}$ \\
  0.1   & 0.3 &  $2.9207\times 10^{-5}$ & $3.2452\times 10^{-4}$ \\
  0.2 & 0.0375 &  $4.7348\times 10^{-4}$ & $ 6.0455\times 10^{-4}$ \\
  0.2  & 0.075 &  $4.7370\times 10^{-4}$ & $6.0516\times 10^{-4}$ \\
  0.2 & 0.15 &  $4.7806\times 10^{-4}$ & $6.1537\times 10^{-4}$ \\
  0.2 & 0.3 &  $6.2856\times 10^{-4}$ & $7.7465\times 10^{-4}$ \\
  \bottomrule
  \bottomrule
  \end{tabular}
\end{table}

Such Gaussian barrier scattering can be readily extended to the two-body situation.
For instance,
we consider two uncorrelated Gaussian particles interacting with a Gaussian barrier,
\begin{equation}\label{eq:barrier}
V\left(x_{1}, x_{2}\right)=H_{1}\exp\left[-\frac{(x_{1}-x_{1\textup{B}})^{2}}{2}\right]+H_{2}\exp\left[-\frac{(x_{2}-x_{2\textup{B}})^{2}}{2}\right].
\end{equation}
Now we shift to the atomic units.
The initial Gaussian wave is $f_{0}(x_{1}, x_{2}, k_{1}, k_{2})=f_{1,0}^{\textup{1D}}\left(x_{1}, k_{1}\right)f_{2,0}^{\textup{1D}}\left(x_{2}, k_{2}\right)$ with $x_{1}^{0}=-12, x_{2}^{0}=-4, k_{1}^{0}=k_{2}^{0}=0.5, a_1=a_2=\sqrt{2}$.
That is, initially, those two wave packets have the same kinetic energy and moved independently at the same direction. The heights of two barriers are chosen as $H_{1}=0$ and $H_{2}=1$ with $x_{1\textup{B}}=0$ and $x_{2\textup{B}}=0$.
The averaged kinetic energy of each particle is about $E_{0}= (\hbar k^{0}_{i})^{2}/2m = 0.125$ and much lower than the barrier height. However, the barrier is set only to forbid the second particle to get through and has no influence on the first particle. Other parameters are set to be: $L_x=20$, $L_k={5\pi}/{6}$, $L_y=90$,
$\Delta t=0.05$, $\Delta x=0.2$, $T=15$.
The $\bm{k}$-domain is divided into $4 \times 4$ elements and each element contains $16 \times 16$ Gauss-Chebyshev collocation points.
The interaction dynamics is shown in Fig.~\ref{fig_6}.
We can observe there that, the second wave packet is almost completely reflected back, while the first one travels transparently through the barrier located at the central area.
This observation coincides exactly with our expectation and demonstrates clearly the accuracy of the method in some sense.

\begin{figure}
    \centering
    \subfigure[$t=0$.]{
    \includegraphics[width=2.9in,height=2.1in]{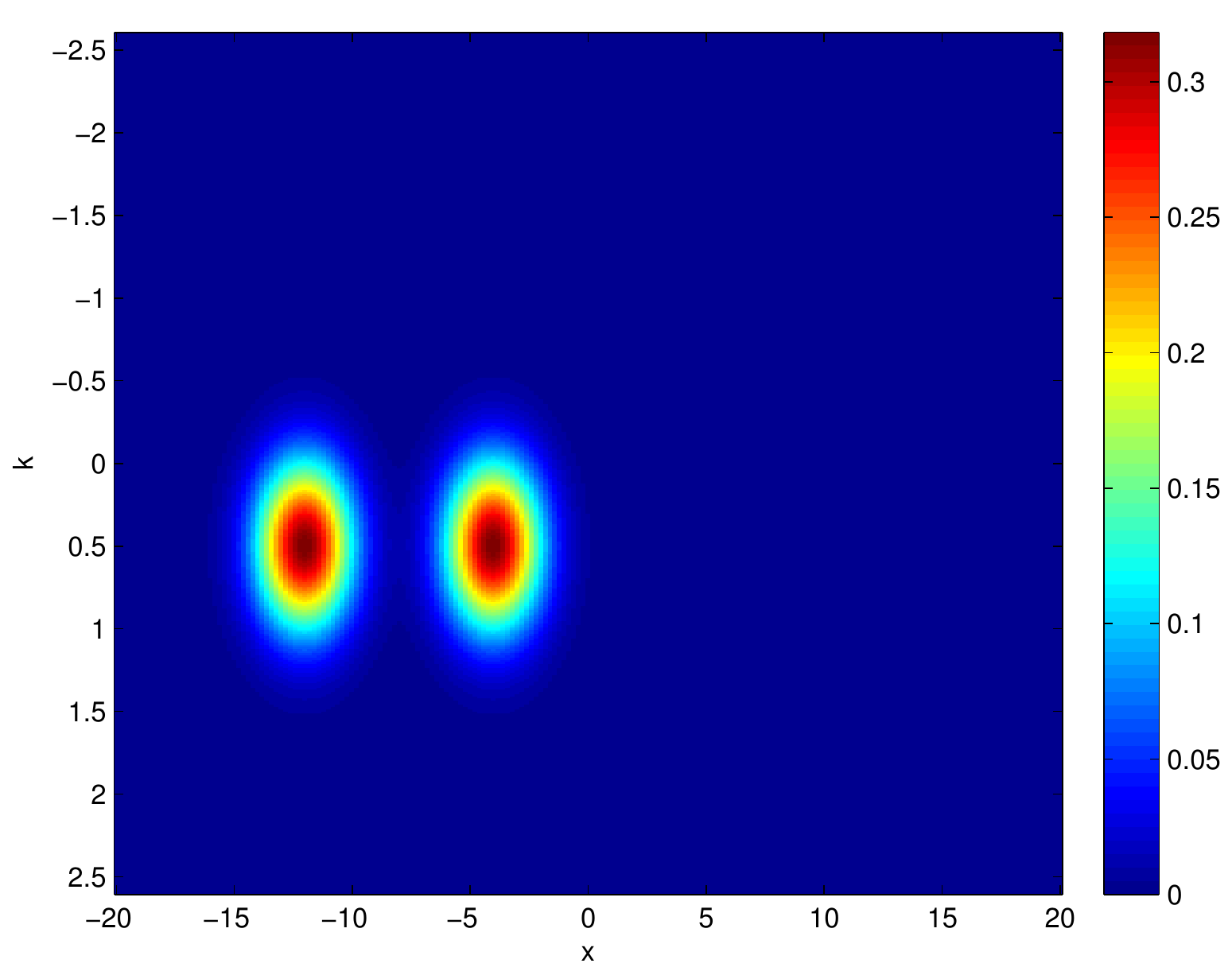}}
    \subfigure[$t=3$.]{
    \includegraphics[width=2.9in,height=2.1in]{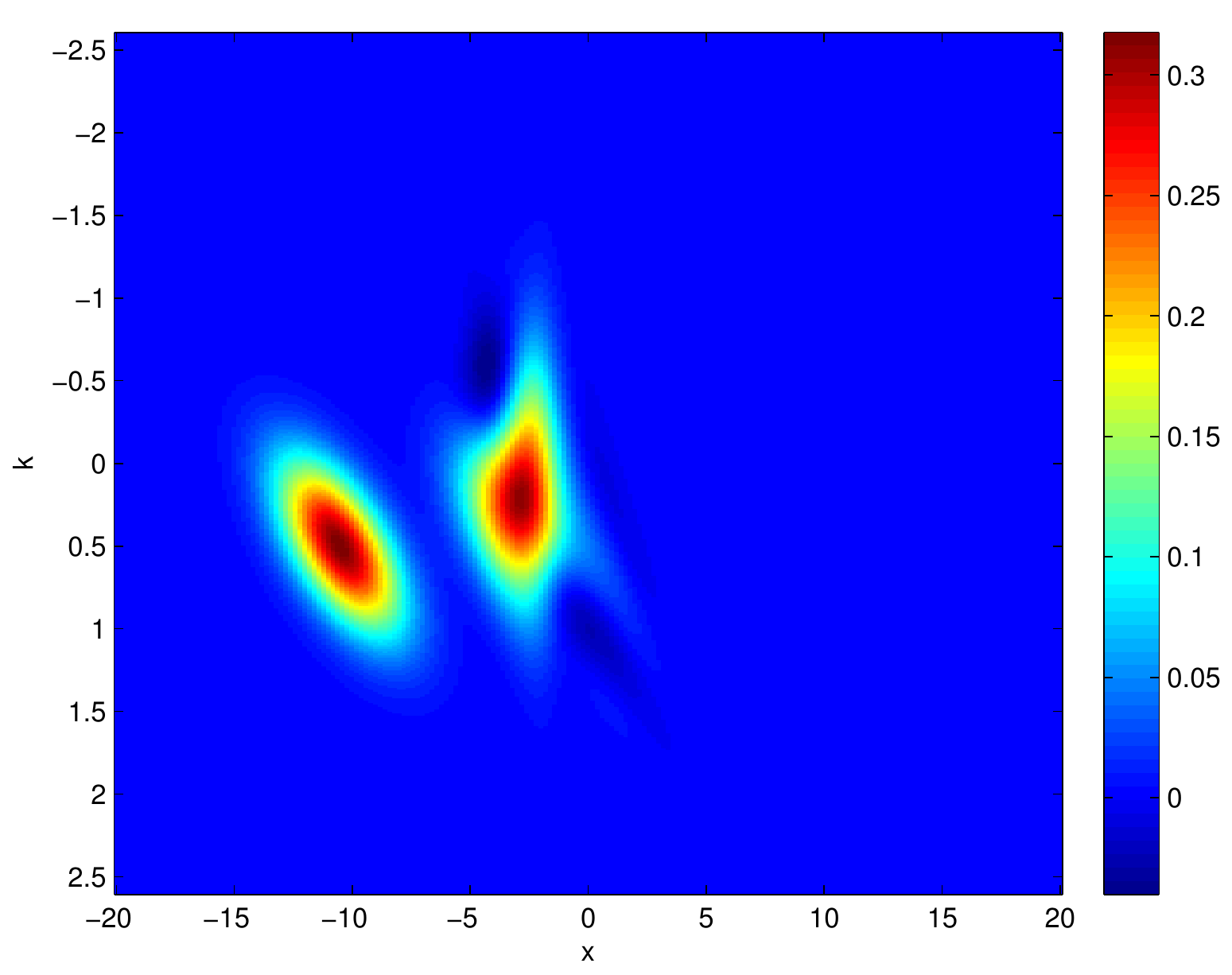}}
        \\
    \centering
    \subfigure[$t=6$.]{
    \includegraphics[width=2.9in,height=2.1in]{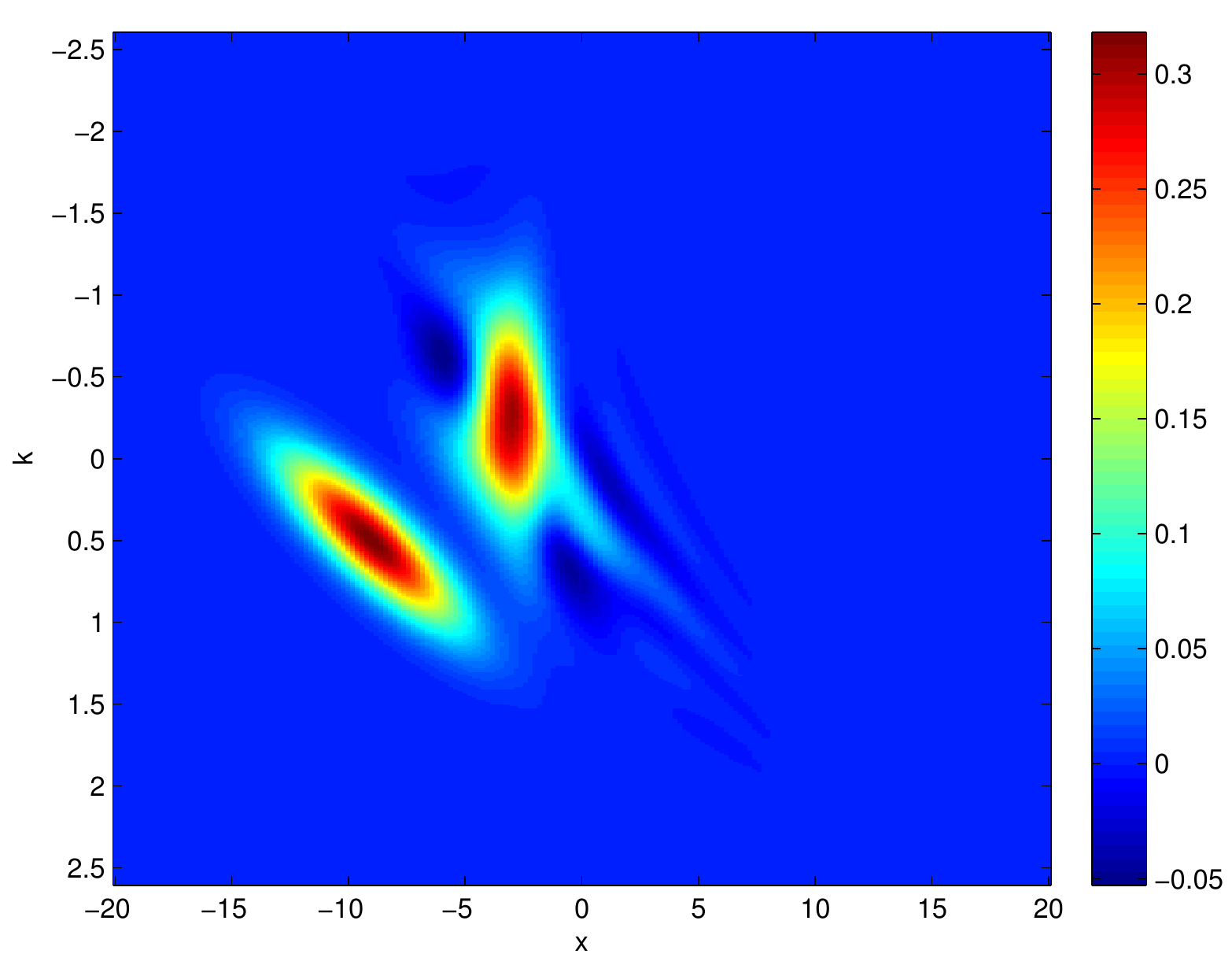}}
    \subfigure[$t=9$.]{
    \includegraphics[width=2.9in,height=2.1in]{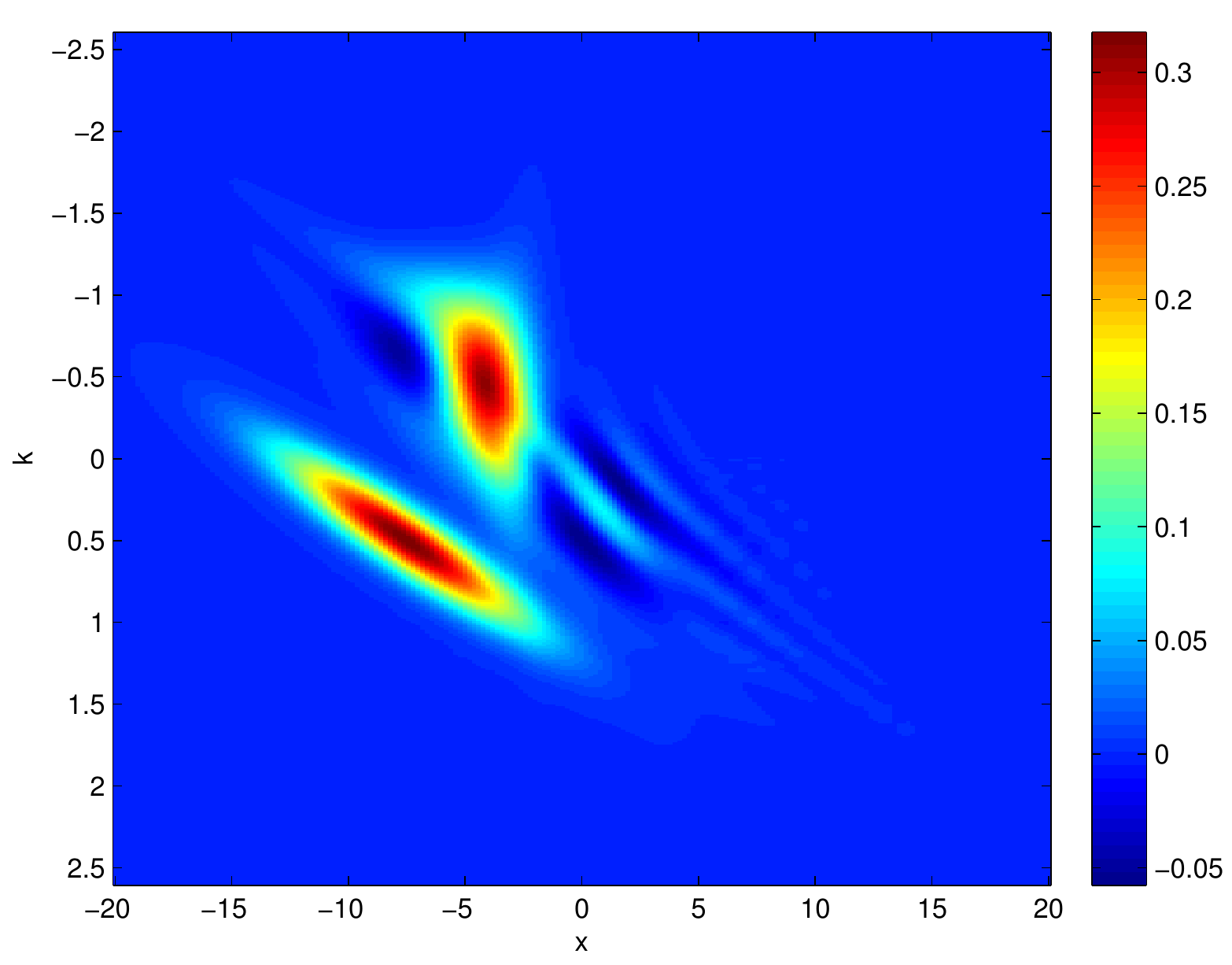}}
    \\
    \centering
    \subfigure[$t=12$.]{
    \includegraphics[width=2.9in,height=2.1in]{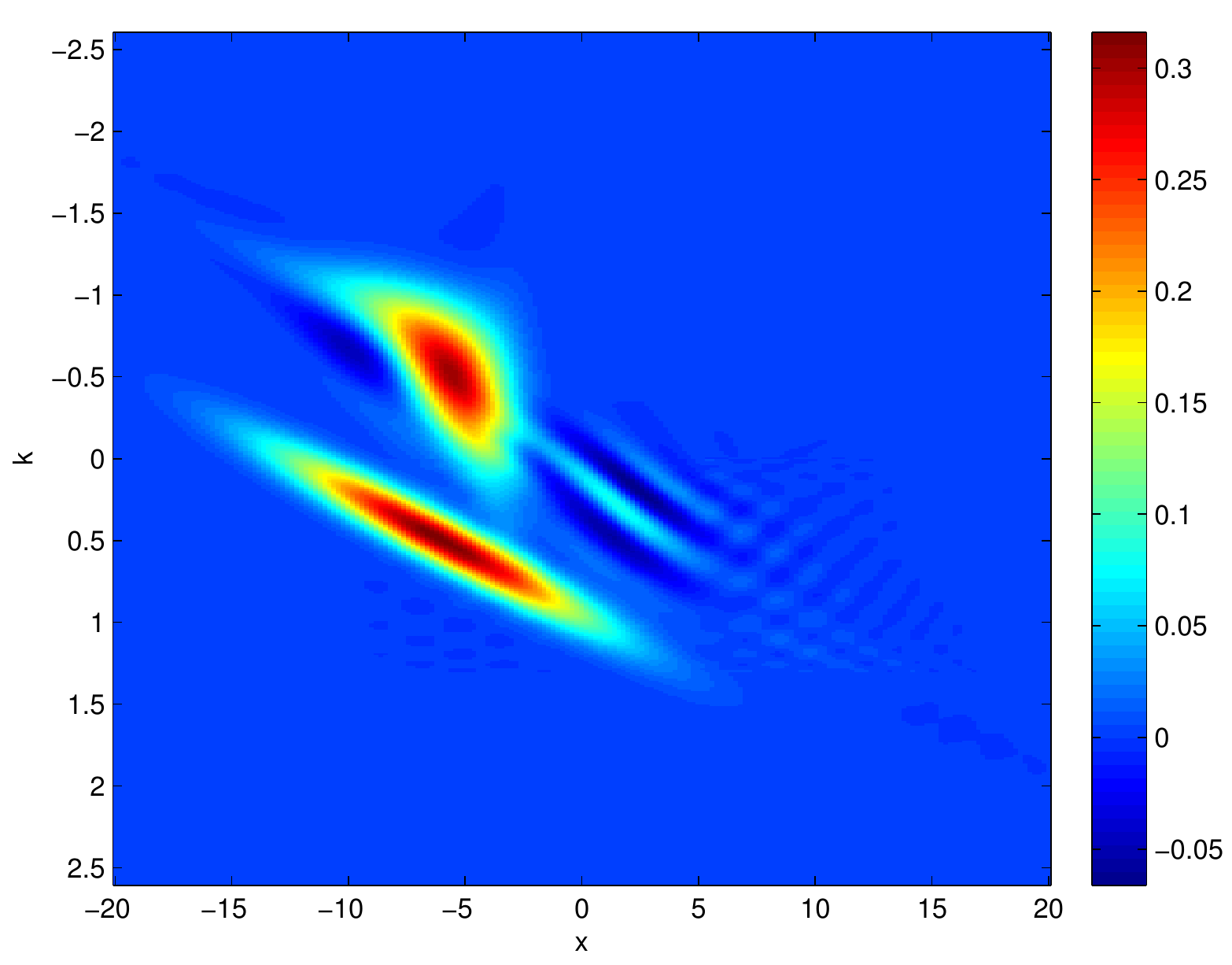}}
    \subfigure[$t=15$.]{
    \includegraphics[width=2.9in,height=2.1in]{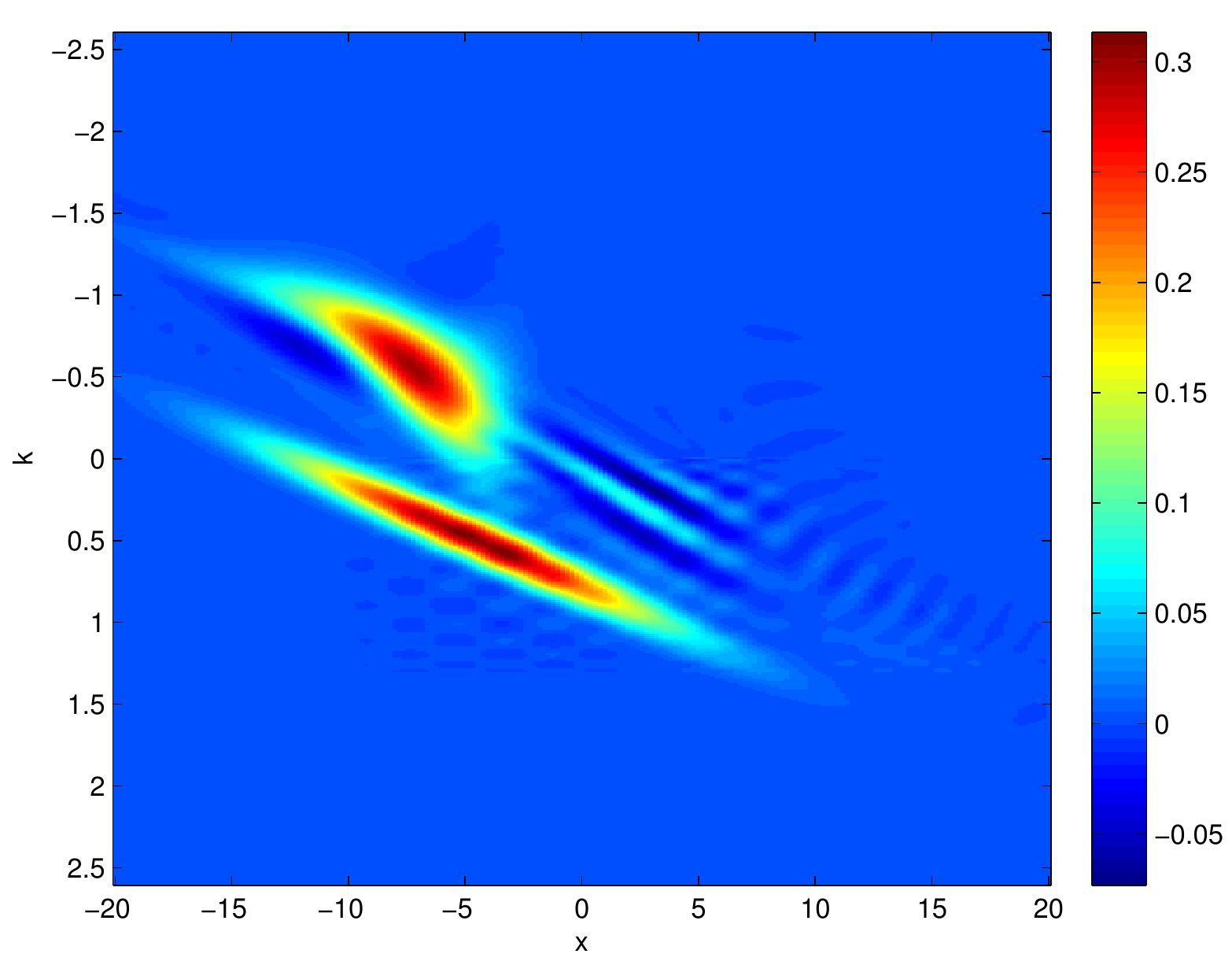}}
     \caption{\small The Gaussian barrier scattering for two uncorrelated particles: The reduced Wigner functions at different instants. Since the barrier is set to block only the second particle with the initial position $x_2^0 = -4$, i.e., $H_1=0$ and $H_2=1$ in Eq.~\eqref{eq:barrier},
the first particle initially located at $x_1^0 = -12$ shows the free advection, while the second one is completely reflected back. }\label{fig_6}
\end{figure}

Finally, let us see the performance in keeping the mass.
For above one-body case, during the time marching until the final time $T=20$ for $\Delta x=0.2$ and $\Delta t=0.05$, $\epsilon_{G}(t)$ is no more than $2.5778\times 10^{-9}$, but $\epsilon_{\textup{mass}}(20)$ is about $-8.6782\times 10^{-4}$ at the final time $T=20$. That is, the total outflow exceeds the total inflow due to the not-a-knot boundary conditions, and it can be improved by enlarging the computational domain. We redo the same simulation in an enlarged $\bm{x}$-domain $[-100~\textup{nm},200~\text{nm}]$ while leaving all other parameters unchanged, and find that
both $\epsilon_{\textup{mass}}$ and $\epsilon_{G}$ are on the same magnitude: $|\epsilon_{\textup{mass}}(t)|\leq 1.4015\times 10^{-10}$, $\epsilon_{G}(t)\leq 2.1739\times 10^{-10}$,
which agrees very well with the theoretical prediction in Section \ref{sec:analysis:mass}.
The same story also happens in the two-body situation.
For $L_x=20$, $\epsilon_{\textup{mass}}(15) = -4.9515\times 10^{-3}$ and $\epsilon_{G}(15) = 1.7389\times 10^{-7}$ at the end time $T=15$,
and they can be improved for $L_x=60$ to $\epsilon_{\textup{mass}}(15) = 1.0315\times 10^{-11}$
and $\epsilon_{G}(15) = 9.1085\times 10^{-12}$.

\subsection{Electron-electron scattering}
\label{sec:result:ee}

Now we turn to discuss a more challenging problem. Consider that two electrons are interacting through the repulsive Coulomb force. In this case, the two electrons are expected to decelerate, scatter and move away from each other. In general, the two-body interaction is given by the bare Coulomb potential
\begin{equation}
V_{\textup{ee}}\left(x_{1},x_2\right)=\frac{1}{\left|x_1-x_2\right|},
\end{equation}
which has a singularity at $x_{1}=x_{2}$. Thus we replace it with the soft-Coulomb potential\cite{LeinKreibichGross2002}
\begin{equation}\label{electron_potential}
V_{\textup{ee}}\left(x_{1}, x_{2}\right)=\frac{1}{\sqrt{\left|x_{1}-x_{2}\right|^{2}+\epsilon_{\textup{ee}}}},
\end{equation}
where $\epsilon_{\textup{ee}}$ is termed the soft parameter.

\begin{figure}[h]
    \centering
    \subfigure[Errors vs. $\log_{10}\Delta x$ ($\Delta t=0.0125$).]{\label{fig_7_a}
    \includegraphics[width=2.9in,height=2.1in]{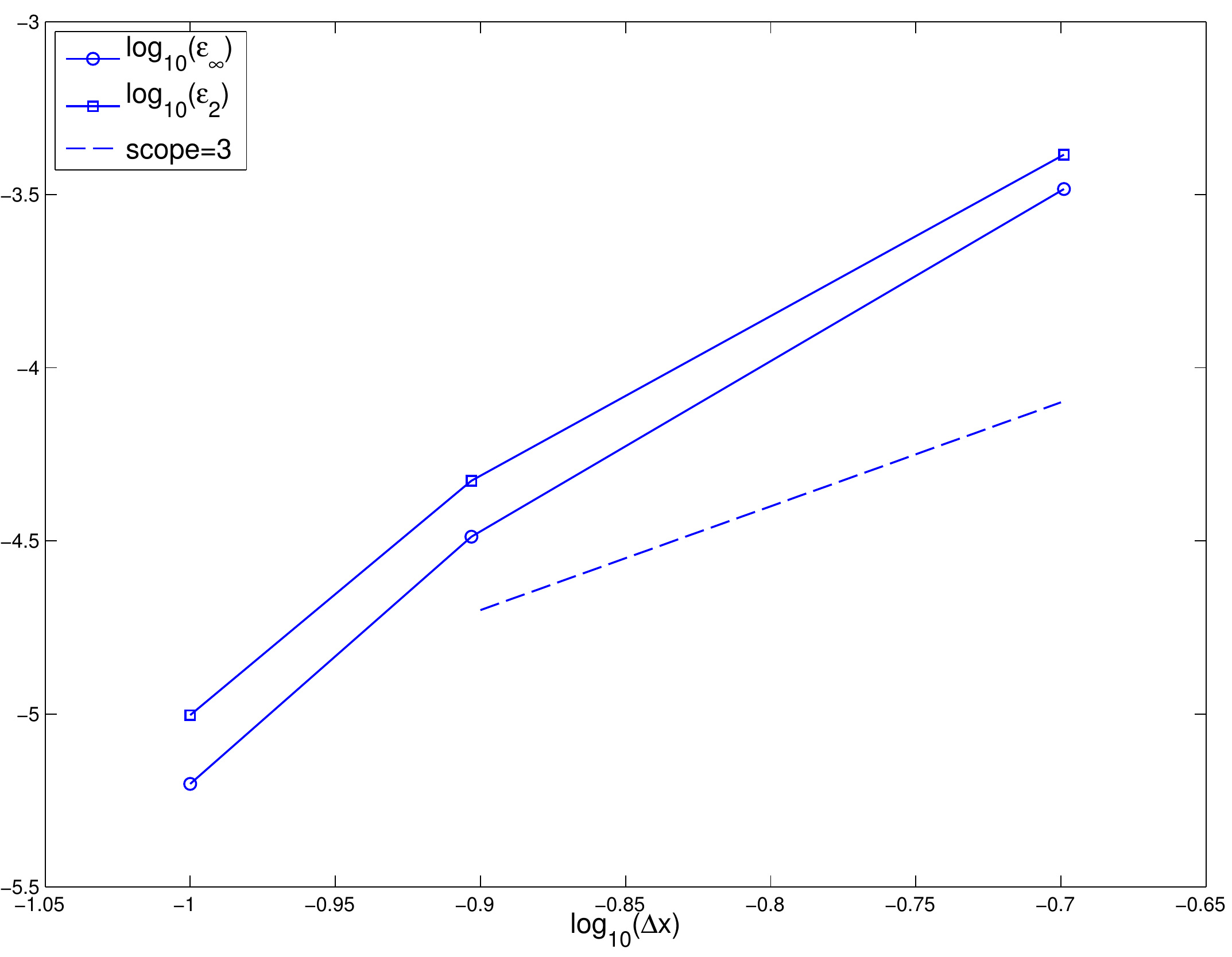}}
     \subfigure[Errors vs. $\log_{10}\Delta t$ ($\Delta x=0.125$).]{\label{fig_7_b}
    \includegraphics[width=2.9in,height=2.1in]{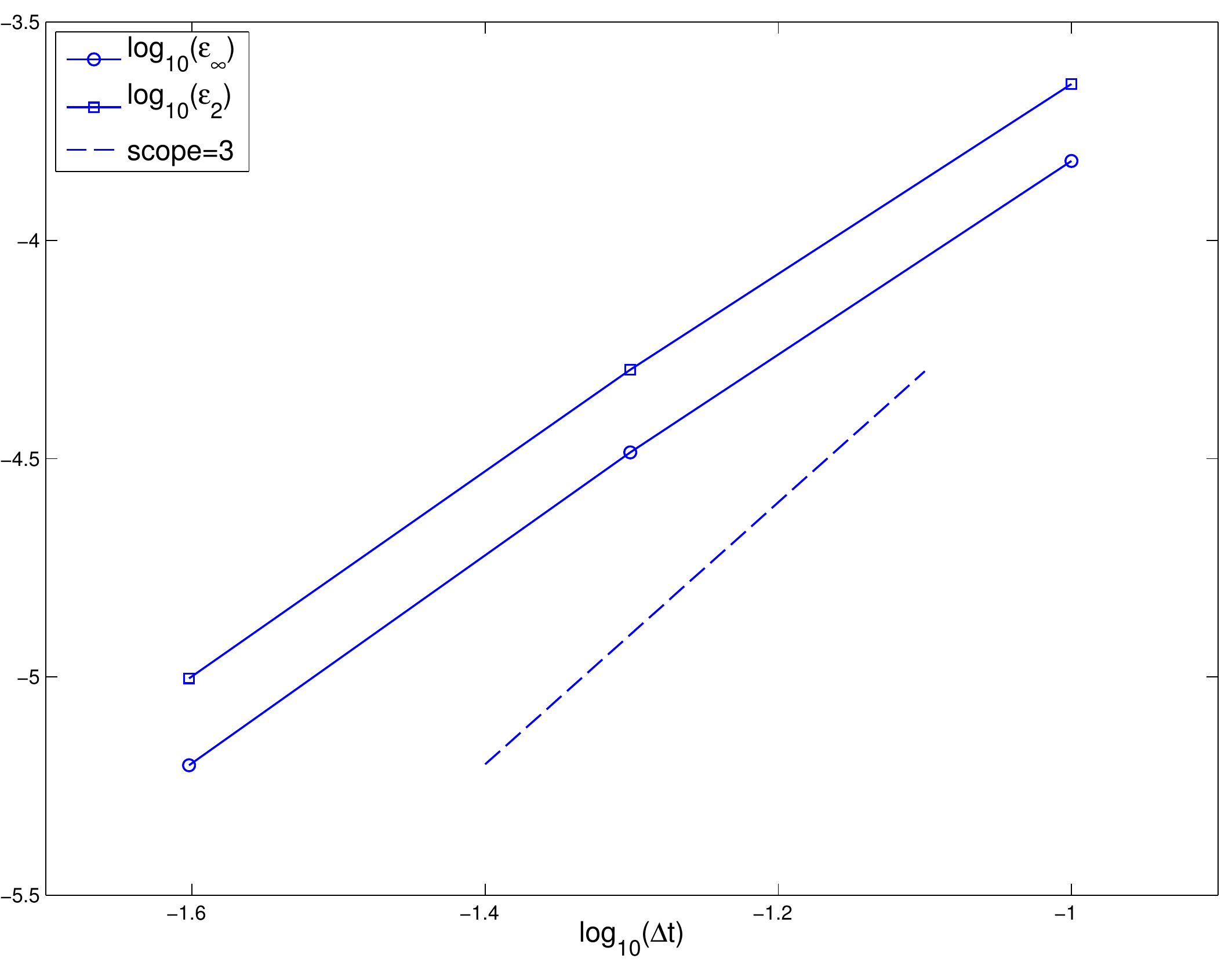}}
     \caption{\small Electron-electron scattering: The convergence order with respect to the spatial spacing $\Delta x$ and the time step $\Delta t$.}
      \label{fig_7}
\end{figure}

The first run is devoted to check the numerical convergence of the proposed method by setting
$L_x = 10$, $L_k = {5\pi}/{6}$, $L_y = 30$, the end time $T=4$,
and $\epsilon_{\textup{ee}} = 1$.
The $\bm{k}$-domain is divided into $4 \times 4$ elements and each element contains $16 \times 16$ Gauss-Chebyshev collocation points. The initial data is shown in Fig.~\ref{fig:fermion}, and the numerical solution obtained on a relatively fine mesh with $\Delta t=0.0125$ and $\Delta x=0.125$ is chosen to be the reference. Table \ref{tab:ee} presents
both $L^{\infty}$- and $L^{2}$-errors at the final time
and Fig.~\ref{fig_7} plots the convergence order with respect to the spatial spacing $\Delta x$ and the time step $\Delta t$. It is easily observed there that the measured convergence rate is around the theoretical value of $3$.
The slight deviation may come from that fact that the reference solution is not really a analytical one.
However, even this reference solution takes almost eight hours with 16 threads parallel running on our computing platform: Dell Poweredge R820 with $4\times$ Intel Xeon processor E5-4620
(2.2 GHz, 16 MB Cache, 7.2 GT/s QPI Speed, 8 Cores, 16 Threads) and 256GB memory.

As we have pointed out in the end of Section \ref{sec:method},
there is no general way to determine $L_y$ except for potentials of compact supports or of exponential decays such as the Gaussian barriers in Eqs.~\eqref{eq:gpot1} and \eqref{eq:barrier}.
Here we propose another simple way to determine $L_y$
roughly and initially by exploiting the exponential decays of the Gaussian Wigner function \eqref{asym_Wigner_init} (see Fig.~\ref{fig:fermion}).
Let
\begin{equation}
\epsilon_{g}(L_y)
=
\max_{(\bm{x},\bm{k})\in\Omega}\left\{\left|g^{T}\left(\bm{x},\bm{k}, 0;L_y\right)-g^{T}\left(\bm{x},\bm{k}, 0;L_y^{\textup{ref}}\right)\right|\right\},
\label{eq:eg}
\end{equation}
where $g^T(\bm{x},\bm{k},t; L_y)$
denotes the numerical approximation of $g(\bm{x},\bm{k},t)$ defined in Eq.~\eqref{eq:g}, which is obtained by
truncating the infinite series \eqref{Poisson_summation_truncated} with the domain given in \eqref{eq:ydom},
and $L_y^{\textup{ref}}$ is the reference length of the $\bm{y}$-domain and usually takes a large value.
Fig.~\ref{fig_Ly} displays above $\epsilon_{g}(L_y)$ for the initial data presented in Fig.~\ref{fig:fermion} and the soft-Coulomb potential with $\epsilon_{\textup{ee}} = 1$, where
we have set $L_y^{\textup{ref}}=240$. It can be easily observed there that $\epsilon_{g}$ is around $10^{-8}$ for $L_y=30$,
which has been used in the convergence test. We will adopt $L_y = 60$ below for longer simulations.
Fig.~\ref{fig_8} shows the reduced Wigner functions until the final time $T=6$ with $\Delta x=0.125$ and $\Delta t=0.05$.
By comparing with the free advection displayed in Fig.~\eqref{fig_4},
we find that,  before $t=4$,
the reduced Wigner function for two electrons, moving initially towards each other, is suppressed in the region $|k|\leq 1$  and possesses a wider expansion in $x$-direction because of the Coulomb deceleration; after that, two electrons tend to scatter out due to the repulsive interaction as well as the dispersion.
It must be noted that the Fermi hole between two electrons exists all the time since they are strongly correlated.
During the interaction dynamics, $\epsilon_{\textup{sym}}$ in Eq.~\eqref{eq:esym} is always around
the machine resolution for double precision,
and the variations of mass at the end time reads: $\epsilon_{\textup{mass}}(6) = -2.7403\times 10^{-3}$ with $\epsilon_{G}(6) = 4.0420\times 10^{-9}$.
When redoing the same simulation for $L_x=45$,
the variations of mass can be drastically reduced to
$\epsilon_{\textup{mass}}(6) = -2.5435\times 10^{-13}$ with $\epsilon_{G}(6) = 2.2190\times 10^{-12}$.

We have also tried a smaller soft parameter, say $\epsilon_{\textup{ee}} = 0.01$, implying a stronger repulsive interaction between two fermions. The interaction dynamics is very similar to those shown in Fig.~\ref{fig_8} corresponding to $\epsilon_{\textup{ee}} = 1$ and thus skipped. The possible reason may be,
the fermions feel the stronger repulsion only when they get close enough, while the long-range interaction between them is just slightly affected. However,
a relatively smaller time step, for example $\Delta t=0.0125$,
must be adopted instead, otherwise the numerical instability may happen. This is possibly related to the stiff gradient of $V(\bm{x})$, because the high-order derivatives of $V(\bm{x})$  may have a significant influence on the quantum dynamics
in view of the Moyal expansion \eqref{Moyal_expansion}.
That is, the time step may be still influenced a little bit by the deformational Courant number $\|\Delta t\cdot \nabla_{\bm{x}} V \|\leq 1$ as already shown for the Vlasov simulations\cite{SonnendruckerRocheBertrand1999},
though it is not restricted by the usual CFL condition.

\begin{table}
 \centering
 \caption{\small Electron-electron scattering:
The $L^{\infty}$-error $\epsilon_{\infty}(t)$ and $L^{2}$-error $\epsilon_{2}(t)$ at $t=4$ for different spatial spacing $\Delta x$ and time step $\Delta t$.
The numerical solution calculated from the finest mesh with $\Delta t=0.0125$ and $\Delta x=0.125$ is regarded as the reference.
}
\label{tab:ee}
 \begin{tabular}{ccccc}
  \toprule
  \toprule
  $\Delta t $ & $\Delta x$ & $\epsilon_{\infty}(4)$ & $\epsilon_{2}(4)$  \\
  \midrule
  0.0125 & 0.125 &  - & - \\
  0.0125 & 0.25 &  $2.5563\times 10^{-5}$ & $4.2093\times 10^{-5}$ \\
  0.0125 & 0.5 &  $3.9108\times 10^{-4}$ & $5.8120\times 10^{-4}$ \\
  0.025 & 0.125 &  $6.2753\times 10^{-6}$ & $9.9228\times 10^{-6}$ \\
  0.025 & 0.25 &  $3.2506\times 10^{-5}$ & $4.7147\times 10^{-5}$ \\
  0.025 & 0.5 &  $3.2845\times 10^{-4}$ & $4.1265\times 10^{-4}$ \\
  0.05   & 0.125 &  $3.2668\times 10^{-5}$ & $5.0519\times 10^{-5}$ \\
  0.05 & 0.25 &  $4.7014\times 10^{-5}$ & $6.5645\times 10^{-5}$ \\
  0.05 & 0.5 &  $3.4393\times 10^{-4}$ & $4.2734\times 10^{-4}$ \\
  0.1 & 0.125 &  $1.5192\times 10^{-4}$ & $2.2787\times 10^{-4}$ \\
  0.1 & 0.25 &  $1.4283\times 10^{-4}$ & $2.1455\times 10^{-4}$ \\
  0.1 & 0.5 &  $4.1791\times 10^{-4}$ & $ 5.0209\times 10^{-4}$ \\
  \bottomrule
  \bottomrule
  \end{tabular}
\end{table}

\begin{figure}
\centering
\includegraphics[width=2.9in,height=2.1in]{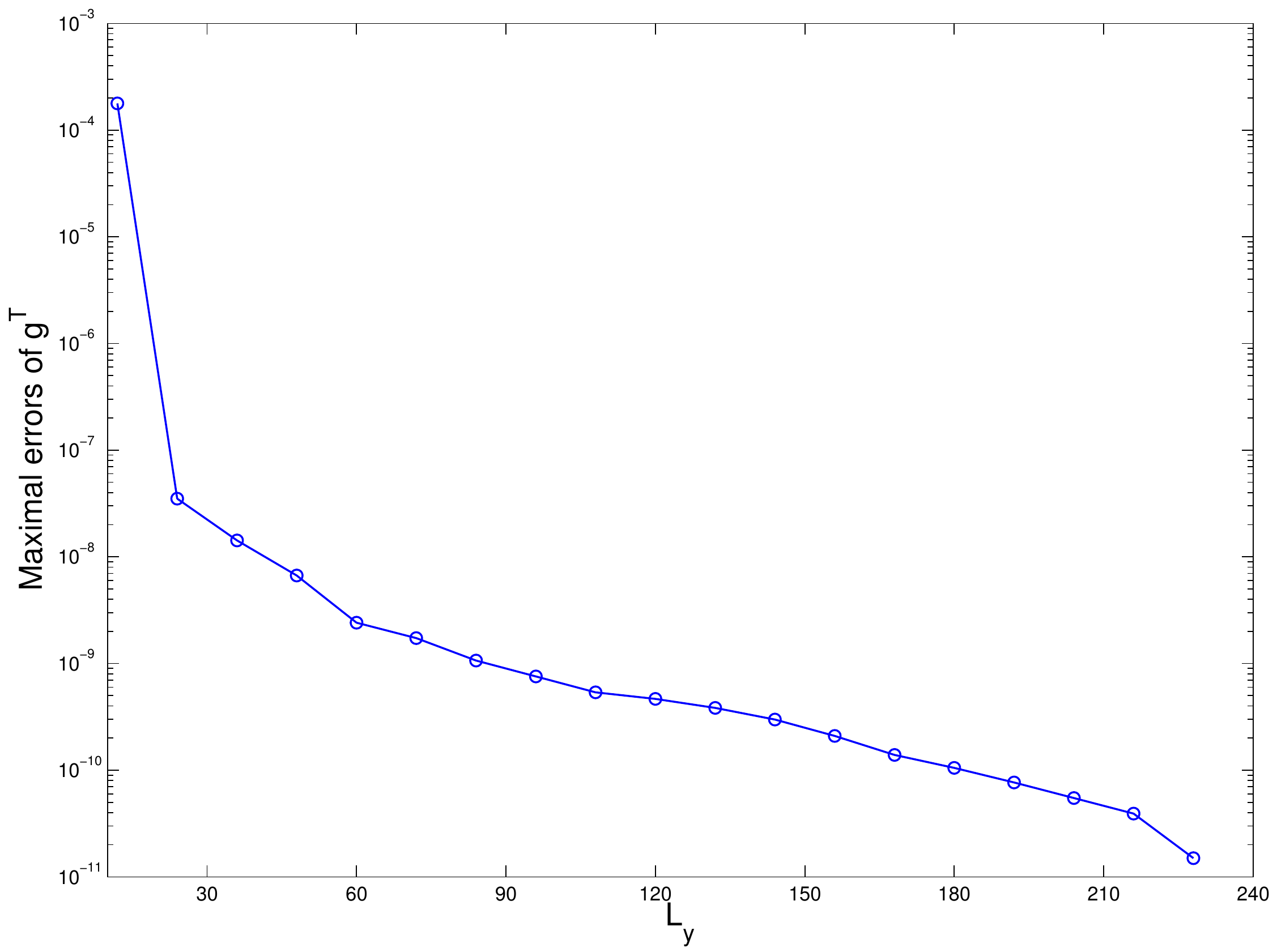}
     \caption{\small Electron-electron scattering: Maximum errors of $g^T(\bm{x},\bm{k},0;L_y)$ (see Eq.~\eqref{eq:eg}) against the truncation length $L_y$ in $\bm{y}$-space. The reference length is set to be $L_y^{\textup{ref}}=240$.
}\label{fig_Ly}
\end{figure}

\begin{figure}
    \centering
    \subfigure[$t=1$.]{
    \includegraphics[width=2.9in,height=2.1in]{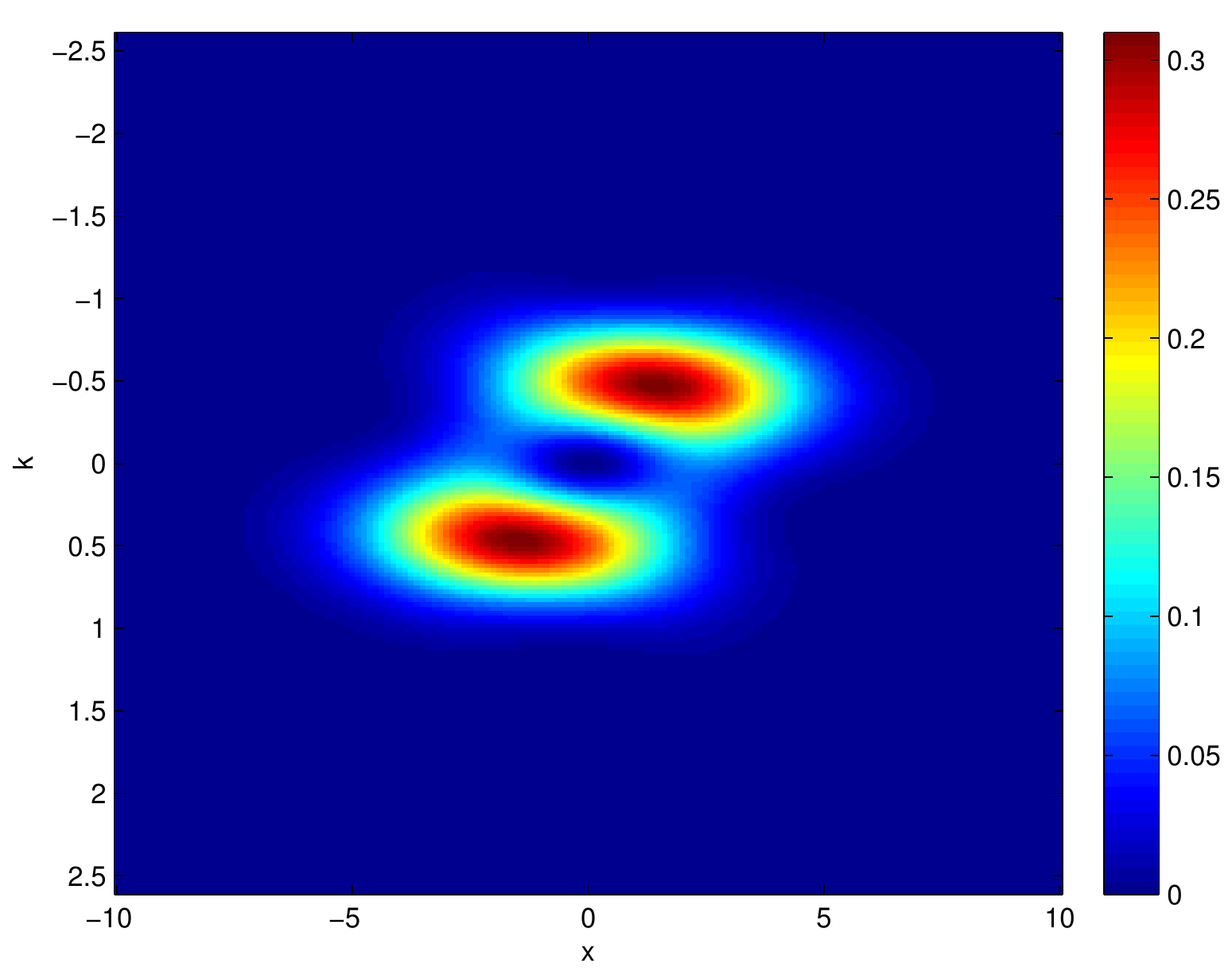}}
    \subfigure[$t=2$.]{
    \includegraphics[width=2.9in,height=2.1in]{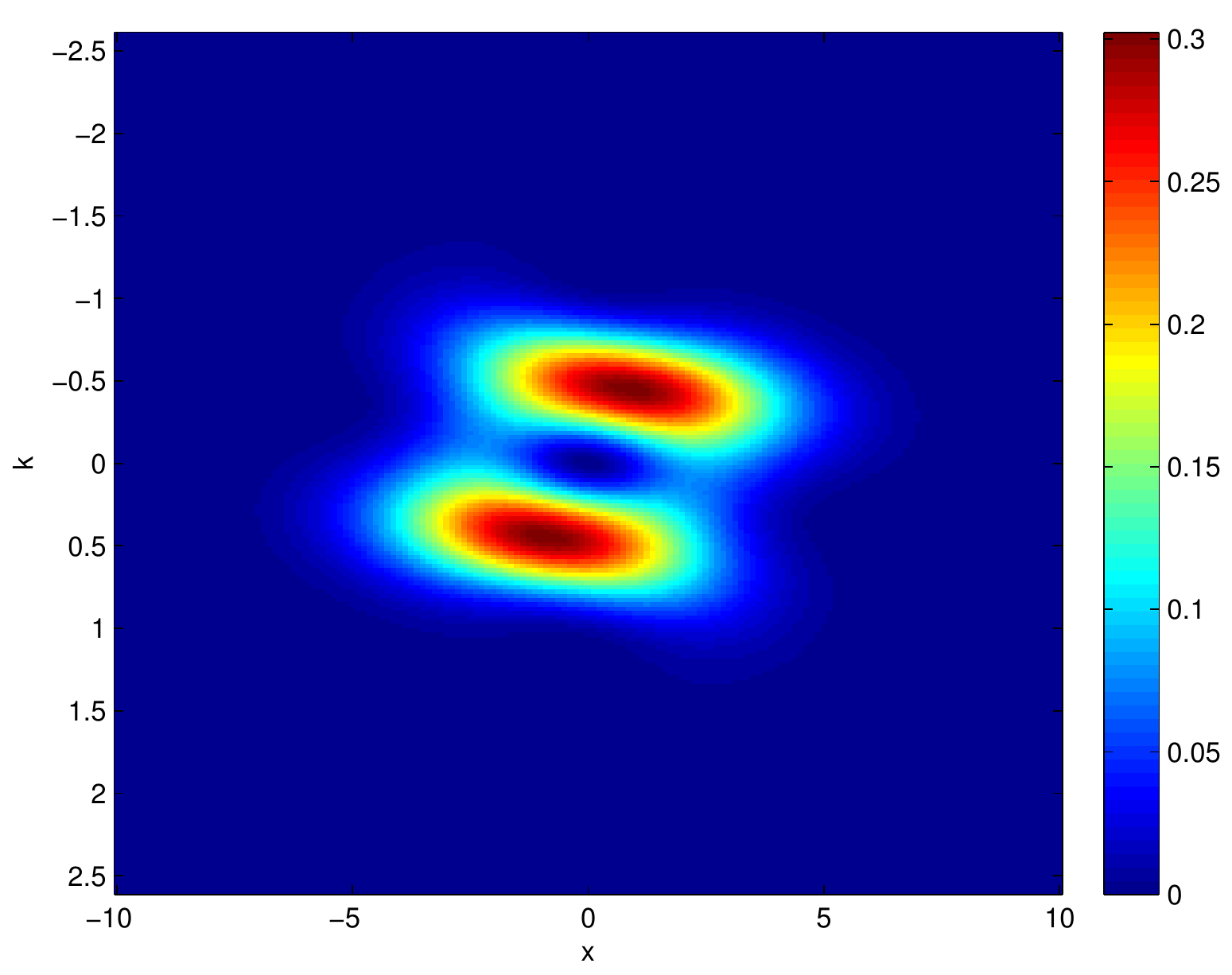}}
        \\
    \centering
    \subfigure[$t=3$.]{
    \includegraphics[width=2.9in,height=2.1in]{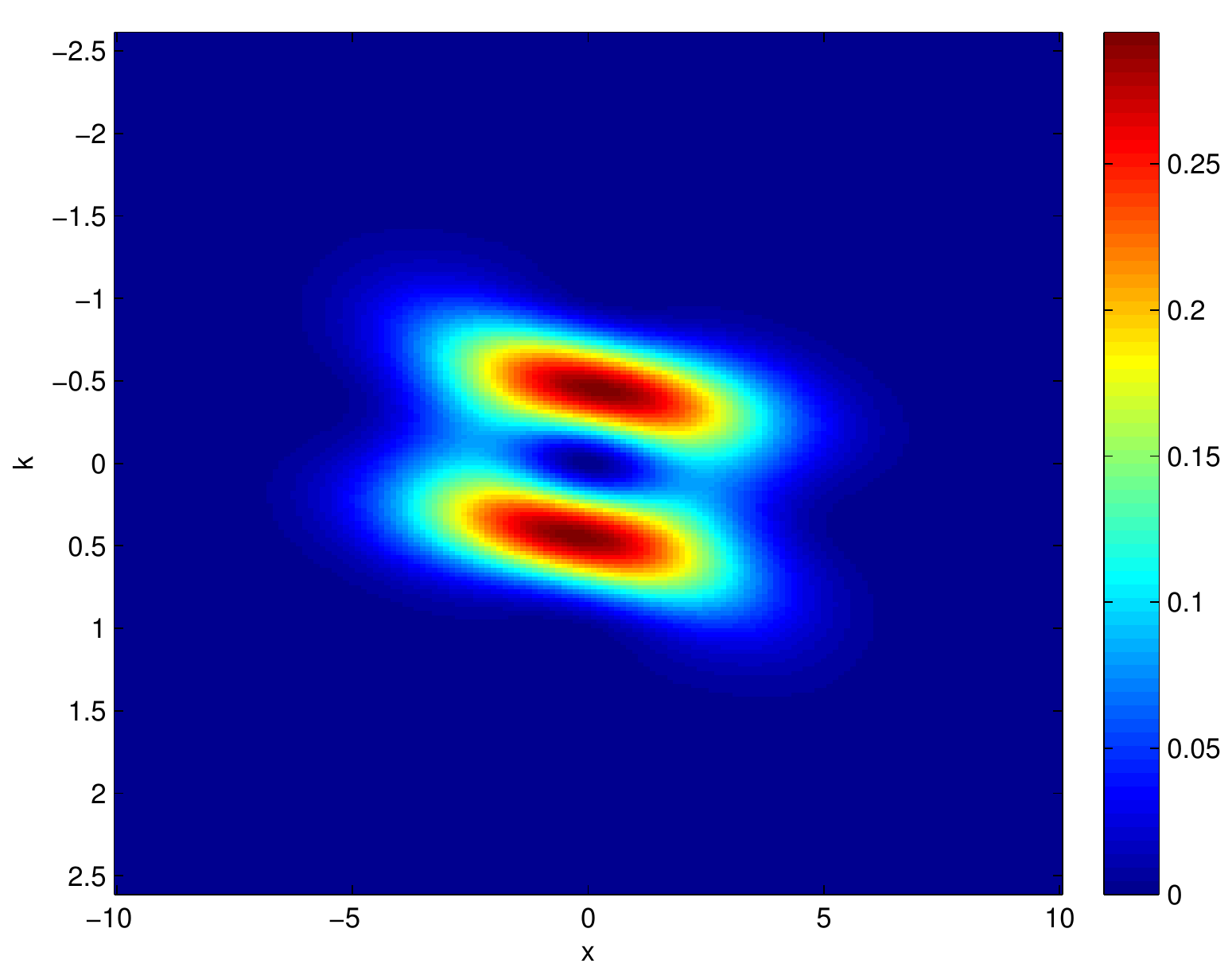}}
    \subfigure[$t=4$.]{
    \includegraphics[width=2.9in,height=2.1in]{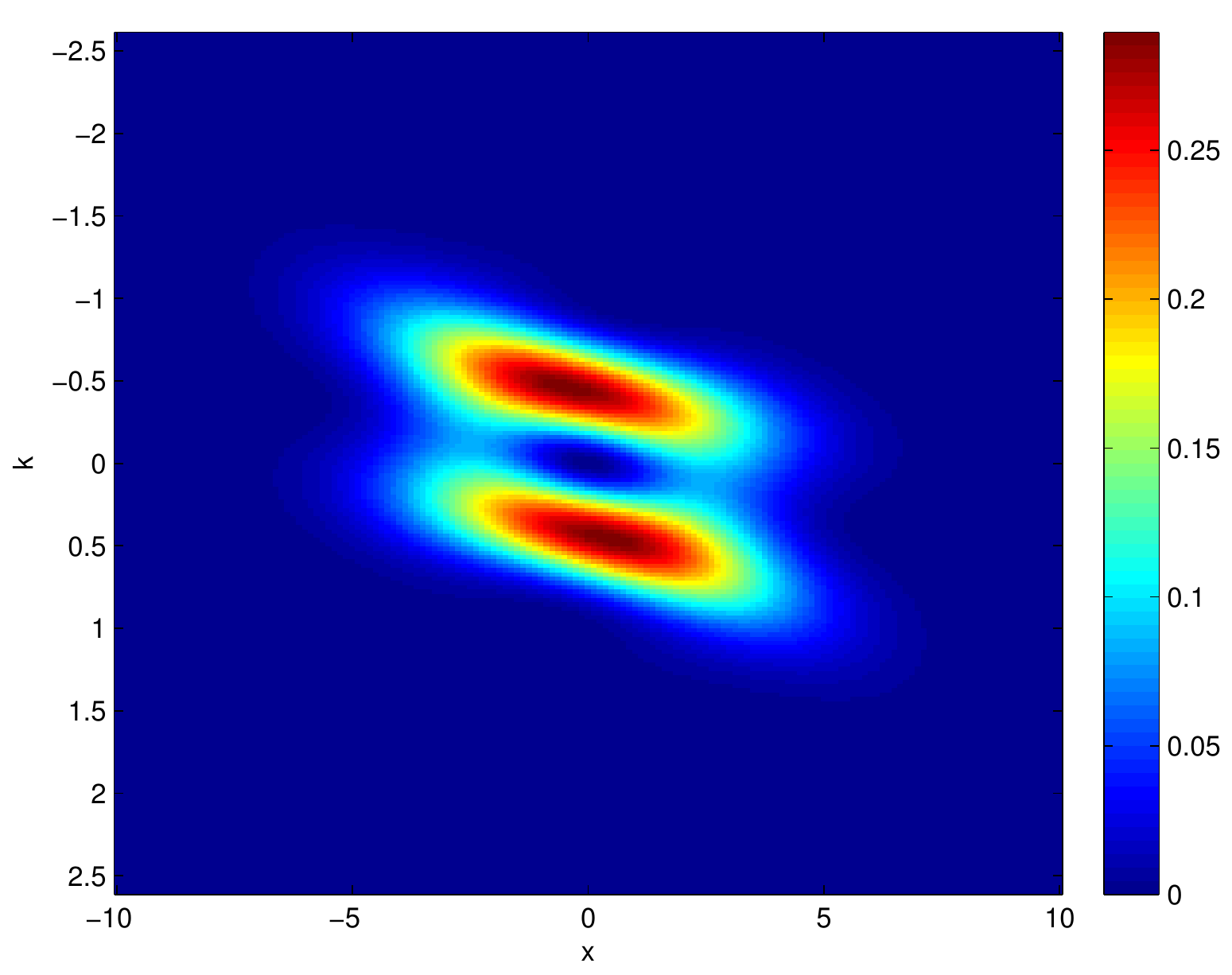}}
    \\
    \centering
    \subfigure[$t=5$.]{
    \includegraphics[width=2.9in,height=2.1in]{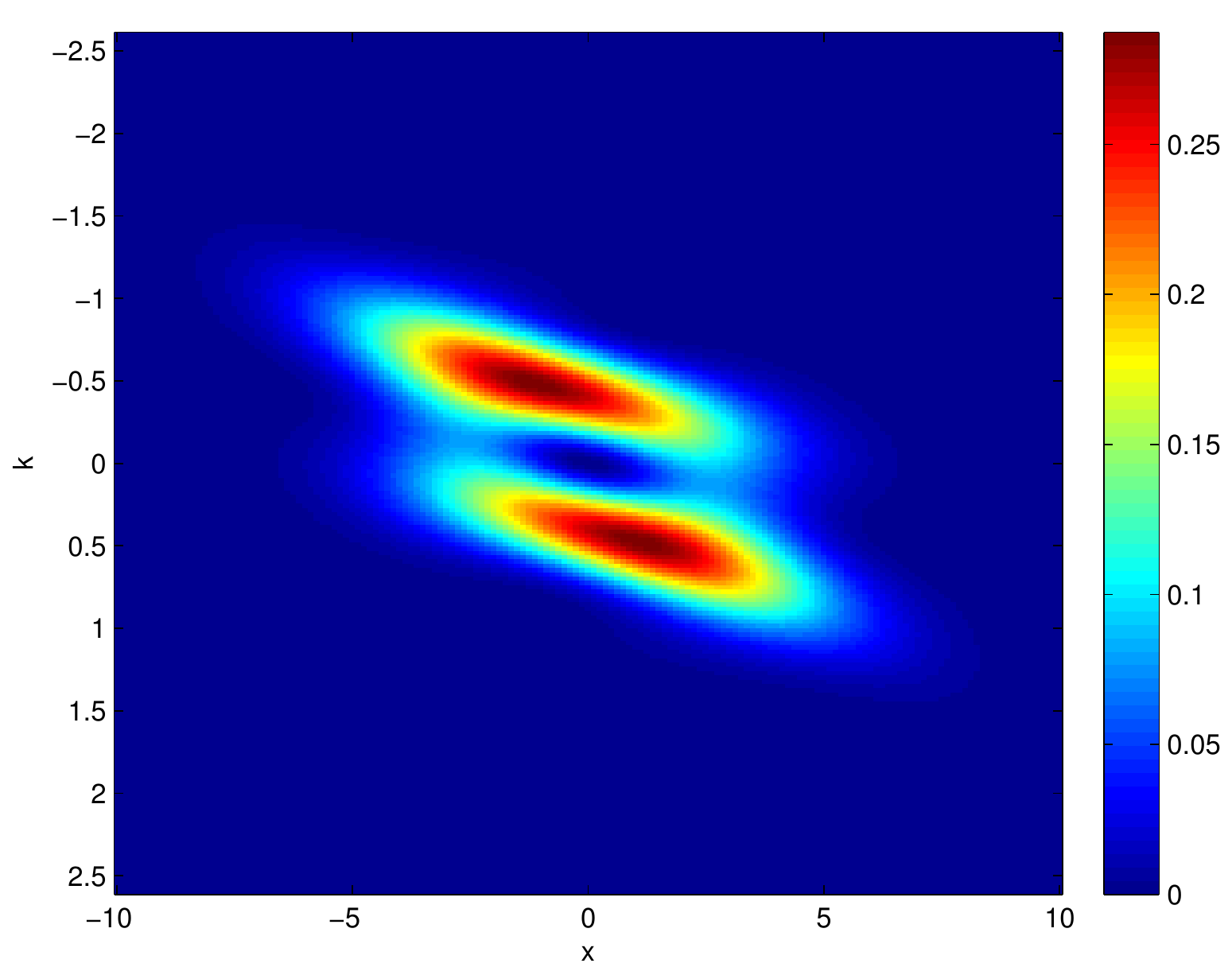}}
    \subfigure[$t=6$.]{
    \includegraphics[width=2.9in,height=2.1in]{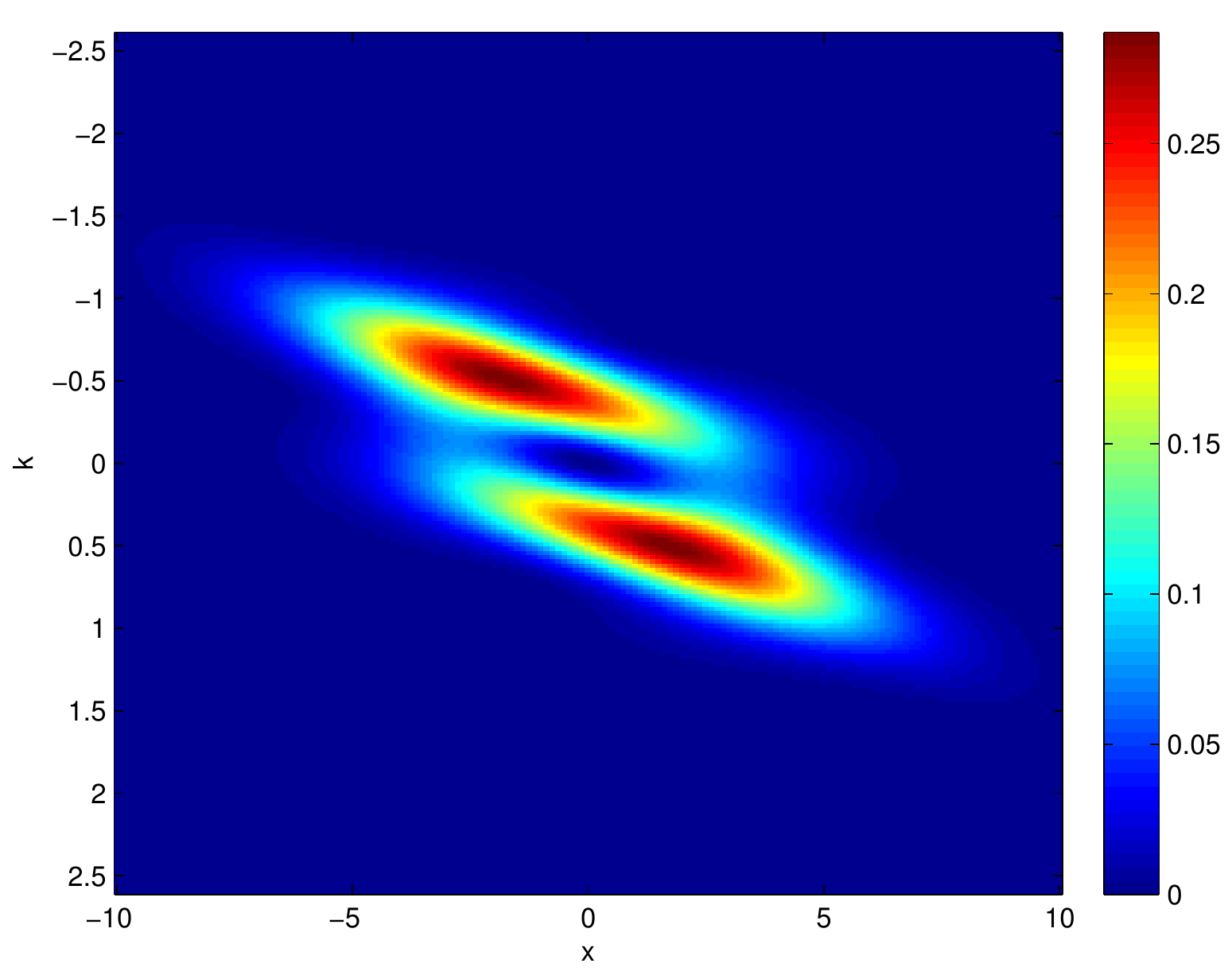}}
     \caption{\small Electron-electron scattering:
The reduced Wigner functions at different instants.
}\label{fig_8}
\end{figure}

\subsection{A Helium-like system}
\label{sec:result:He}

\begin{figure}
    \centering
    \subfigure[$t=1$.]{
    \includegraphics[width=1.9in,height=1.4in]{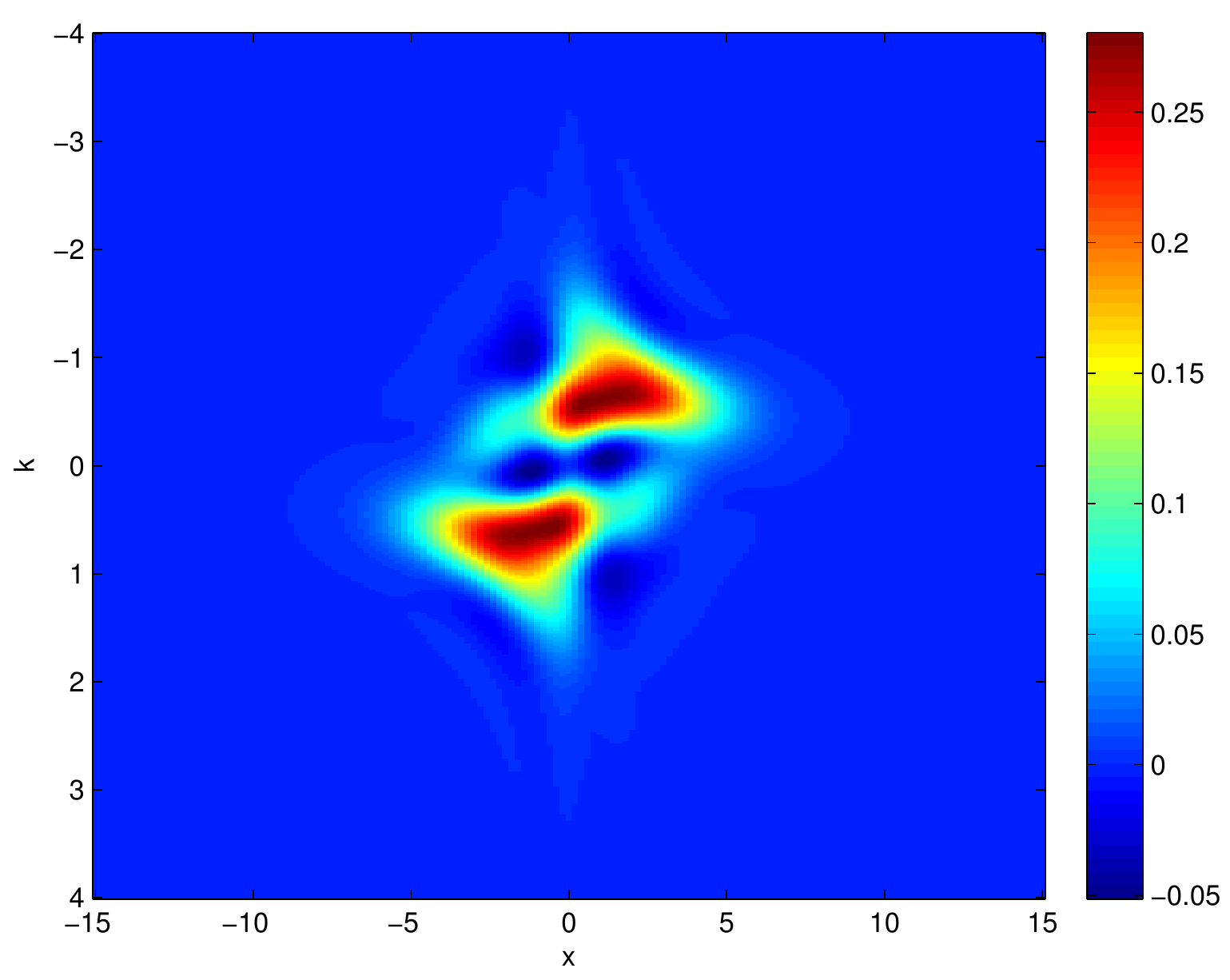}}
    \subfigure[$t=2$.]{
    \includegraphics[width=1.9in,height=1.4in]{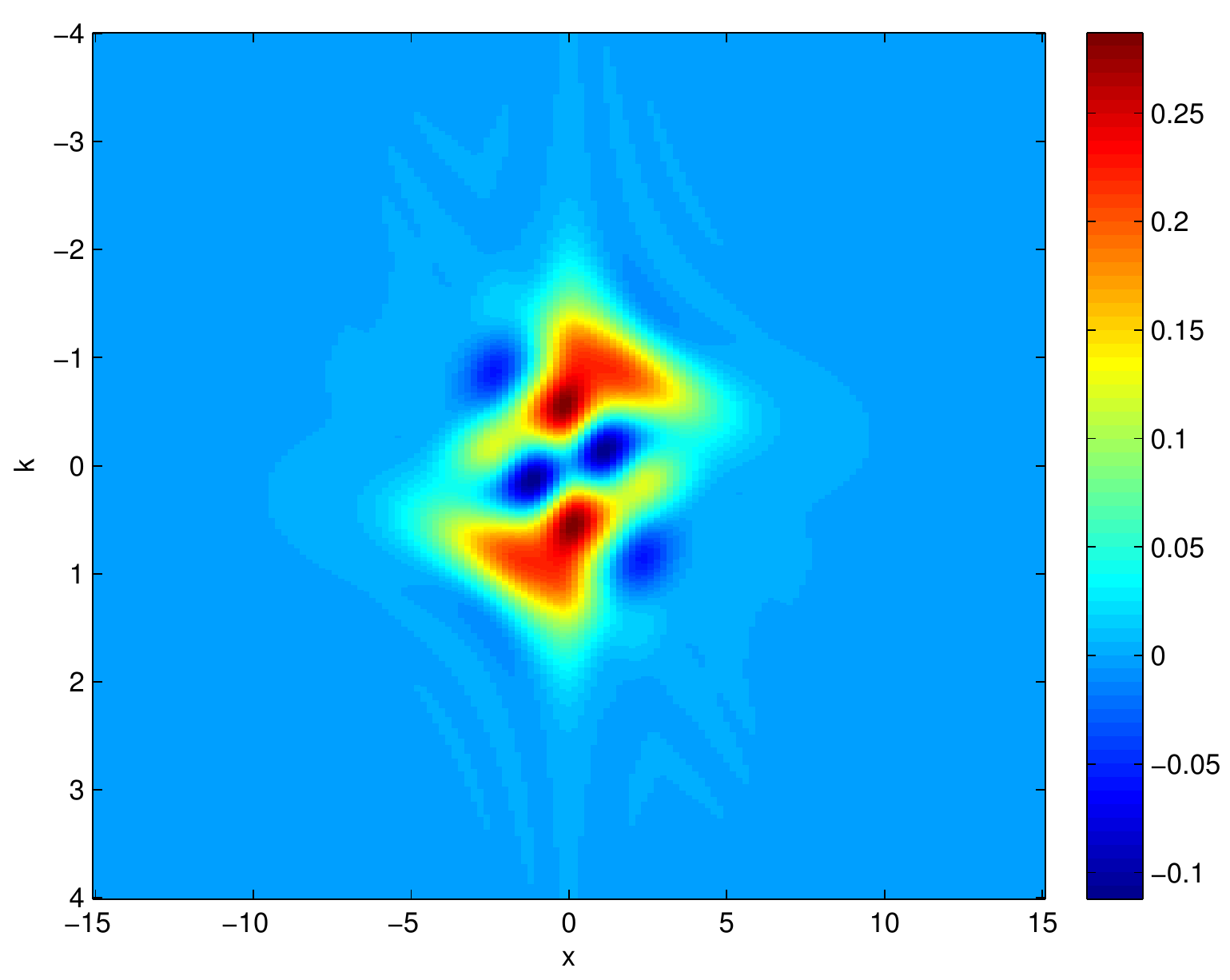}}
    \subfigure[$t=3$.]{
    \includegraphics[width=1.9in,height=1.4in]{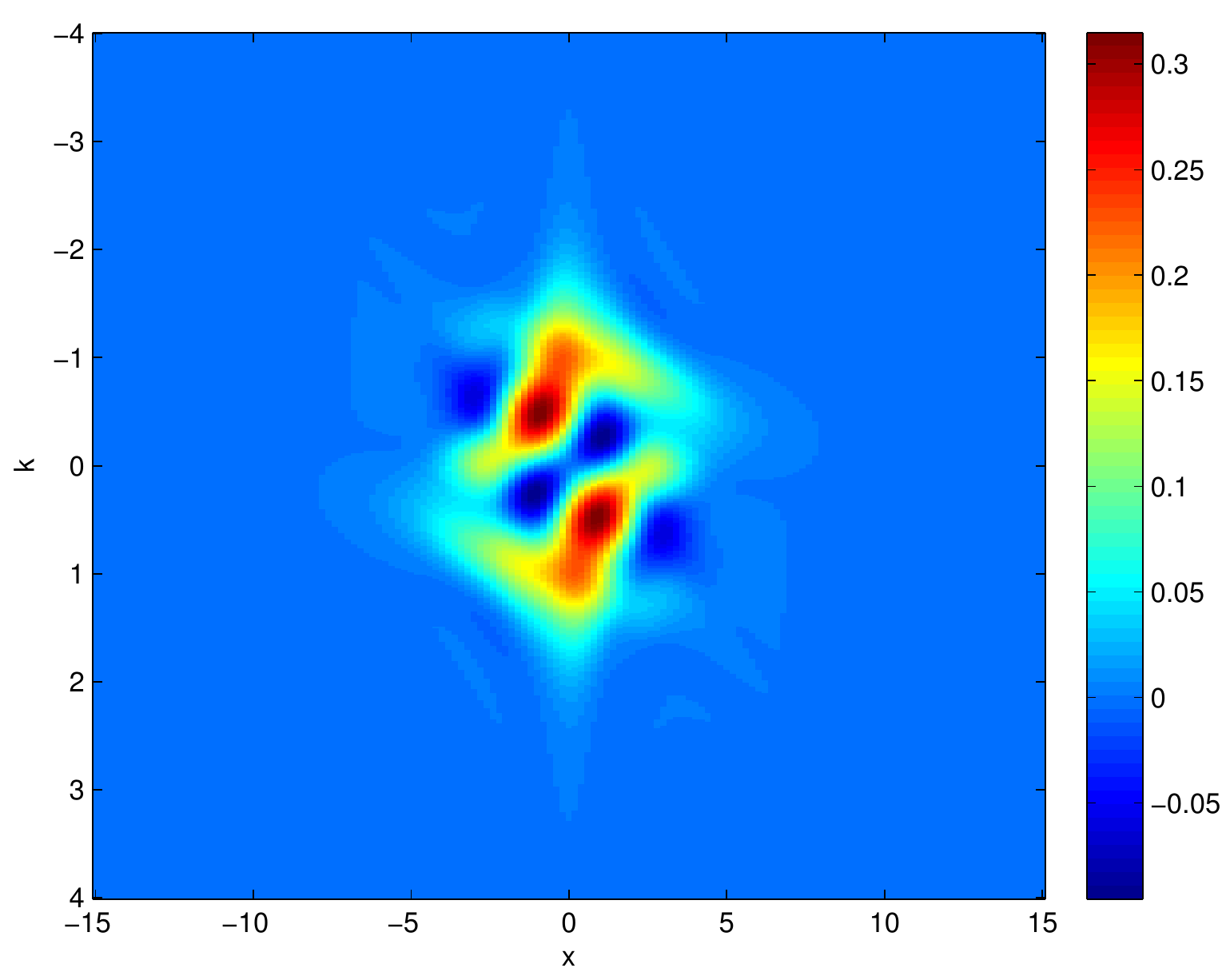}}
    \\
    \centering
    \subfigure[$t=4$.]{
    \includegraphics[width=1.9in,height=1.4in]{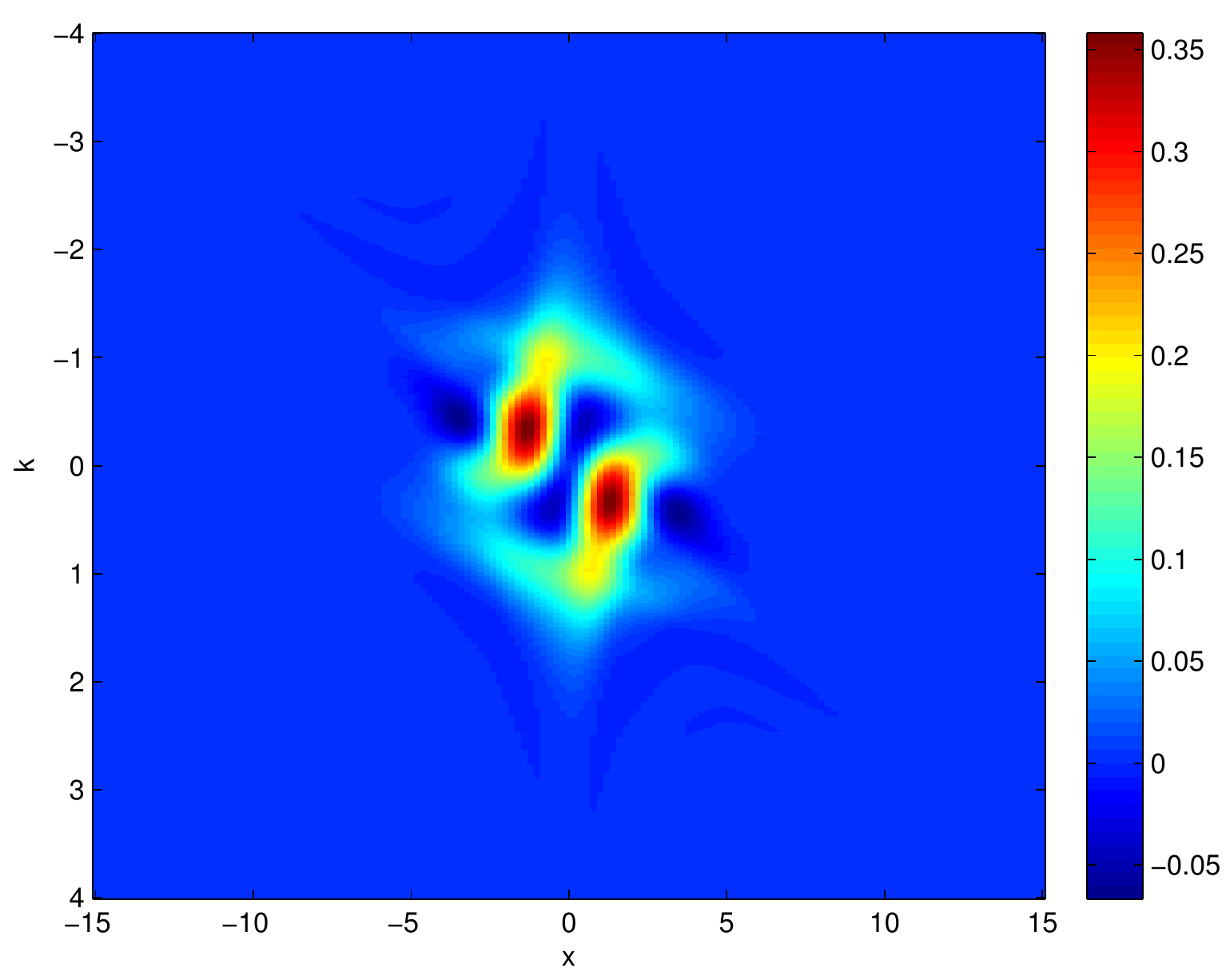}}
    \subfigure[$t=5$.]{
    \includegraphics[width=1.9in,height=1.4in]{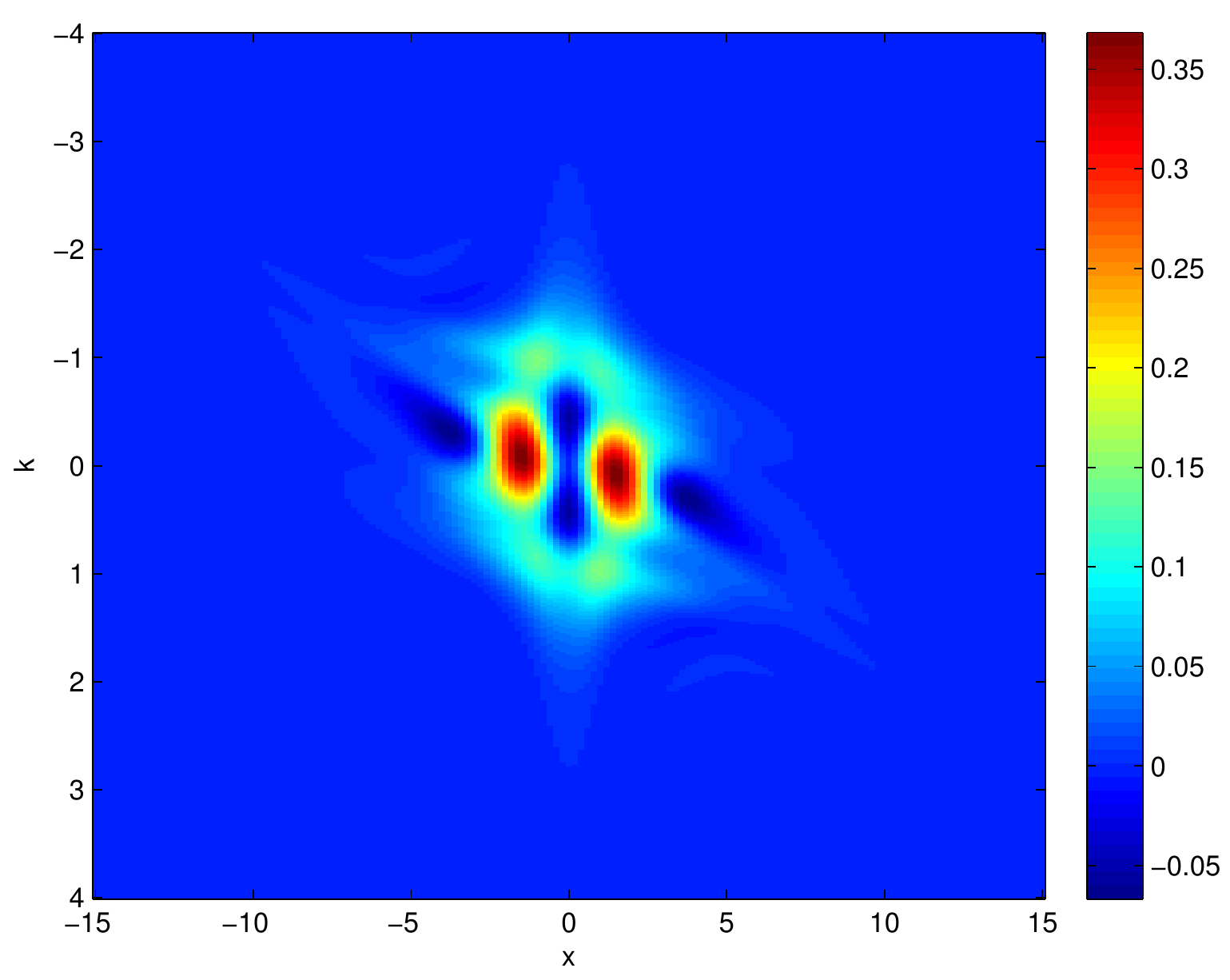}}
    \subfigure[$t=6$.]{
    \includegraphics[width=1.9in,height=1.4in]{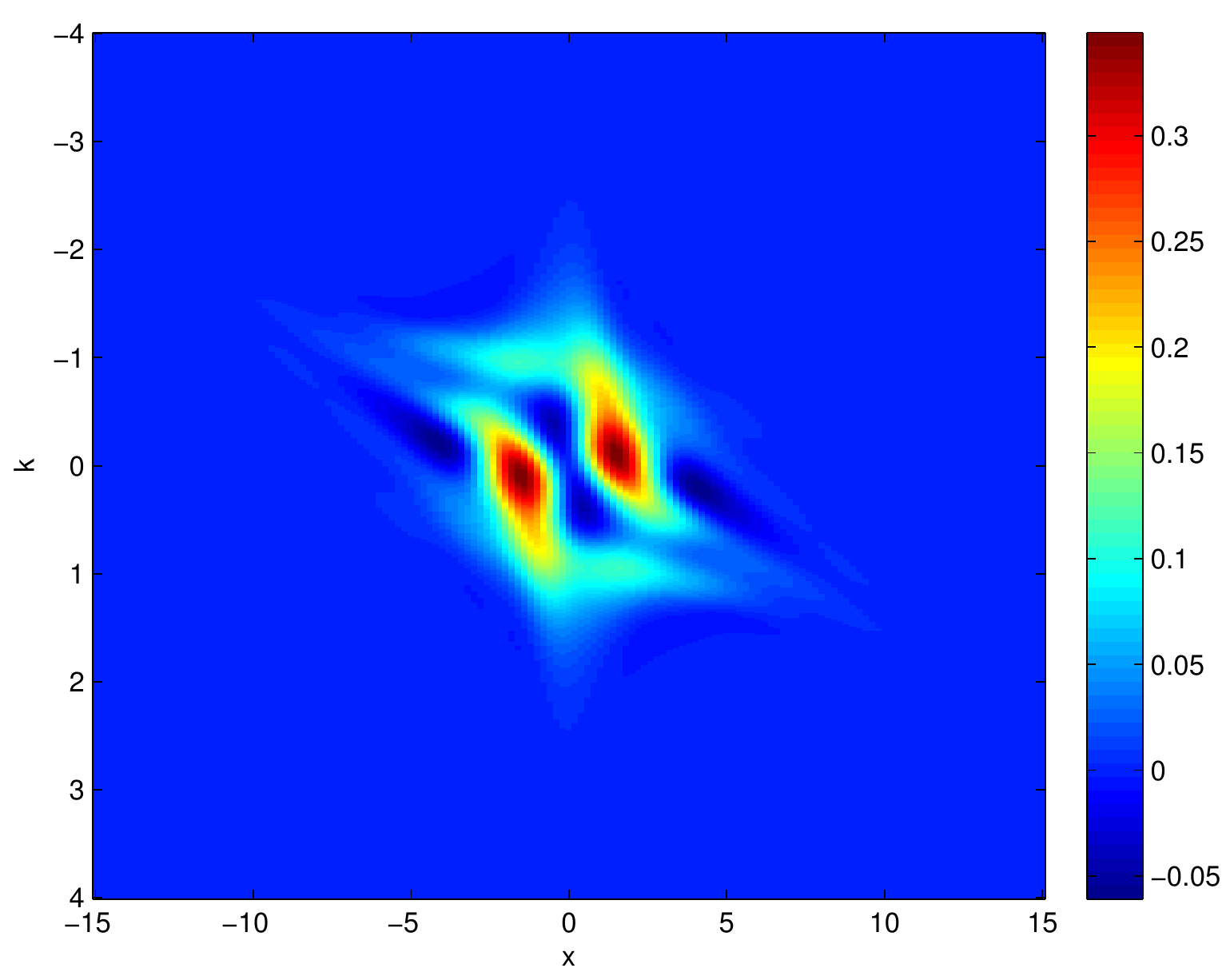}}
    \\
    \centering
    \subfigure[$t=7$.]{
    \includegraphics[width=1.9in,height=1.4in]{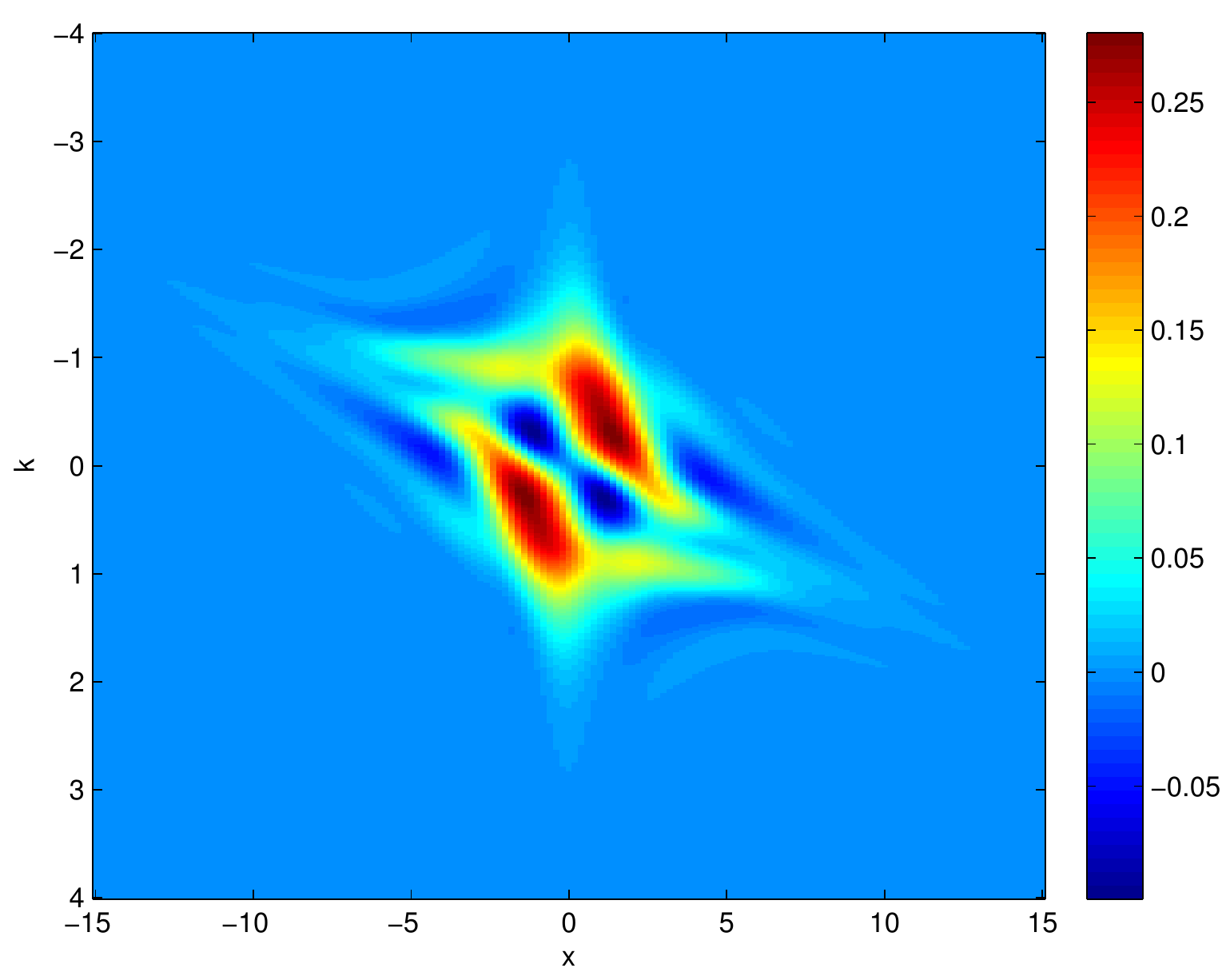}}
    \subfigure[$t=8$.]{
    \includegraphics[width=1.9in,height=1.4in]{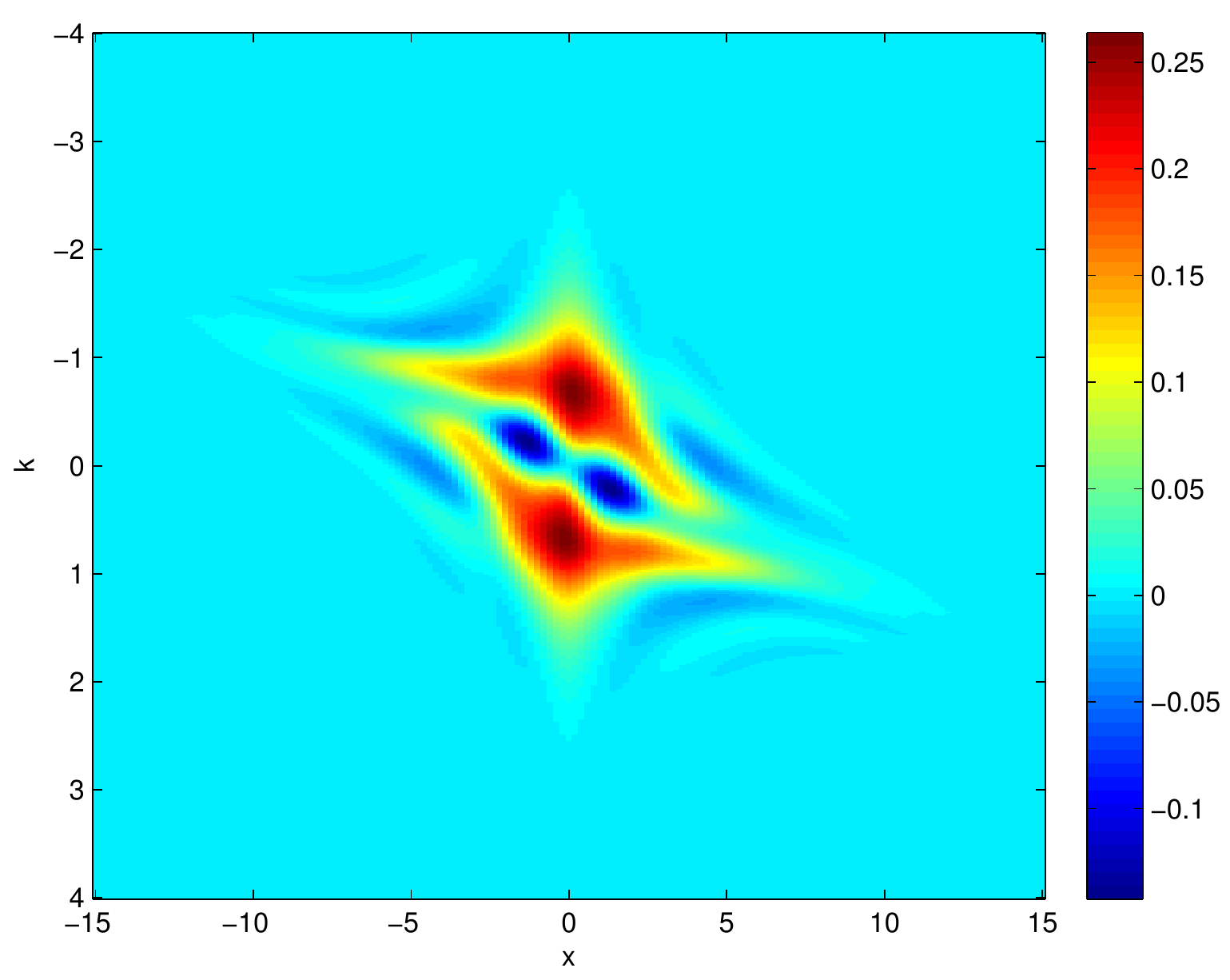}}
    \subfigure[$t=9$.]{
    \includegraphics[width=1.9in,height=1.4in]{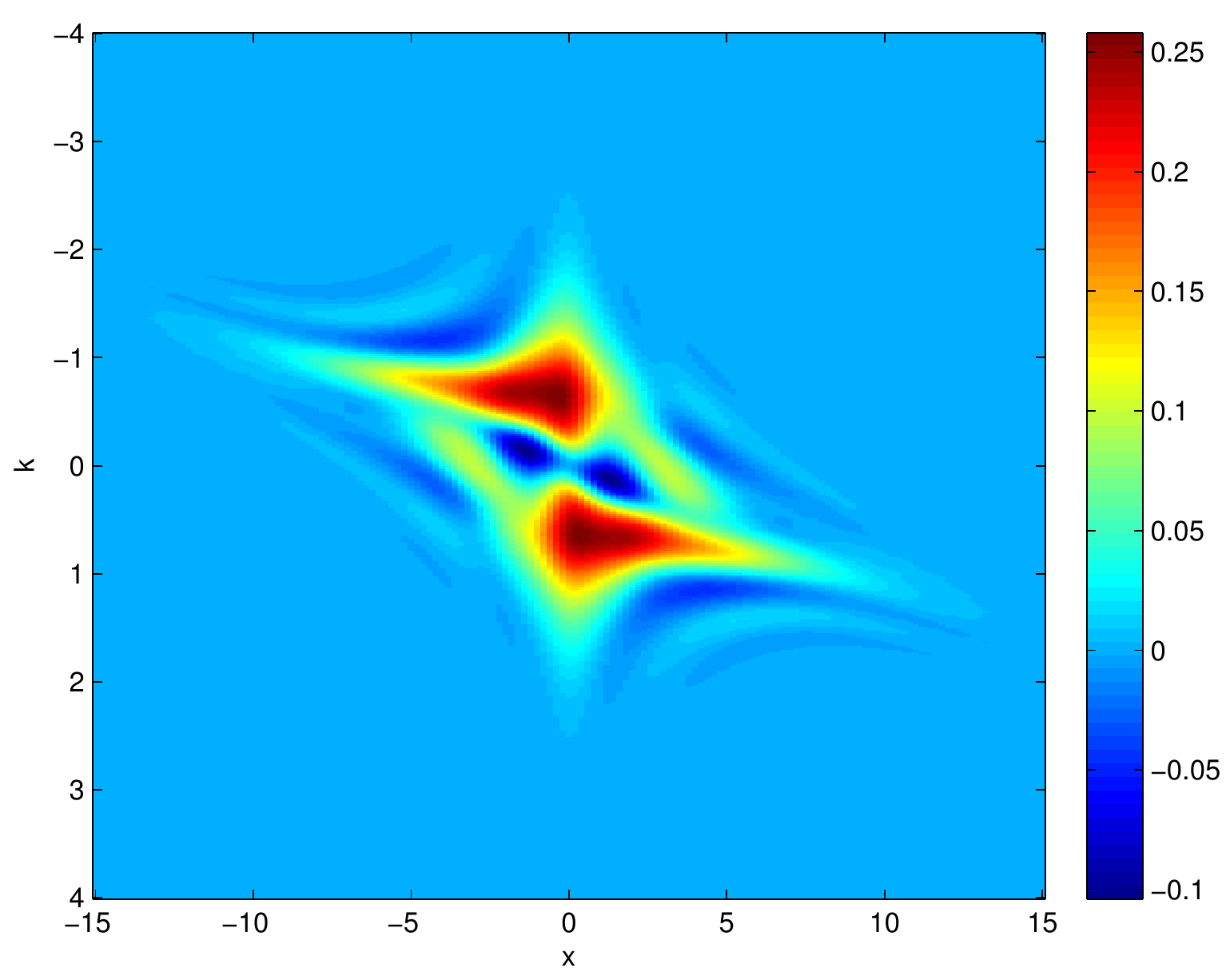}}
    \\
    \centering
    \subfigure[$t=10$.]{
    \includegraphics[width=1.9in,height=1.4in]{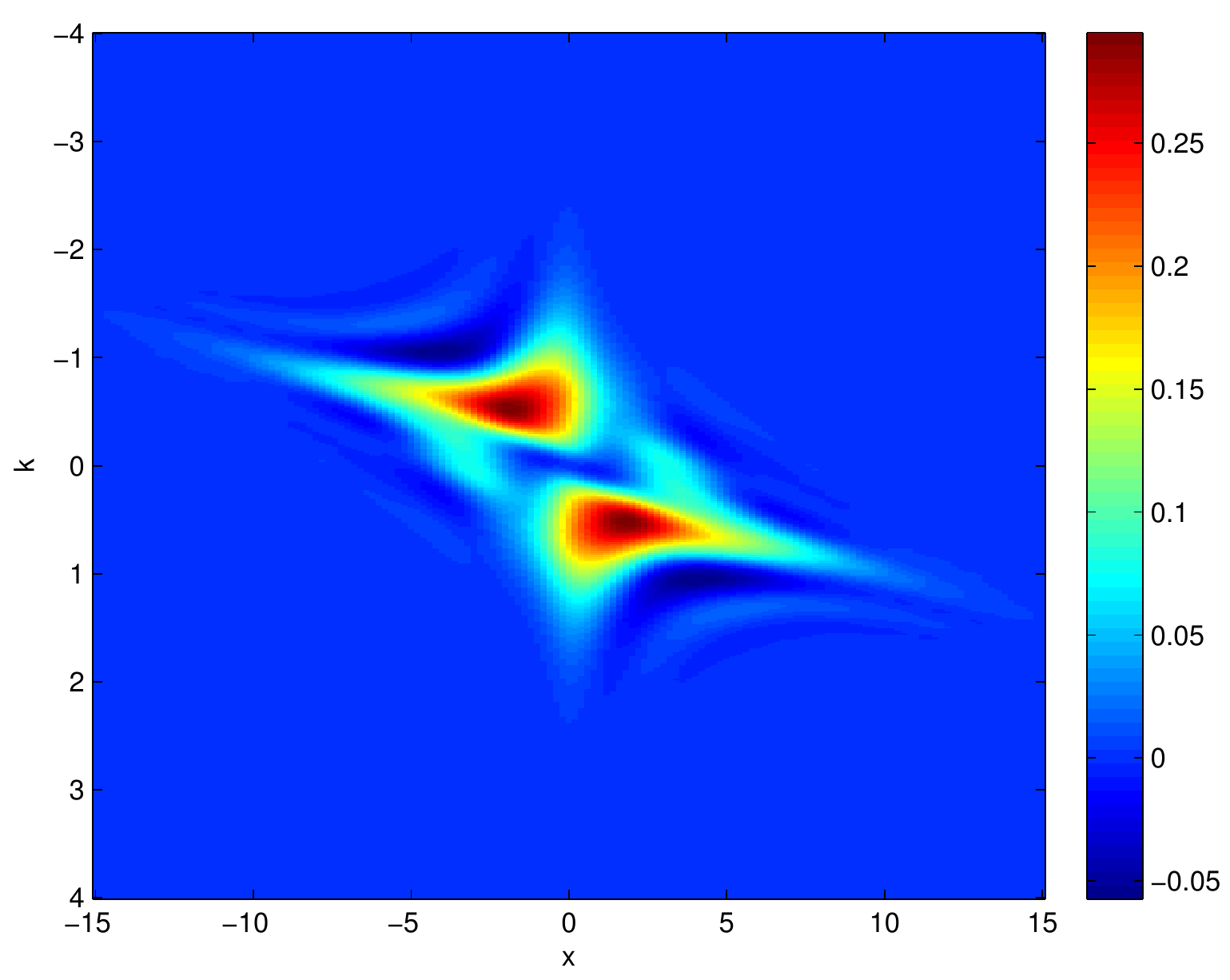}}
    \subfigure[$t=11$.]{
    \includegraphics[width=1.9in,height=1.4in]{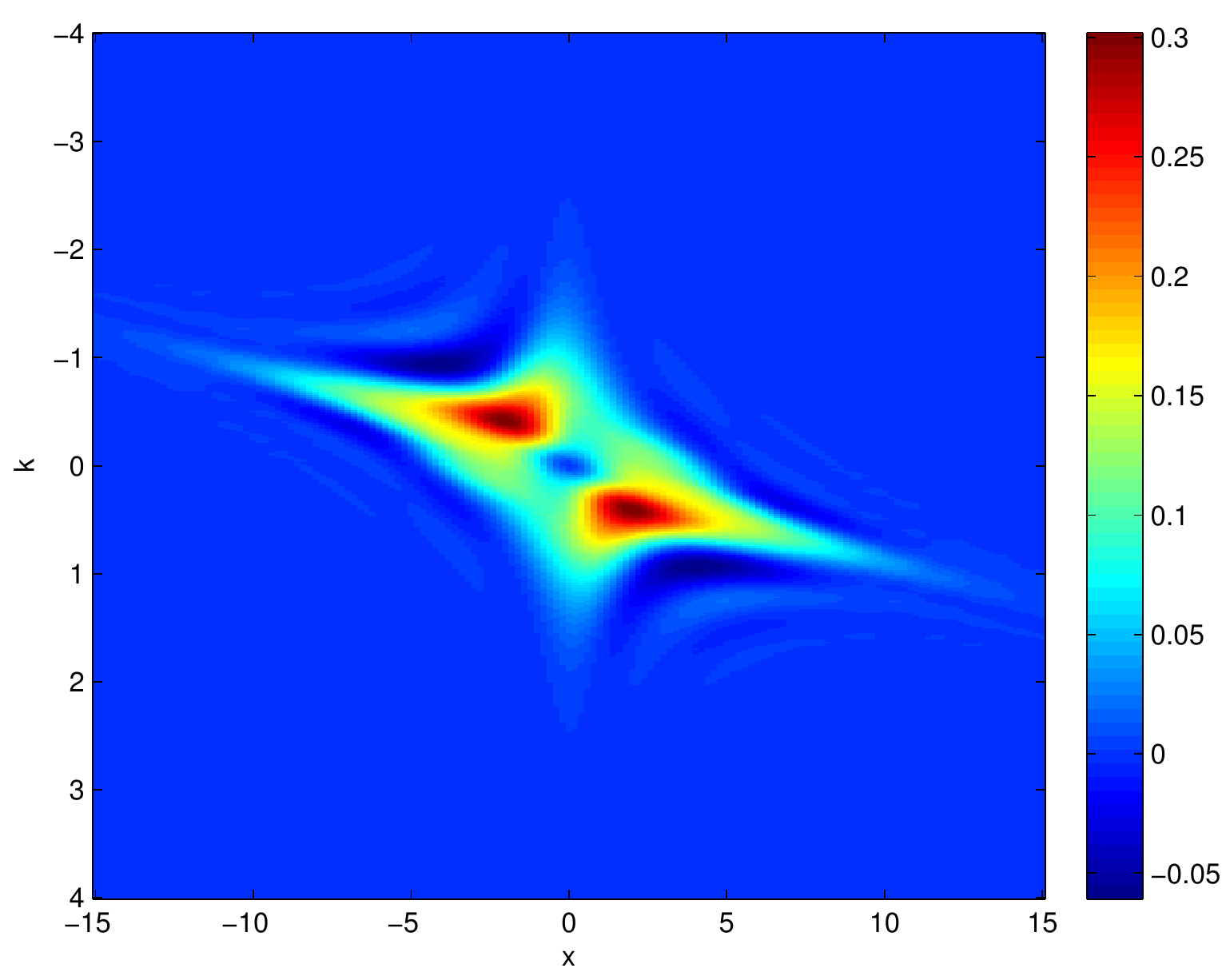}}
    \subfigure[$t=12$.]{
    \includegraphics[width=1.9in,height=1.4in]{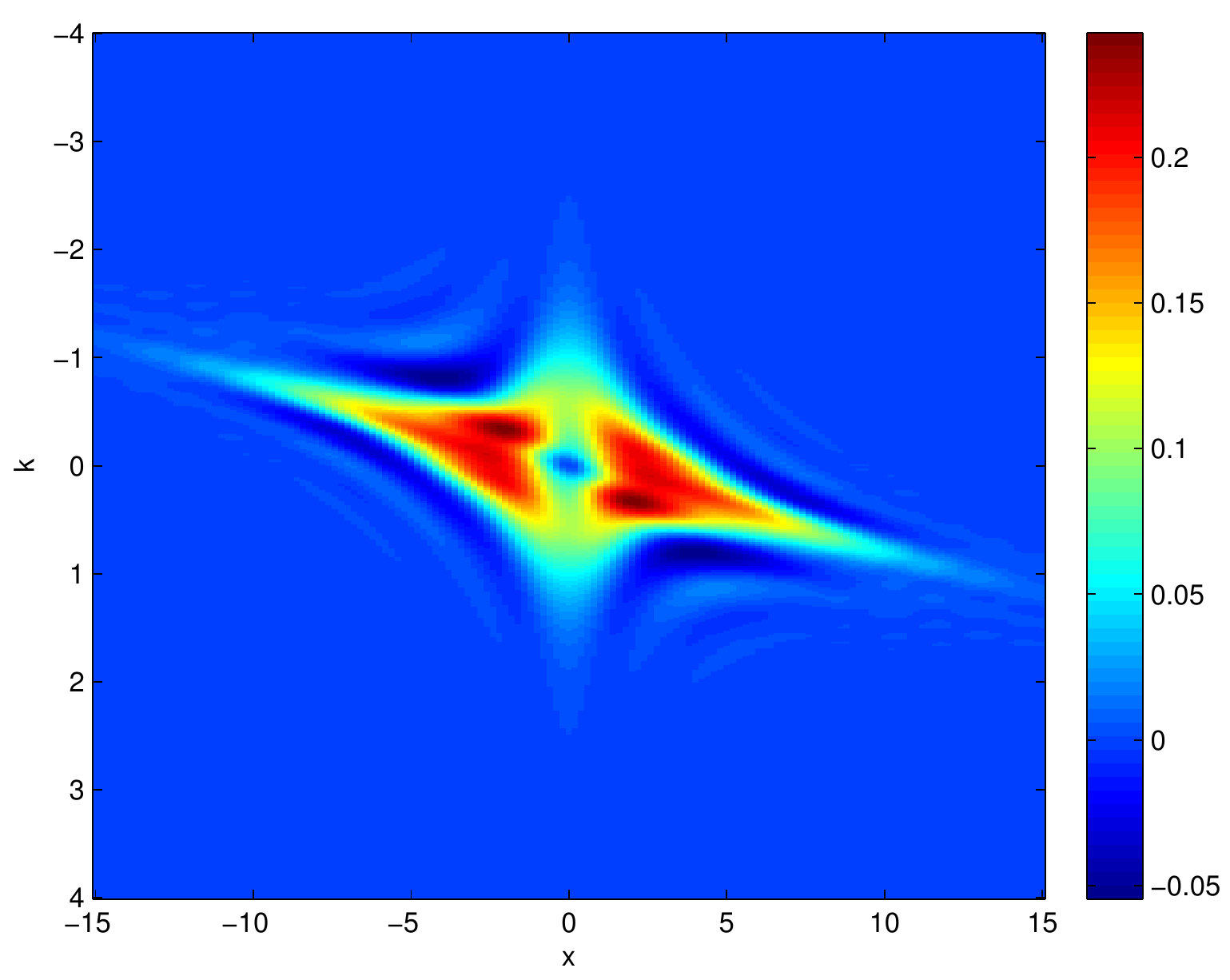}}
    \\
    \centering
    \subfigure[$t=13$.]{
    \includegraphics[width=1.9in,height=1.4in]{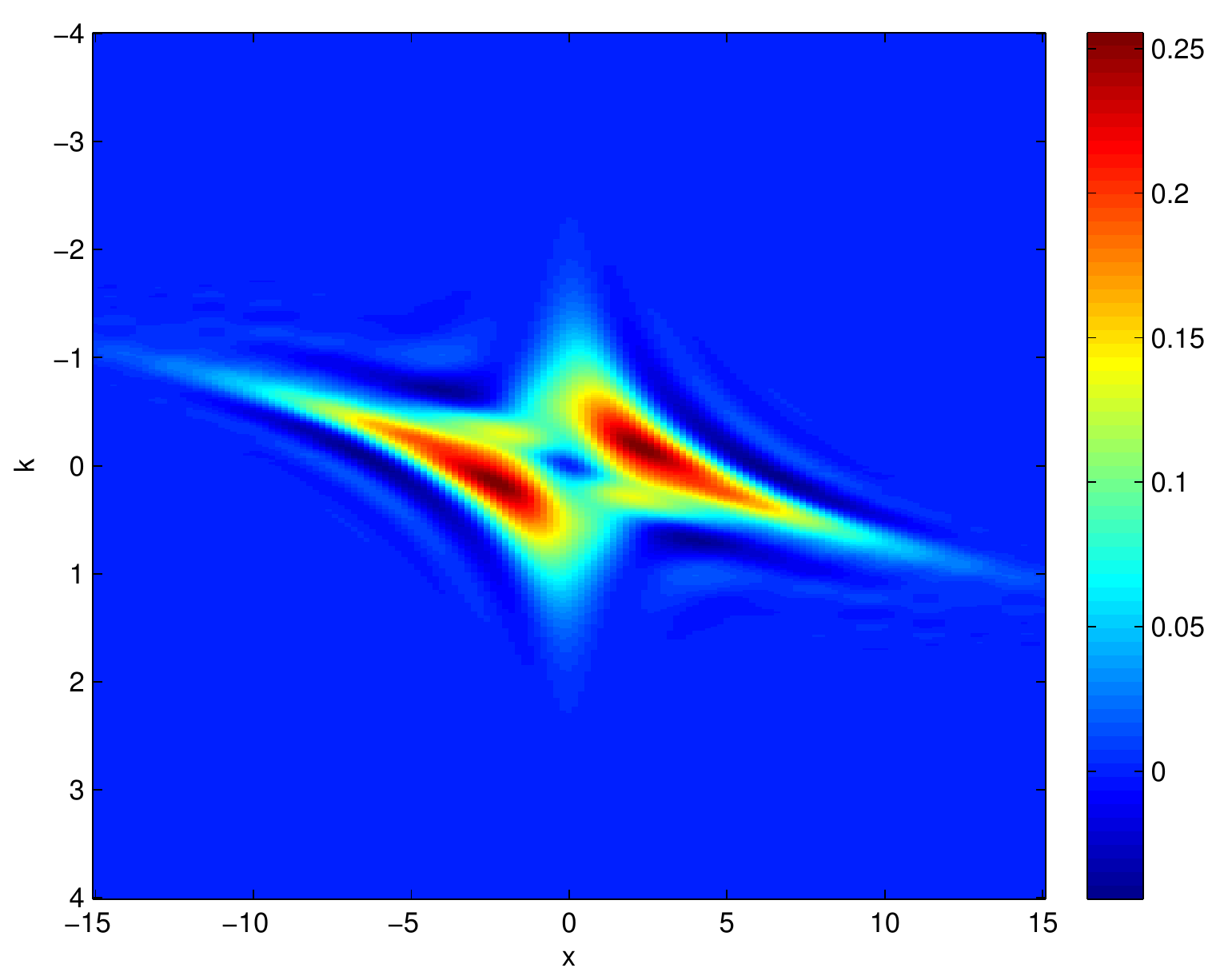}}
    \subfigure[$t=14$.]{
    \includegraphics[width=1.9in,height=1.4in]{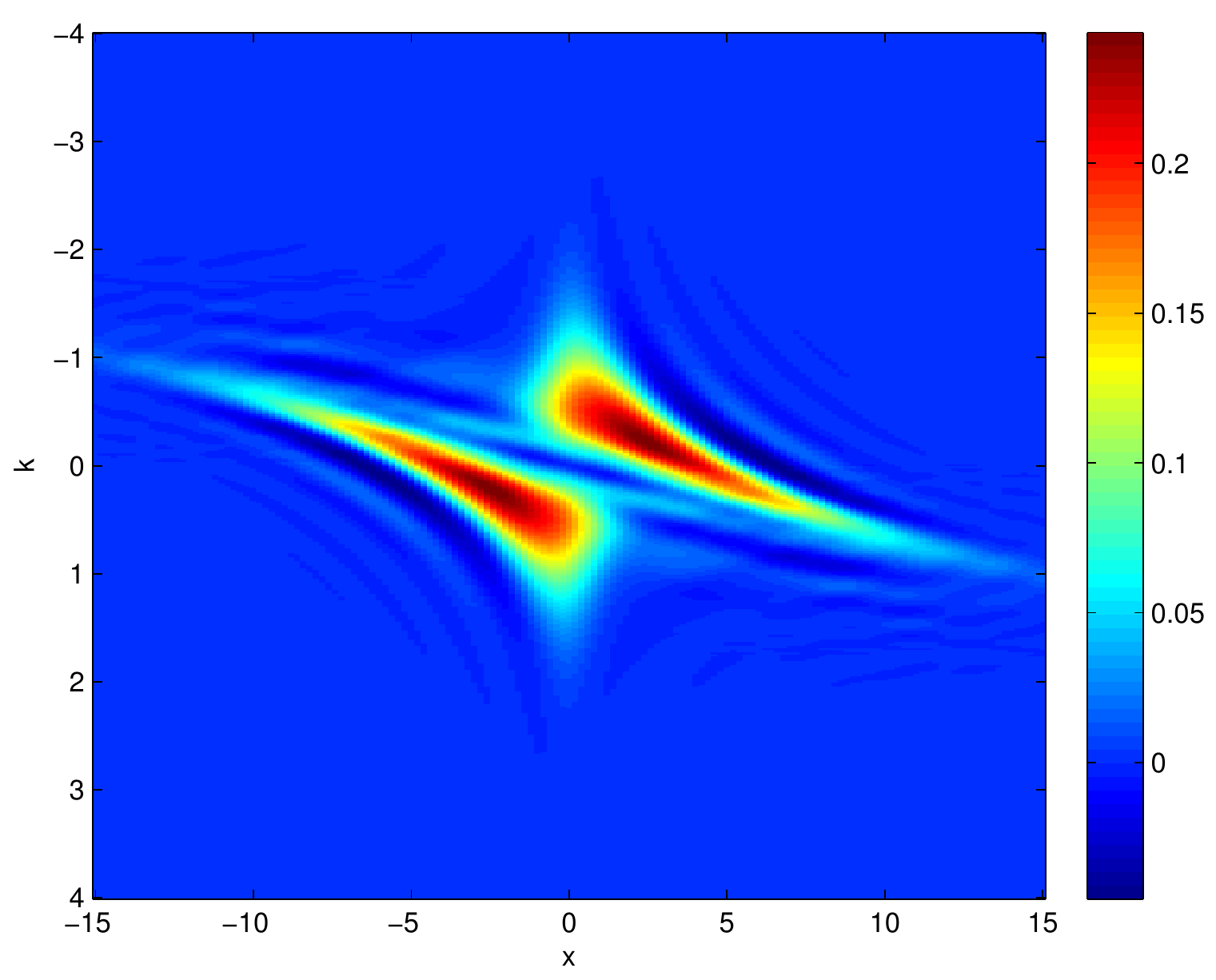}}
     \subfigure[$t=15$.]{
    \includegraphics[width=1.9in,height=1.4in]{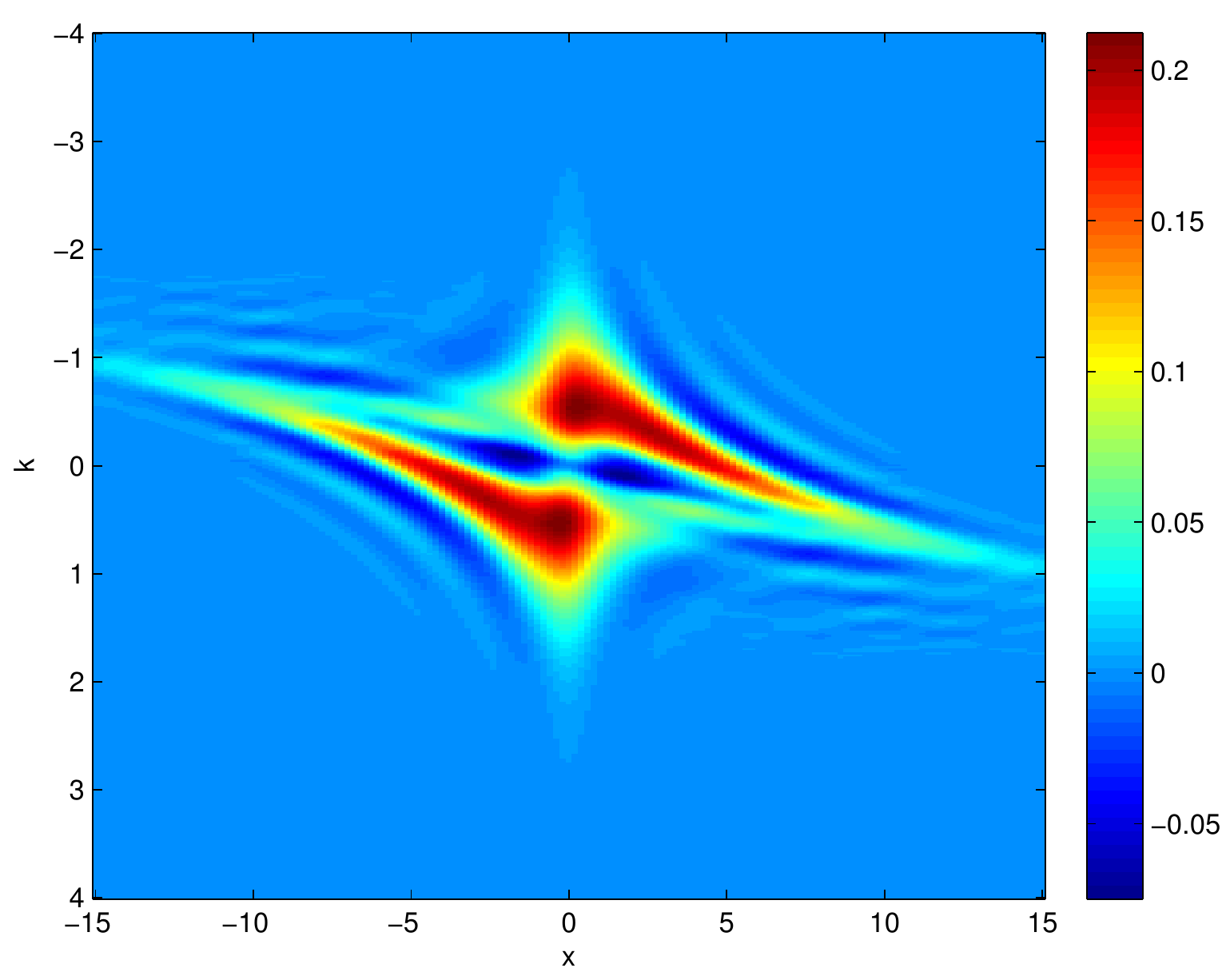}}
     \caption{\small The Helium-like system: The reduced Wigner functions at different instants.
}\label{fig_10}
\end{figure}

As the final example, we consider a Helium-like system composed of two electrons. Besides the repulsive Coulomb force \eqref{electron_potential} with $\epsilon_{\textup{ee}}=1$, they are both attracted by a Helium atom at a fixed position $\bm{x}=\bm{0}$.
To describe the nucleon-electron interaction,
we still adopt the attractive soft-Coulomb potential\cite{LeinKreibichGross2002}
\begin{equation}
V_{\textup{ne}}\left(x_{1}, x_{2}\right)=-\frac{Z}{\sqrt{\left|x_{1}\right|^{2}+\epsilon_{\textup{ne}}}}-\frac{Z}{\sqrt{\left|x_{2}\right|^{2}+\epsilon_{\textup{ne}}}},
\end{equation}
with the atomic number $Z=2$ (for the Helium atom) and the soft parameter $\epsilon_{\textup{ne}}=1$ to remove the singularity at $\bm{x}=0$.
To strike a balance between the accuracy and the efficiency,
we set: $L_x=15$, $L_k = {5\pi}/{3}$, $L_y=60$, $\Delta x =0.2$, $\Delta t=0.05$ and $T=15$.
The $\bm{k}$-domain is divided into $6\times 6$ elements and each element contains $20 \times 20$ Gauss-Chebyshev collocation points.
The same initial data shown in Fig.~\ref{fig:fermion} are adopted.

To demonstrate the dynamics of electrons more clearly,
we take snapshots of the reduced Wigner function $F\left(x, k, t\right)$ defined in Eq.~\eqref{reduced_Wigner} from $t=1$ to $t=15$, as shown in Fig.~\ref{fig_10}. At the early stage before $t=5$,
the reduced Wigner function is forced to be localized in the central area due to the nucleon-electron interaction.
Afterwards, both dispersion and correlation show up clearly and lead to a highly oscillating structure in the phase space,
in which each peak of positive value is followed by a valley of negative value. Furthermore, we find that the Wigner function rotates around the Helium atom periodically with the approximate period of $6$. This periodic behavior may reflect a kind of simple harmonic vibration of the one-dimensional electrons.
We could always observe there a concentration of the negative Wigner function in the central area, accounting for the electron-electron interaction, because the negative distribution is related to the regions that are experimentally forbidden by the uncertainty principle. We should point out that the Fermi hole structure does be always there though it shows two branches, for example, at $t=7,15$, in contrast to the numerical results by
the signed particle MCM\cite{SellierDimov2015}. This may cast doubts on the accuracy of the signed particle MCM, as its numerical resolution might be too poor to catch the quantum interference and coherence precisely.
Although a recent study showed the accuracy of the signed particle MCM for one-dimension one-body situation\cite{ShaoSellier2015},
a more comprehensive study to validate the accuracy of the many-body Wigner MCM is highly desired in this regard.

\section{Conclusion and outlook}
\label{sec:conclusion}

An efficient and accurate deterministic method is proposed in this work for a direct simulation of the many-body Wigner equation. It resolves the Lagrangian advection on the spatial space by an explicit multistep characteristic method, and the shifted grid points are interpolated through piecewise cubic splines. The nonlocal Wigner interaction term is tackled by a highly accurate Chebyshev spectral element method, and
the resulted advective-spectral-mixed method relaxes the usual CFL restriction on the time step and achieves the third-order convergence. Moreover, it is able to maintain the mass conservation and the physical symmetry relation for identical particle systems. Several typical numerical experiments for one-body and two-body quantum systems in one-dimensional spatial space show
the appearance of both Pauli exclusion principle and uncertainty principle in the phase space.
The proposed method can be straightforwardly employed in the high dimensional one-body problem, thereby making it possible to perform time-dependent Wigner simulations in the two or three dimensional semiconductor device. In principle, it can also resolve the nonlinear Wigner quantum models, such as the Wigner-Poisson system. We would like to discuss this topic as well as a more appropriate formulation of quantum boundary conditions in subsequent papers.

\section*{Acknowledgement}
This research was supported by grants from the National Natural Science Foundation of China (Nos.~11471025, 91330110, 11421101).


\begin{thebibliography}{10}

\bibitem{Wigner1932}
E.~Wigner.
\newblock On the quantum corrections for thermodynamic equilibrium.
\newblock {\em Phys. Rev.}, 40:749--759, 1932.

\bibitem{tatarskiui1983}
V.~I. Tatarski{\u\i}.
\newblock The {Wigner} representation of quantum mechanics.
\newblock {\em Sov. Phys. Usp}, 26:311--327, 1983.

\bibitem{JacoboniBordone2004}
C.~Jacoboni and P.~Bordone.
\newblock The {W}igner-function approach to non-equilibrium electron transport.
\newblock {\em Rep. Prog. Phys.}, 67:1033--1071, 2004.

\bibitem{DiasPrata2004}
N.~C. Dias and J.~N. Prata.
\newblock Admissible states in quantum phase space.
\newblock {\em Ann. Phys.}, 313:110--146, 2004.

\bibitem{bk:MarkowichRinghoferSchmeiser1990}
P.~A. Markowich, C.~A. Ringhofer, and C.~Schmeiser.
\newblock {\em Semiconductor Equations}.
\newblock Springer-Verlag, Wien-New York, 1990.

\bibitem{th:Biegel1997}
B.~A. Biegel.
\newblock {\em Quantum Electronic Device Simulation}.
\newblock PhD thesis, Stanford University, 1997.

\bibitem{bk:Balescu1975}
R.~Balescu.
\newblock {\em Equilibrium and Nonequilibrium Statistical Mechanics}.
\newblock John Wiley \& Sons, New York, 1975.

\bibitem{bk:Schleich2011}
W.~P. Schleich.
\newblock {\em Quantum Optics in Phase Space}.
\newblock Wiley-VCH, Berlin, 2011.

\bibitem{bk:Leonhardt1997}
U.~Leonhardt.
\newblock {\em Measuring the Quantum State of Light}.
\newblock Cambridge University Press, New York, 1997.

\bibitem{LeibfriedPfauMonroe1998}
D.~Leibfried, T.~Pfau, and C.~Monroe.
\newblock {Shadows and mirrors: Reconstructing quantum states of atom motion}.
\newblock {\em Phys. Today}, April:22--28, 1998.

\bibitem{Zurek1991}
W.~H. Zurek.
\newblock {Decoherence and the transition from quantum to classical}.
\newblock {\em Phys. Today}, October:36--44, 1991.

\bibitem{Zachos2002}
C.~Zachos.
\newblock Deformation quantization: quantum mechanics lives and works in
  phase-space.
\newblock {\em Int. J. Mod. Phys. A}, 17:297--316, 2002.

\bibitem{Frensley1987}
W.~R. Frensley.
\newblock Wigner-function model of a resonant-tunneling semiconductor device.
\newblock {\em Phys. Rev. B}, 36:1570--1580, 1987.

\bibitem{Frensley1990}
W.~R. Frensley.
\newblock Boundary conditions for open quantum systems driven far from
  equilibrium.
\newblock {\em Rev. Mod. Phys.}, 62:745--791, 1990.

\bibitem{JensenBuot1991}
K.~L. Jensen and F.~A. Buot.
\newblock The methodology of simulating particle trajectories through tunneling
  structures using a {W}igner distribution approach.
\newblock {\em IEEE Trans. Electron Devices}, 38:2337--2347, 1991.

\bibitem{Ringhofer1990}
C.~Ringhofer.
\newblock A spectral method for the numerical simulation of quantum tunneling
  phenomena.
\newblock {\em SIAM J. Numer. Anal.}, 27:32--50, 1990.

\bibitem{SuhFeixBertrand1991}
N.-D. Suh, M.~R. Feix, and P.~Bertrand.
\newblock Numerical simulation of the quantum {Liouville-Poisson} system.
\newblock {\em J. Comput. Phys.}, 94:403--418, 1991.

\bibitem{ArnoldRinghofer1996}
A.~Arnold and C.~Ringhofer.
\newblock A operator splitting method for the {W}igner-{P}oisson problem.
\newblock {\em SIAM J. Numer. Anal.}, 33:1622--1643, 1996.

\bibitem{ShaoLuCai2011}
S.~Shao, T.~Lu, and W.~Cai.
\newblock Adaptive conservative cell average spectral element methods for
  transient {Wigner} equation in quantum transport.
\newblock {\em Commun. Comput. Phys.}, 9:711--739, 2011.

\bibitem{LiLuWangYao2014}
R.~Li, T.~Lu, Y.~Wang, and W.~Yao.
\newblock Numerical validation for high order hyperbolic moment system of
  {Wigner} equation.
\newblock {\em Commun. Comput. Phys.}, 15:569--595, 2014.

\bibitem{FurtmaierSucciMendoza2015}
O.~Furtmaier, S.~Succi, and M.~Mendoza.
\newblock Semi-spectral method for the {Wigner} equation.
\newblock {\em J. Comput. Phys.}, Online, 2015.

\bibitem{DordaSchurrer2015}
A.~Dorda and F.~Sch{\"u}rrer.
\newblock A {WENO-solver} combined with adaptive momentum discretization for
  the {Wigner} transport equation and its application to resonant tunneling
  diodes.
\newblock {\em J. Comput. Phys.}, 284:95--116, 2015.

\bibitem{NedjalkovSchwahaSelberherr2013}
M.~Nedjalkov, P.~Schwaha, S.~Selberherr, J.~M. Sellier, and D.~Vasileska.
\newblock {Wigner quasi-particle attributes -- An asymptotic perspective}.
\newblock {\em Appl. Phys. Lett.}, 102:163113, 2013.

\bibitem{NedjalkovKosinaSelberherrRinghoferFerry2004}
M.~Nedjalkov, H.~Kosina, S.~Selberherr, C.~Ringhofer, and D.~K. Ferry.
\newblock Unified particle approach to {W}igner-{B}oltzmann transport in small
  semiconductor devices.
\newblock {\em Phys. Rev. B}, 70:115319, 2004.

\bibitem{SellierNedjalkovDimov2014}
J.~M. Sellier, M.~Nedjalkov, I.~Dimov, and S.~Selberherr.
\newblock A benchmark study of the {Wigner Monte-Carlo} method.
\newblock {\em Monte Carlo Methods Appl.}, 20:43--51, 2014.

\bibitem{SellierDimov2015}
J.~M. Sellier and I.~Dimov.
\newblock On the simulation of indistinguishable fermions in the many-body
  {Wigner} formalism.
\newblock {\em J. Comput. Phys.}, 280:287--294, 2015.

\bibitem{ShaoSellier2015}
S.~Shao and J.~M. Sellier.
\newblock Comparison of deterministic and stochastic methods for time-dependent
  {Wigner} simulations.
\newblock {\em J. Comput. Phys.}, 300:167--185, 2015.

\bibitem{CervenkaEllinghausNedjalkov2015}
J.~Cervenka, P.~Ellinghaus, and M.~Nedjalkov.
\newblock Deterministic solution of the discrete {Wigner} equation.
\newblock In I.~Dimov, S.~Fidanova, and I.~Lirkov, editors, {\em Numerical
  Methods and Applications}, pages 149--156, 2015.

\bibitem{CancellieriBordoneJacoboni2007}
E.~Cancellieri, P.~Bordone, and C.~Jacoboni.
\newblock Effect of symmetry in the many-particle {Wigner} function.
\newblock {\em Phys. Rev. B}, 76:214301, 2007.

\bibitem{bk:HairerNorsettWanner1993}
E.~Hairer, S.~P. N{\o}rsett, and G.~Wanner.
\newblock {\em Solving Ordinary Differential Equations I: Nonstiff Problems}.
\newblock Springer-Verlag, Berlin, 2nd edition, 2009.

\bibitem{bk:Boor2001}
C.~{de Boor}.
\newblock {\em {A Practical Guide to Splines}}.
\newblock Springer-Verlag, New York, revised edition, 2001.

\bibitem{SonnendruckerRocheBertrand1999}
E.~Sonnendr{\"u}cker, J.~Roche, P.~Bertrand, and A.~Ghizzo.
\newblock The {semi-Lagrangian} method for the numerical resolution of the
  {Vlasov} equation.
\newblock {\em J. Comput. Phys.}, 149:201--220, 1999.

\bibitem{CancellieriBordoneBertoni2004}
E.~Cancellieri, P.~Bordone, A.~Bertoni, G.~Ferrari, and C.~Jacoboni.
\newblock Wigner function for identical particles.
\newblock {\em J. Comput. Electron.}, 3:411--415, 2004.

\bibitem{bk:Pazy1983}
A.~Pazy.
\newblock {\em Semigroups of Linear Operators and Applications to Partial
  Differential Equations}.
\newblock Springer-Verlag, New York, 1983.

\bibitem{HugMenkeSchleich1998I}
M.~Hug, C.~Menke, and W.~P. Schleich.
\newblock Modified spectral method in phase space: Calculation of the {W}igner
  function. {I}. {F}undamentals.
\newblock {\em Phys. Rev. A}, 57:3188--3205, 1998.

\bibitem{CrouseillesMehrenbergerSonnendrucker2010}
N.~Crouseilles, M.~Mehrenberger, and E.~Sonnendr{\"u}cker.
\newblock Conservative {semi-Lagrangian} schemes for {Vlasov} equations.
\newblock {\em J. Comput. Phys.}, 229:1927--1953, 2010.

\bibitem{bk:ShenTangWang2011}
J.~Shen, T.~Tang, and L.-L. Wang.
\newblock {\em Spectral Methods: Algorithms, Analysis and Applications}.
\newblock Springer-Verlag, Berlin, 2011.

\bibitem{JiangCaiTsu2011}
H.~Jiang, W.~Cai, and R.~Tsu.
\newblock Accuracy of the {Frensley} inflow boundary condition for {Wigner}
  equations in simulating resonant tunneling diodes.
\newblock {\em J. Comput. Phys.}, 230:2031--2044, 2011.

\bibitem{Swarztrauber1982}
P.~N. Swarztrauber.
\newblock Vectorizing the {FFTs}.
\newblock In G.~Rodrigue, editor, {\em Parallel Computations}, pages 51--83.
  Academic Press, 1982.

\bibitem{bk:PressTeukolskyVetterlingFlannery1992}
W.~H. Press, S.~A. Teukolsky, W.~T. Vetterling, and B.~P. Flannery.
\newblock {\em Numerical Recipes in {FORTRAN}: The Art of Scientific
  Computing}.
\newblock Cambridge University Press, Cambridge, second edition, 1992.

\bibitem{web:McCune2010}
D.~McCune.
\newblock {PSPLINE} -- a library of spline and {Hermite} cubic interpolation
  routines for 1d, 2d, and 3d datasets on rectilinear grids.
\newblock {\tt http://w3.pppl.gov/ntcc/PSPLINE/}.

\bibitem{LeinKreibichGross2002}
M.~Lein, T.~Kreibich, E.~K.~U. Gross, and V.~Engel.
\newblock {Strong-field ionization dynamics of a model H$_2$ molecule}.
\newblock {\em Phys. Rev. A}, 65:033403, 2002.

\end{thebibliography}

\end{document}